\documentclass[a4paper,11pt]{scrartcl}

\usepackage{graphicx}%
\usepackage{dcolumn}%
\usepackage{bm}%
\usepackage{mathrsfs}
\usepackage{amsmath,amsthm,amssymb,amsfonts}
\usepackage{dsfont}
\usepackage{xcolor}
\usepackage{tikz}
\usetikzlibrary{shapes.geometric} %
\usepackage{mathabx}%
\usepackage{braket}
\usepackage{comment}
\usepackage{hyperref}
\usepackage{enumitem}
\setlist[enumerate]{itemsep=0mm} %
\usetikzlibrary{quantikz2} %
\usepackage{amsthm}
\usepackage{subcaption}

\newtheorem{theorem}{Theorem}
\newtheorem*{theorem*}{Theorem}
\newtheorem{lemma}{Lemma}
\newtheorem{definition}{Definition}
\newtheorem{corollary}{Corollary}
\newtheorem{remark}{Remark}

\newcommand{\gen}{\mathsf{Gen}}
\newcommand{\eval}{\mathsf{Eval}}
\newcommand{\enc}{\mathsf{Enc}}
\newcommand{\dec}{\mathsf{Dec}}
\newcommand{\sk}{\mathsf{sk}}

\newcommand{\tr}{\mathsf{tr}}
\newcommand{\la}{\lambda}
\newcommand{\encx}{\mathtt{x}}
\newcommand{\enca}{\mathtt{a}}
\newcommand{\ency}{\mathtt{y}}
\newcommand{\encb}{\mathtt{b}}

\newcommand{\cs}{{C}$^*$}
\newcommand{\Tr}{\mathrm{Tr}}
\newcommand{\cG}{\mathcal{G}}
\newcommand{\mk}{\mathsf{k}}
\newcommand{\negl}{\mathsf{negl}}

\usepackage{caption}
\usepackage{authblk}
\renewcommand\Affilfont{\small} 

\makeatletter
\newcommand\email[2][]%
   {\newaffiltrue\let\AB@blk@and\AB@pand
      \if\relax#1\relax\def\AB@note{\AB@thenote}\else\def\AB@note{\relax}%
        \setcounter{Maxaffil}{0}\fi
      \begingroup
        \let\protect\@unexpandable@protect
        \def\thanks{\protect\thanks}\def\footnote{\protect\footnote}%
        \@temptokena=\expandafter{\AB@authors}%
        {\def\\{\protect\\\protect\Affilfont}\xdef\AB@temp{#2}}%
         \xdef\AB@authors{\the\@temptokena\AB@las\AB@au@str
         \protect\\[\affilsep]\protect\Affilfont\AB@temp}%
         \gdef\AB@las{}\gdef\AB@au@str{}%
        {\def\\{, \ignorespaces}\xdef\AB@temp{#2}}%
        \@temptokena=\expandafter{\AB@affillist}%
        \xdef\AB@affillist{\the\@temptokena \AB@affilsep
          \AB@affilnote{}\protect\Affilfont\AB@temp}%
      \endgroup
       \let\AB@affilsep\AB@affilsepx
}
\makeatother
\title{Bounding the asymptotic quantum value of all multipartite compiled non-local games}
\date{}

    \author[1]{Matilde Baroni}
    \author[2]{Dominik Leichtle}
    \author[3]{Siniša Janković}
    \author[4]{Ivan Šupić}
    \affil[1]{Sorbonne Université, CNRS, LIP6, 4 place Jussieu, 75005 Paris, France}
    \email[1]{\url{matilde.baroni@lip6.fr}}
    \affil[2]{School of Informatics, University of Edinburgh, 10 Crichton Street, Edinburgh EH8 9AB, United Kingdom}
    \email[2]{\url{dominik.leichtle@ed.ac.uk}}
    \affil[3]{Faculty of Physics, University of Belgrade, Studentski Trg 12-16, 11000 Belgrade, Serbia}
    \affil[4]{Universit\'e Grenoble Alpes, CNRS, Grenoble INP, LIG, 38000 Grenoble, France}
    \email[4]{\url{ivan.supic@univ-grenoble-alpes.fr}}

\begin{document}

\maketitle

\begin{abstract}

    Non-local games are a powerful tool to distinguish between correlations possible in classical and quantum worlds.
    Kalai \emph{et al.} (STOC'23) proposed a compiler that converts multipartite non-local games into interactive protocols with a single prover, relying on cryptographic tools to remove the assumption of physical separation of the players.
    While quantum completeness and classical soundness of the construction have been established for all multipartite games, quantum soundness is known only in the special case of bipartite games.

    In this paper, we prove that the Kalai \emph{et al.}'s compiler indeed achieves quantum soundness for all multipartite compiled non-local games, by showing that any correlations that can be generated in the asymptotic case correspond to quantum commuting strategies.

    Our proof uses techniques from the theory of operator algebras, and relies on a characterisation of sequential operationally no-signalling strategies as quantum commuting operator strategies in the multipartite case, thereby generalising several previous results.
    On the way, we construct universal C*-algebras of sequential PVMs and prove a new chain rule for Radon-Nikodym derivatives of completely positive maps on C*-algebras which may be of independent interest.      \\
    \\
    \textbf{Keywords:} nonlocality, quantum cryptography, operator algebra.
\end{abstract}

\newpage
\tableofcontents

\newpage
\section{Introduction}

Non-local games are an indispensable tool in device-independent quantum certification, enabling one to verify the presence of quantum properties, such as entanglement, without characterization or trustworthiness of the devices involved~\cite{Brunner_2014}.
The standard setup involves a verifier and at least two provers - also called devices; the verifier distributes classical inputs to each prover and checks whether the returned classical outputs satisfy some condition, called the game predicate.
During this interaction, the provers are not allowed to communicate; however they are allowed to share a strategy before the game starts, for example they are allowed to share an entangled resource.
The success probability in such games enables to conclude device-independently that the provers possess genuine quantum capabilities.

The major drawback of traditional device-independent certification via non-local games is the enforcement of the no-communication requirement.
The most standard choice is to impose that the provers are space-like separated, but ensuring it is experimentally demanding and restricts applicability. This motivates the question: \emph{Can we transfer the power of non-local games, in particular their strong quantum-certification guarantees, from a multi-prover setting to a single-prover scenario?}

The work of Kalai, Lombardi, Vaikuntanathan, and Yang~\cite{KLVY22Quantum} proposes a solution to this problem.
They define a compiler that maps every $\mathsf{k}$-players non-local game into a single-prover interactive protocol of $2\mathsf{k}$-rounds, by replacing the physical separation constraint with cryptographic functionalities.
More specifically, they propose to use quantum fully homomorphic encryption (QFHE)~\cite{mahadev2020classical,brakerski2018quantum}
to encrypt the classical communication between the verifier and the single prover; this allows to maintain the prover’s ignorance about the verifier’s questions, by still allowing them to homomorphically compute circuits on the encrypted inputs.

Kalai \emph{et al.} showed that this compilation procedure preserves many desirable properties of the original non-local game.
Using the homomorphic property they prove classical/quantum completeness for all $\mathsf{k}$-partite games, that is, there is an explicit and efficient classical/quantum strategy that achieves the classical/quantum bound of the original nonlocal game.
From the security Kalai \emph{et al.} also derive the classical soundness for the compilation of all $\mathsf{k}$-partite games, showing that a classical single prover cannot achieve a winning probability exceeding that of classical provers in the original non-local game, up to negligible deviations dictated by the cryptographic security parameter.
These two statements together imply that compiled non-local games demonstrate a quantum-classical gap in the single prover setting.
However, in \cite{KLVY22Quantum} they did not provide an upper bound to the optimal score of a quantum single prover.
Since then, a considerable amount of effort has been put into the characterisation of quantum compiled strategies.
Follow-up works by Natarajan and Zhang~\cite{NZ2023Bounding} and Brakerski \emph{et al.}~\cite{brakerski2023simple} prove that the quantum bound is preserved in the specific case of the compiled CHSH game.
Later advancements generalized quantum soundness to broader classes of games, like XOR games, as shown independently by Baroni \emph{et al.}~\cite{baroni2024quantumboundscompiledxor} and Cui \emph{et al.}~\cite{CMMN2024Computational}, and tilted CHSH games as shown by Mehta \emph{et al.}~\cite{mehta2024selftestingcompiledsettingtiltedchsh}.
More recently, a result by Kulpe \emph{et al.}~\cite{KMPSW24bound} demonstrated that, in the limit of the cryptography being perfect, the quantum bound of any bipartite compiled game remains upper-bounded by the commuting-operator quantum bound of the original non-local game.

All of this previous work leaves the behaviour of quantum provers in compiled protocols resulting from \emph{multipartite} non-local games almost entirely unexplored.
It is known however that in the multipartite setting effects arise that differ fundamentally from the nature of bipartite quantum non-locality and entanglement.
In fact, genuinely multipartite phenomena such as post-quantum steering~\cite{postquantum-steering} seemed to indicate that a generalization of Kulpe \emph{et al.}'s result~\cite{KMPSW24bound} beyond two players, \emph{i.e.}, quantum soundness of the KLVY compiler in the multipartite case, might not hold.
A generalization of their bound to scenarios beyond two players, even for any specific game, remained so far unknown, being made more complicated by the fact that all proof techniques up to now rely on the bipartite structure of the provers' strategies.

Beyond being interesting from a purely theoretical perspective, compiling multipartite games is of practical relevance.
Experimentally enforcing the pairwise space-like separation of multiple players might be considerably more difficult than in the bipartite case, or even impossible in certain scenarios.
Hence, the use of cryptography to replace this necessity might be of substantial benefit.
Finally, as remarked by Metger \emph{et al.}~\cite{metger2024succinctargumentsqmastandard}, sound compilation techniques for nonlocal games with larger numbers of players might have consequences as well in the fields of cryptography and complexity theory, more specifically in the construction of \emph{succinct classical arguments for QMA}.

In this work, we explore for the first time the quantum soundness for the compilation of multipartite non-local games. Our main contribution is a multipartite extension of the asymptotic quantum soundness result for all bipartite scenarios.
More specifically, we prove that for all compiled multipartite games, in the limit of the cryptography being perfect, the quantum prover’s achievable winning probability is upper-bounded by the commuting-operator quantum bound of the original multipartite non-local game. 

\subsection{Technical overview of our results}

We begin by recalling the core ideas of the KLVY compiler.
In a general $\mathsf{k}$-partite non-local game, $\mathsf{k}$ space-like separated provers receive from a verifier a classical input and return a classical output, possibly following a pre-shared strategy that can make use of pre-shared resources.
The compilation procedure maps any such game to a single-prover interactive game, with $\mathsf{k}$ rounds of interaction. 
In every round, the verifier sends an encrypted input and waits for the corresponding encrypted output; intuitively, every round corresponds to a non-local prover.
A graphical representation of this can be found in Figure~\ref{fig:scheme}, more precisely the vertical arrow on the left.

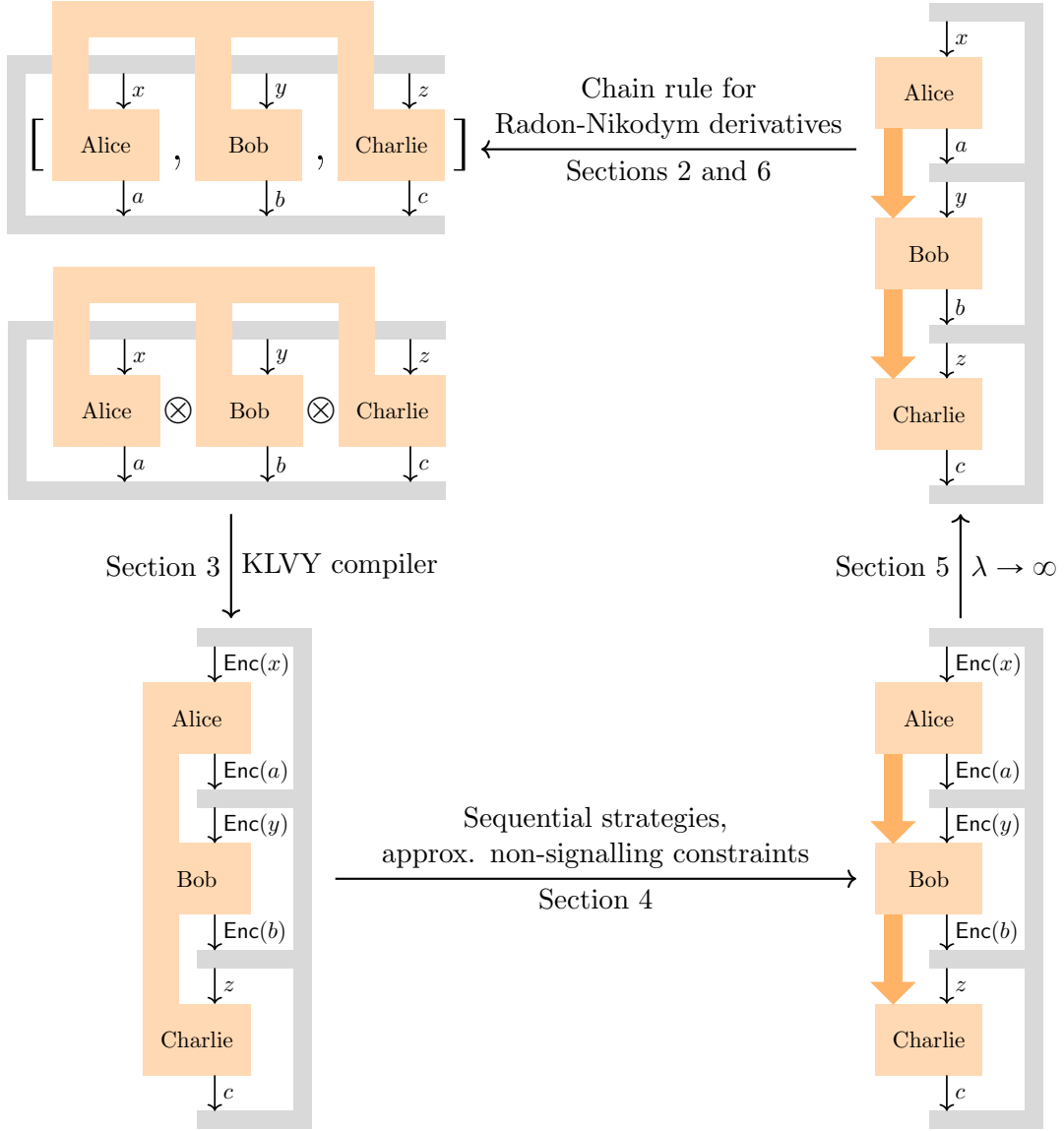
\begin{figure}[htp]
    \centering
\begin{tikzpicture}[scale=1.4, every node/.style={anchor=center}]

  \coordinate (A) at (-1, 0);   %
  \coordinate (B) at (6, 0);    %
  \coordinate (C) at (6, 6);    %
  \coordinate (D) at (-1, 4.75);   %
  \coordinate (E) at (-0.85, 7.3);   %

  \node at (A) {\scalebox{0.8}{%

\begin{tikzpicture}[scale=0.3]
    \fill[gray!30] (8.4,0) -- (9.4,0) --(9.4,28) -- (8.4,28) -- cycle;
    \fill[gray!30] (3,0) -- (8.5,0) --(8.5,1) -- (3,1) -- cycle;
    \fill[gray!30] (3,9) -- (8.5,9) --(8.5,10) -- (3,10) -- cycle;
    \fill[gray!30] (3,18) -- (8.5,18) --(8.5,19) -- (3,19) -- cycle;
    \fill[gray!30] (3,27) -- (8.5,27) --(8.5,28) -- (3,28) -- cycle;

    \fill[orange!30] (2,3) rectangle (6,7);
    \fill[orange!30] (2,12) rectangle (6,16);
    \fill[orange!30] (2,21) rectangle (6,25);

    \fill[orange!30] (0,3) -- (2,3) -- (2,25) -- (0,25) -- cycle;

    \node at (3,5){Charlie};
    \node at (3,14) {Bob};
    \node at (3,23) {Alice};

    \draw[->, thick] (4,27) -- (4,25) node[right,midway] {$\mathsf{Enc}(x)$};
    \draw[->, thick] (4,18) -- (4,16)node[right,midway] {$\mathsf{Enc}(y)$};
    \draw[->, thick] (4,9) -- (4,7) node[right,midway] {$z$};
    
    \draw[->, thick] (4,21) -- (4,19)node[right,midway] {$\mathsf{Enc}(a)$};
    \draw[->, thick] (4,12) -- (4,10)node[right,midway] {$\mathsf{Enc}(b)$};
    \draw[->, thick] (4,3) -- (4,1)node[right,midway] {$c$};

\end{tikzpicture}
}};
  \node at (B) {\scalebox{0.8}{%

\begin{tikzpicture}[scale=0.3]
    \fill[gray!30] (8.4,0) -- (9.4,0) --(9.4,28) -- (8.4,28) -- cycle;
    \fill[gray!30] (3,0) -- (8.5,0) --(8.5,1) -- (3,1) -- cycle;
    \fill[gray!30] (3,9) -- (8.5,9) --(8.5,10) -- (3,10) -- cycle;
    \fill[gray!30] (3,18) -- (8.5,18) --(8.5,19) -- (3,19) -- cycle;
    \fill[gray!30] (3,27) -- (8.5,27) --(8.5,28) -- (3,28) -- cycle;

    \fill[orange!30] (0,3) rectangle (6,7);
    \fill[orange!30] (0,12) rectangle (6,16);
    \fill[orange!30] (0,21) rectangle (6,25);

    \fill[orange!60] (0.5,12) rectangle (1.5,8);
    \node[isosceles triangle,
	isosceles triangle apex angle=90,
	fill=orange!60,
        rotate=-90,
	minimum size =5] (T90)at (1,7.8){};

        \fill[orange!60] (0.5,21) rectangle (1.5,17);
    \node[isosceles triangle,
	isosceles triangle apex angle=90,
	fill=orange!60,
        rotate=-90,
	minimum size =5] (T90)at (1,16.8){};

    \node at (3,5){Charlie};
    \node at (3,14) {Bob};
    \node at (3,23) {Alice};

    \draw[->, thick] (4,27) -- (4,25) node[right,midway] {$\mathsf{Enc}(x)$};
    \draw[->, thick] (4,18) -- (4,16)node[right,midway] {$\mathsf{Enc}(y)$};
    \draw[->, thick] (4,9) -- (4,7) node[right,midway] {$z$};
    
    \draw[->, thick] (4,21) -- (4,19)node[right,midway] {$\mathsf{Enc}(a)$};
    \draw[->, thick] (4,12) -- (4,10)node[right,midway] {$\mathsf{Enc}(b)$};
    \draw[->, thick] (4,3) -- (4,1)node[right,midway] {$c$};

	isosceles triangle apex angle=90,
	fill=orange!30,
        rotate=-90,
	minimum size =5] (T90)at (6,0){};
    
\end{tikzpicture}

}};
  \node at (C) {\scalebox{0.8}{%

\begin{tikzpicture}[scale=0.3]
    \fill[gray!30] (8.4,0) -- (9.4,0) --(9.4,28) -- (8.4,28) -- cycle;
    \fill[gray!30] (3,0) -- (8.5,0) --(8.5,1) -- (3,1) -- cycle;
    \fill[gray!30] (3,9) -- (8.5,9) --(8.5,10) -- (3,10) -- cycle;
    \fill[gray!30] (3,18) -- (8.5,18) --(8.5,19) -- (3,19) -- cycle;
    \fill[gray!30] (3,27) -- (8.5,27) --(8.5,28) -- (3,28) -- cycle;

    \fill[orange!30] (0,3) rectangle (6,7);
    \fill[orange!30] (0,12) rectangle (6,16);
    \fill[orange!30] (0,21) rectangle (6,25);

    \fill[orange!60] (0.5,12) rectangle (1.5,8);
    \node[isosceles triangle,
	isosceles triangle apex angle=90,
	fill=orange!60,
        rotate=-90,
	minimum size =5] (T90)at (1,7.8){};

        \fill[orange!60] (0.5,21) rectangle (1.5,17);
    \node[isosceles triangle,
	isosceles triangle apex angle=90,
	fill=orange!60,
        rotate=-90,
	minimum size =5] (T90)at (1,16.8){};

    \node at (3,5){Charlie};
    \node at (3,14) {Bob};
    \node at (3,23) {Alice};

    \draw[->, thick] (4,27) -- (4,25) node[right,midway] {$x$};
    \draw[->, thick] (4,18) -- (4,16)node[right,midway] {$y$};
    \draw[->, thick] (4,9) -- (4,7) node[right,midway] {$z$};
    
    \draw[->, thick] (4,21) -- (4,19)node[right,midway] {$a$};
    \draw[->, thick] (4,12) -- (4,10)node[right,midway] {$b$};
    \draw[->, thick] (4,3) -- (4,1)node[right,midway] {$c$};

	isosceles triangle apex angle=90,
	fill=orange!30,
        rotate=-90,
	minimum size =5] (T90)at (6,0){};
    
\end{tikzpicture}

}};
  \node at (D) {\scalebox{0.8}{%
\begin{tikzpicture}[scale=0.3]
    \fill[gray!30] (0,0) -- (24,0) -- (24,1) -- (0,1) -- cycle;
    \fill[gray!30] (-0.5,0) -- (0.5,0) -- (0.5,10) -- (-0.5,10) -- cycle;
    \fill[gray!30] (0,9) -- (24,9) -- (24,10) -- (0,10) -- cycle;

    \fill[orange!30] (4,3) rectangle (8,7);
    \fill[orange!30] (12,3) rectangle (16,7);
    \fill[orange!30] (20,3) rectangle (24,7);

    \draw[->, thick] (6,9) -- (6,7) node[right,midway] {$x$};
    \draw[->, thick] (14,9) -- (14,7) node[right,midway] {$y$};
    \draw[->, thick] (22,9) -- (22,7) node[right,midway] {$z$};
    
    \draw[->, thick] (6,3) -- (6,1) node[right,midway] {$a$};
    \draw[->, thick] (14,3) -- (14,1) node[right,midway] {$b$};
    \draw[->, thick] (22,3) -- (22,1) node[right,midway] {$c$};

    \fill[orange!30] (2,11) -- (20,11) -- (20,13) -- (2,13) -- cycle;
    \fill[orange!30] (2,3) -- (4,3) -- (4,11) -- (2,11) -- cycle;
    \fill[orange!30] (10,3) -- (12,3) -- (12,11) -- (10,11) -- cycle;
    \fill[orange!30] (18,3) -- (20,3) -- (20,11) -- (18,11) -- cycle;

    \node at (5,5) {Alice};
    \node at (13,5) {Bob};
    \node at (21,5) {Charlie};

    \node at (9,5) {\LARGE $\otimes$};
    \node at (17,5) {\LARGE $\otimes$};

\end{tikzpicture} %
}};
  \node at (E) {\scalebox{0.8}{%
\begin{tikzpicture}[scale=0.3]
    \fill[gray!30] (0,0) -- (24,0) -- (24,1) -- (0,1) -- cycle;
    \fill[gray!30] (-0.5,0) -- (0.5,0) -- (0.5,10) -- (-0.5,10) -- cycle;
    \fill[gray!30] (0,9) -- (24,9) -- (24,10) -- (0,10) -- cycle;

    \fill[orange!30] (4,3) rectangle (8,7);
    \fill[orange!30] (12,3) rectangle (16,7);
    \fill[orange!30] (20,3) rectangle (24,7);

    \draw[->, thick] (6,9) -- (6,7) node[right,midway] {$x$};
    \draw[->, thick] (14,9) -- (14,7) node[right,midway] {$y$};
    \draw[->, thick] (22,9) -- (22,7) node[right,midway] {$z$};
    
    \draw[->, thick] (6,3) -- (6,1) node[right,midway] {$a$};
    \draw[->, thick] (14,3) -- (14,1) node[right,midway] {$b$};
    \draw[->, thick] (22,3) -- (22,1) node[right,midway] {$c$};

    \fill[orange!30] (2,11) -- (20,11) -- (20,13) -- (2,13) -- cycle;
    \fill[orange!30] (2,3) -- (4,3) -- (4,11) -- (2,11) -- cycle;
    \fill[orange!30] (10,3) -- (12,3) -- (12,11) -- (10,11) -- cycle;
    \fill[orange!30] (18,3) -- (20,3) -- (20,11) -- (18,11) -- cycle;

    \node at (5,5) {Alice};
    \node at (13,5) {Bob};
    \node at (21,5) {Charlie};

    \node[left] at (2,5) {\Huge $[$};
    \node at (9,4) {\Huge ,};
    \node at (17,4) {\Huge ,};
    \node[right] at (24,5) {\Huge $]$};%

\end{tikzpicture} %
}};

  \draw[->, thick] (-1,3.5) -- node[right] {KLVY compiler} node[left] {Section~\ref{sec:compiled-non-local-games}} (-1,2.5);
  \draw[->, thick] (0,0) -- node[above, text width=7cm, align=center] {Sequential strategies, \\ approx. non-signalling constraints} node[below] {Section~\ref{sec:quantum-compiler}} (5,0);
  \draw[->, thick] (6,2.5) -- node[right] {$\lambda \to \infty$} node[left] {Section~\ref{sec:algebraic-strategies}} (6,3.5);
  \draw[->, thick] (5,7) -- node[above, text width=5cm, align=center] {Chain rule for \\ Radon-Nikodym derivatives} node[below] {Sections~\ref{sec:chain-rule-rn}~and~\ref{sec:asymptotic-bound}} (1.4,7);

\end{tikzpicture}     \caption{Graphical overview of our results and contributions.
    In all scenarios, the verifier is depicted in grey and the prover(s) are depicted in orange.
    The horizontal axis represents the spatial dimension, whereas time is flowing vertically, from top to bottom.
    Although all scenarios are depicted here involving three players, all results in our manuscript directly generalise and are proven in the more general scenario of any finite number of players.
    \textit{Starting at the top-left, counter-clockwise:}
    non-local game;
    KLVY-compiled game without space-like separations \textit{(bottom-left)};
    Sequential game with approximate no-signalling \textit{(bottom-right)};
    Algebraic strategies with operational no-signalling \textit{(top-right)};
    non-local game with commuting-operator strategy \textit{(top-left, again)}.
    }
    \label{fig:scheme}
\end{figure}

In this work we focus on the quantum soundness of the KLVY compiler.
While classical soundness was proven already in the original protocol, the proof techniques couldn't be directly translated to the quantum setting.
Kulpe \emph{et al.}~\cite{KMPSW24bound} showed the first general result about quantum soundness for compiled games; more precisely, they proved that the asymptotic quantum value of any compiled two-player game is upper bounded by the game’s commuting-operator value.
Their approach combines cryptographic and operator-algebraic techniques; we will now recall their core ideas, as they set the stage for our generalisations, which extend their result to all $\mathsf{k}$-players games.

\begin{theorem*}[Informal version of Theorem~\ref{th:asymptotic-quantum-soundness-k-players}] For every $\mk$-player non-local game $\cG$, the asymptotic quantum value of the compiled game $\cG_\la$ is upper bounded by the quantum commuting operator value of $\cG$, \emph{i.e.}, $\limsup_{\la\to\infty}\omega_q(\cG_\la) \leq \omega_{qc}(\cG)$.
\end{theorem*}

One of the recurring themes when generalising the previous results of ~\cite{KMPSW24bound} to the multipartite case is the need to deal with quantum instruments, \emph{i.e.}, maps of classical-quantum information.
In the bipartite case, it is sufficient to handle quantum states and measurements as the two natural objects that arise in the study of quantum sequential strategies in compiled games.
When considering more than two players, any player in a sequential strategy that is neither the first nor the last, accepts classical input and a quantum state, and implements a physical map that outputs classical information alongside quantum information that is passed on to the next player.
These players are most generally described by aforementioned \emph{quantum instruments}.
Besides being a natural generalisation of both states and measurements, quantum instruments additionally solve the issue of composition.
Algebraically, this will later analogously require us to generalise from states to the study of completely positive maps.

The first key idea of~\cite{KMPSW24bound} is a sequential characterisation of bipartite quantum commuting operator correlations, motivated by the conditions imposed on strategies in the compiled game by idealised cryptography with perfect security. They recognise that this problem can be solved as a corollary of a famous algebraic theorem, the Radon-Nikodym Theorem for positive functionals.
Our first contribution is a generalisation of this result to arbitrary numbers of participating players: a characterisation of multipartite quantum commuting operator correlations as exactly those correlations arising from sequential strategies with the new notion of \emph{operational no-signalling}.
This new notion in the context of non-local games is reminiscent of the notions of preparation and transformation equivalences that are considered in contextuality, establishing an interesting link between the two settings.
The new characterisation is obtained as a consequence of two core observations.
First, Belavkin's generalisation of the Radon–Nikodym Theorem \cite{BS86radon} from functionals to completely positive maps is sufficient to treat the case of quantum instruments.
Second, we require a composition theorem to preserve commuting relations through the composition of quantum instruments, \emph{i.e.}, to deal with iterative applications of Belavkin's Radon-Nikodym theorem.
To this end, we prove a new chain rule for Radon-Nikodym derivatives of CP maps.
In Figure~\ref{fig:scheme}, this step is depicted by the top arrow.

The second technical contribution of~\cite{KMPSW24bound} is to link computational security with information-theoretic security in the limit of the security parameter approaching infinity.
To this end, appropriate algebraic structure is introduced which is fit to capture and explain the mathematical objects arising in this limit.
As our second core contribution, we construct the algebraic structures necessary to generalise these arguments to deal with quantum instruments and compositions of quantum instruments, that are naturally arising in the study of the multipartite case.
Then, by appropriately characterising the constraints arising from the security of the cryptographic scheme for finite levels of security $\lambda$, we can take the limit of $\lambda$ approaching infinity and obtain an idealised version of perfect security for the limiting objects.
This argumentation corresponds to the arrows in the lower right corner of Figure~\ref{fig:scheme}.
Relying on the previously developed characterisation of the thus arising correlations, we finally end up with a characterisation in terms of a fully commuting operator strategy.

Thanks to the sequentiality of the compiled protocol, a single prover in the compiled game can be described by an interactive quantum algorithm with $\mk$ rounds of interaction.
In this way, the first round can be modelled as a state preparation, the last one as a measurement, and the ones in between as instruments.
The causal structure imposed by sequentiality immediately implies one-way no-signalling, since the future cannot signal to the past. The security of the employed cryptography enforces approximate no-signalling in the other direction: the prover’s responses can not directly depend on the plaintext inputs of other rounds, even those in the past.
The authors of~\cite{NZ2023Bounding,KMPSW24bound} realised that in the case of state preparation, an even stronger condition is true.
If $\sigma_{a|x}^\lambda$ describes the assemblage of states prepared by the prover during the first round of the compiled game, they observed that
\begin{equation*}
    \sum_a \sigma_{a|x}^\lambda \approx_\lambda \sum_{a'} \sigma_{a'|x'}^\lambda, \qquad \forall x,x'.
\end{equation*}
This is to be understood as an algebraic condition and hence goes beyond mere conditions on resulting correlations.
Employing block-encoding techniques, we show that operators implementable by quantum-polynomial-time (QPT) circuits in fact satisfy stronger algebraic constraints, an approximate form of operational no-signalling for quantum instruments.
Denoting by $I_{b|y}^\lambda$ the quantum instrument implemented by the prover in any of the other rounds, we prove that the following is true:
\begin{equation*}
    \sum_b I_{b|y}^\lambda \approx_\lambda \sum_{b'} I_{b'|y'}^\lambda, \qquad \forall y,y'.
\end{equation*}

However, because the dimension of the Hilbert space underlying each strategy may vary with $\la$, it becomes difficult to extract a limiting object directly. Therefore, we transition to an algebraic model that captures these strategies uniformly. We define a sequence of universal $C^*$-algebras generated by the sequential measurements associated with each round. Any efficient quantum prover’s strategy induces a *-homomorphism from these universal $C^*$-algebras to bounded operators on some Hilbert space. This representation facilitates the application of compactness arguments, allowing us to extract a limiting strategy.

Finally, by combining the algebraic representation, compactness, and the chain rule, we extract a limiting strategy from the sequence of algebraic strategies indexed by $\la$, and show that this limiting strategy corresponds to a $\mk$-partite commuting-operator strategy. This establishes the desired asymptotic upper-bound on the quantum value of compiled games.

\subsection{Technical contributions of independent interest}

In this subsection we highlight several auxiliary technical contributions that allowed achieving the above result. These may be of independent interest beyond the scope of compiled games, as they generalize or refine tools in operator algebras, quantum information, and quantum cryptography.

\paragraph{Chain rule for Radon–Nikodym derivatives of CP maps on C*-algebras.}
The Radon-Nikodym Theorem allows the representation of a CP map between C*-algebras relative to the Stinespring dilation of a dominant CP map on a concrete Hilbert space.
In this representation, the dominated CP map appears in form of a bounded operator on this Hilbert space, the so-called Radon-Nikodym derivative, that commutes with the image of the representation of the dominant CP map.
The main obstacle for the generalization of this result to a chain of CP maps is that repeated applications of the Radon-Nikodym Theorem and thereby repeated dilations lead to a hierarchy of Hilbert spaces connected by isometries and derivative operators scattered all over the hierarchy.
Straightforward lifting of the derivatives through the isometries into the same Hilbert space breaks commutation relations and the unitality necessary to identify the derivative operators with POVM elements.
We circumvent these caveats through the use of alternative liftings that rely on the explicit structure of the involved Stinespring dilations and the Radon-Nikodym commutation relations, and establish a novel chain rule that allows one to compose Radon–Nikodym derivatives for a sequence of completely positive (CP) maps on C*-algebras.
The consequence is a representation of chains of CP maps on the Hilbert space given by the Stinespring dilation of a dominant chain of CP maps.
All commutation relations of the Radon-Nikodym derivatives obtained in this manner are preserved.
\begin{theorem*}[Informal version of theorem~\ref{th:chained-rn-k}]
    Let $S^{(i)}, R^{(i)}$ for $i=1,\dots,k$ be chains of completely positive maps on C*-algebras such that $S^{(i)} \leq R^{(i)}$. Then, there exists a representation on a Hilbert space $\mathcal{K}$ such that
    \begin{align*}
        S^{(1)} \circ \dots \circ S^{(k)} (\mathfrak{a}) = V^\ast F^{(1)} \dots F^{(k)} \pi(\mathfrak{a}) V, \qquad \forall \mathfrak{a} \in \mathscr{A},
    \end{align*}
    where all Radon-Nikodym derivatives, the bounded operators $F^{(i)} \in \mathsf{B}(\mathcal{K})$, are pairwise commuting, $V : \mathcal{H} \to \mathcal{K}$ is a bounded operator, and $\pi : \mathscr{A} \to \mathsf{B}(\mathcal{K})$ is a *-representation. 
\end{theorem*}
In practical terms, if we have multiple measurement processes in sequence (each giving rise to a CP map relative to a prior state or algebra), our chain rule provides a systematic way to derive a single cumulative Radon–Nikodym derivative that captures all stages at once. This result is crucial for inductively extending two-player operator-algebraic techniques to $k$ players, as it ensures consistency when relating the algebraic strategies across sequential rounds.

\paragraph{Characterization of sequential operationally non-signalling correlations as commuting-operator correlations.}
The Radon-Nikodym Theorem can be seen as an algebraic and infinite-dimensional generalization of the finite-dimensional S-G-HJW theorem.
Building upon this powerful tool, we give a complete characterization of the set of correlations arising from $k$-partite sequential \emph{operationally non-signalling} strategies (i.e. no information is transmitted between different rounds of the interaction) by showing that they exactly coincide with the set of $k$-partite commuting-operator correlations. Conversely, any commuting-operator strategy for the original $k$-player game can be realized by a $k$-player sequential strategy that satisfies the operational non-signalling constraints.
This bridges the gap between the cryptographic limits of computational single-prover models and the usual entangled-multiprover model, generalizing earlier two-player results to the fully general sequential scenario.

\begin{theorem*}[Informal version of Corollary~\ref{cor:k-seq-players}]
    For any finite number of players, the correlations generated by sequential, operationally no-signalling strategies are exactly those generated by commuting operator strategies.
\end{theorem*}

\paragraph{Universal C*-algebras of sequential PVMs and POVMs.}
We construct a \emph{universal C*-algebra} that is generated by a collection of sequential measurement operators (PVM or POVM elements) corresponding to multiple rounds of interaction. This construction generalizes the standard universal C*-algebra of a (static) PVM (POVM) to the setting of a sequential PVM (POVM) (measurements performed in a fixed temporal order). The universal property of this algebra means that any concrete quantum strategy implementing those sequential measurements is represented by a *-homomorphism from the universal algebra into the strategy’s operator algebra (see Figure~\ref{fig:intro-univ-algebras-3-players}). This tool allows us to reason about arbitrary sequential strategies in an abstract algebraic way and is instrumental in our proofs of the above results, especially in unifying the treatments of different rounds within a single mathematical framework.

\begin{figure}[ht]
    \begin{subfigure}[t]{0.45\textwidth}
        \centering
\begin{tikzpicture}
    \node (ABC) at (2.5,2) {\( \mathscr{A}_B \)};
    \node (C) at (5.5,0) {\( \mathbb{C} \)};
    \node (BH2) at (2.5,0) {\( \mathsf{B}(\mathcal{H}^\lambda) \)};

    \draw[->] (ABC) -- (C) node[midway, above,  yshift=5pt] {\( \phi_{a|x}^\lambda \)};
    \draw[->] (ABC) -- (BH2) node[midway, left] {\( \theta^\lambda \)};
    \draw[->] (BH2) -- (C) node[midway,below] {\( \mathsf{Tr}(\sigma_{a|x}^\lambda \cdot) \)};

\end{tikzpicture}
     \end{subfigure}
    \begin{subfigure}[t]{0.45\textwidth}
        \centering
\begin{tikzpicture}
    \node (AC) at (0,2) {\( \mathscr{A}_C \)};
    \node (ABC) at (2.5,2) {\( \mathscr{A}_{B \to C} \)};
    \node (C) at (5.5,0) {\( \mathbb{C} \)};
    \node (BH1) at (0,0) {\( \mathsf{B}(\mathcal{H}^\lambda) \)};
    \node (BH2) at (2.5,0) {\( \mathsf{B}(\mathcal{H}^\lambda) \)};

    \draw[->] (AC) -- (ABC) node[midway, above] {\( T_{b|y} \)};
    \draw[->] (ABC) -- (C) node[midway, above,  yshift=5pt] {\( \phi_{a|x}^\lambda \)};
    \draw[->] (AC) -- (BH1) node[midway, left] {\( \theta^\lambda_{C} \)};
    \draw[->] (ABC) -- (BH2) node[midway, left] {\( \theta^\lambda_{BC} \)};
    \draw[->] (BH1) -- (BH2) node[midway,below] {\( B^{\lambda,*}_{b|y} \)};
    \draw[->] (BH2) -- (C) node[midway,below] {\( \mathsf{Tr}(\sigma_{a|x}^\lambda \cdot) \)};

\end{tikzpicture}
     \end{subfigure}
    \caption{Commutative diagrams of universal C*-algebras capturing PVMs (POVMs) in the two-player case \textit{(left)} and sequential PVMs (POVMs) in the three-player case \textit{(right)}.}
    \label{fig:intro-univ-algebras-3-players}
\end{figure}
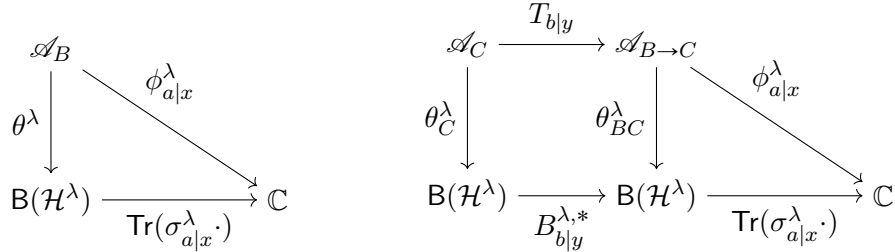

\paragraph{Algebraic consequences of IND-CPA security via block encodings.}
Let $\sigma_{\enc(x)}$ be an efficiently prepared quantum state from the ciphertext of $x$ under some encryption scheme.
\cite{NZ2023Bounding} noticed that the IND-CPA security of encryption schemes imposes approximate no-signalling conditions on $\sigma_{\enc(x)}$ with respect to $x$ not only when this state is measured using an efficiently implementable POVM $\{M_b\}$, but also when given access to ``expectation values'' of polynomials in the POVM elements, $\Tr [ \sigma_{\enc(x)} P\{M_b\}]$.
An extremely useful tool in the treatment of these, in general unphysical, quantities are block encodings; in particular they are used to describe the efficient implementability and estimatability of expectation values of $P\{M_b\}$.
Relying on algorithmic estimators for efficiently block-encodable observables, \cite{NZ2023Bounding} managed to extend approximate no-signalling of states of the form $\sigma_{\enc(x)}$ all the way to algebras generated by elements of efficiently implementable POVMs, a crucial step in the further algebraic treatment of compiled strategies.
Generalizing to the multipartite case, we require a stronger condition that does not only hold for (efficiently preparable) quantum states carrying information about the ciphertext of $x$, but more generally for (efficiently implementable) transformations of quantum information, in the shape of quantum instruments and channels.
\begin{theorem*}[Informal version of theorem~\ref{th:3pl-ind-cpa}]
    Let $\{B_y^\lambda\}$ be a family of QPT-implementable quantum channels which receive as classical input an encryption of $y$. Let further $\mathcal{L}^\lambda, \mathcal{R}^\lambda, \mathcal{M}^\lambda$ be families of uniformly bounded QPT-block-encodable operators, and $\rho^\lambda$ be a family of efficiently preparable states. It then holds that
    \begin{align*}
        \left| \Tr\left[ \mathcal{M}^\lambda B_{y}^\lambda \left( \mathcal{R}^\lambda \rho^\lambda \mathcal{L}^{\lambda,*} \right) \right] - \Tr\left[ \mathcal{M}^\lambda B_{y'}^\lambda \left( \mathcal{R}^\lambda \rho^\lambda \mathcal{L}^{\lambda,*} \right) \right] \right| \leq \textsf{negl}(\lambda),
    \end{align*}
    for all inputs $y,y'$.
\end{theorem*}
This result is a generalization of its analogue for states presented in \cite{NZ2023Bounding,KMPSW24bound}, as states can be seen as special cases of quantum channels without (quantum) input.

\subsection{Discussion and open questions}

This work settles an important open question about the compiled nonlocal games introduced in~\cite{KLVY22Quantum}, by proving that their compiler is asymptotically sound for efficient quantum provers in multipartite games.
In particular, the maximum score that can be reached by efficient quantum provers in the compiled game is asymptotically upper-bounded by the maximum score reachable by commuting-operator quantum strategies in the original game, pre-compilation.
This holds for all multipartite games.
\paragraph{Tighter asymptotic bounds.}
The upper-bound obtained in this paper is given by the commuting-operator quantum value of the nonlocal game.
The compiler's quantum completeness guarantees that the spatial quantum value of the nonlocal game, in which permissible strategies involve non-communicating provers operating in distinct Hilbert spaces, can indeed be achieved by efficient quantum strategies in the compiled game.
Since it is known that there exist games for which these two values do not coincide \cite{ji2022mipre}, this leaves open a gap, and the question whether the asymptotic bounds obtained in this work and previously are indeed tight.

\paragraph{Non-asymptotic quantum soundness.}
One remaining question is whether one can show not just asymptotic but also quantitative quantum soundness. In other words, our result shows that the quantum bound of any compiled game approaches the quantum commuting operator bound of the original non-local game as the security parameter goes to infinity.
This is a necessary prerequisite for solving a more practical question: what can we say about the quantum bound of compiled games for finite values of the security parameter? In~\cite{Xiangling}, the authors answer to this question for all bipartite games.
We believe that their methods involving non-commutative polynomial optimizations and our quantum and algebraic characterization of prover strategies can be used to obtain quantitative quantum soundness for all compiled games.

\paragraph{Self-testing.}
Following the precise characterization of quantum and classical bounds of compiled nonlocal games, it is instructive to explore the self-testing properties of these games. While we already have some results related to self-testing from compiled bipartite games~\cite{NZ2023Bounding,CMMN2024Computational,mehta2024selftestingcompiledsettingtiltedchsh}, these results usually allow to make statements about the provers behavior in the last round, when it operates on plaintexts. It remains to be seen whether our quantum instrument-based approach could allow us to reach more complete characterizations of prover's strategies in compiled games solely from the achieved score in the game.

\paragraph{Compilation from different cryptographic assumptions.}
The compiler proposed in \cite{KLVY22Quantum} relies on the employment of Quantum Fully Homomorphic Encryption (QFHE) as a cryptographic building block.
Another compiler~\cite{Bacho_compiled_trapdoor} has been proposed, relying on Trapdoor Claw-Free Functions.
It is an interesting problem whether a similar compilation of nonlocal games, trading space-like separation for computational assumptions, can be constructed from different, or weaker cryptographic primitives.
As they go beyond the scope of the present paper, we leave these questions for exploration by future work.

\subsection{Organization of the paper}

As this work lies in the intersection of different fields, and contains contributions of independent interest therein, we chose to structure the sections of this paper to be as independent and self-contained as possible, to favour accessibility.
As such, the necessary preliminaries are usually presented in the beginning of each technical section, followed by our technical contributions and applications to the main goal of this work, the quantum soundness of compiled games.
As the bipartite case has been studied in previous work, we usually follow the order of reiterating known results and proof strategies in the bipartite case, before introducing the conceptual novelties used in the generalization to the tripartite case.
The multipartite case then usually follows straightforwardly by iterative application of the techniques developed to treat the tripartite scenario.

The contributions in this work are grouped into the following areas.
We start with Section~\ref{sec:chain-rule-rn}, which is structured to be completely independent from the rest of the paper, and presents the technical proof of the chain rule for the Radon-Nikodym theorem.
Sections~\ref{sec:compiled-non-local-games} through \ref{sec:asymptotic-bound} are sorted to follow the main line of our overall argumentation.
More technically, the contributions are: modelling efficient quantum strategies and block-encodings (Section~\ref{sec:quantum-compiler}), universal C*-algebras of sequential PVMs and limiting algebraic strategies (Section~\ref{sec:algebraic-strategies}), and characterizations of sequential algebraic strategies, and the relation with Radon-Nikodym derivatives of CP-maps and the quantum soundness of compiled multipartite nonlocal games (Section~\ref{sec:asymptotic-bound}).
A graphical overview of this paper's structure is depicted in Figure~\ref{fig:scheme}.

In more detail, Section~\ref{sec:chain-rule-rn} focuses solely on the formal proof of the chain rule for Radon-Nikodym theorem.
\begin{itemize}
    \item Section~\ref{subsec:prelims-c-star-algebras} provides preliminaries on C*-algebras, and representation and dilation theorems.
    \item Section~\ref{subsec:chain-rule-rn-2-maps} formulates and proves the chain rule for Radon-Nikodym derivatives for the composition of two CP maps.
    \item Finally, Section~\ref{subsec:chain-rule-rn-general} then states and proves the most general form of the chain rule for compositions of arbitrary finite numbers of CP maps.
\end{itemize}

Section~\ref{sec:compiled-non-local-games} introduces the technical definitions of nonlocal and compiled games:
\begin{itemize}
    \item Section~\ref{subsec:multi-nl} defines nonlocal games and different classes of correlations, with particular attention to the often overlooked multipartite case.
    \item Section~\ref{subsec:KLVYcompiler} describes the KLVY compiler and the QFHE cryptographic primitive. This includes a discussion of prior completeness and soundness results.
\end{itemize}

In Section~\ref{sec:quantum-compiler}, we aim to model the strategies that can arise from efficient quantum provers in the compiled game. Up to this point, our results are of a non-asymptotic nature.
\begin{itemize}
    \item Section~\ref{subsec:eff-quantum-str} mathematically models the operations of quantum provers in the compiled game, and discusses a way to purify the prover's operations to unitaries and projective measurements.
    \item In Section~\ref{subsec:constraints-quantum-IND-CPA}, we use the sequential structure of the protocol to derive conditions that the cryptography imposes on the correlations that the now time-like separated players can feasibly generate. This is formulated in the language of block-encodings and includes a generalization of previous block-encoding techniques to quantum instruments.
\end{itemize}

Section~\ref{sec:algebraic-strategies} introduces useful mathematical structures, specific universal C*-algebras, that allow us to capture and discuss the asymptotic behaviour of QPT provers for the limit of the security parameter tending to infinity. It turns out that the conditions that we found in Section~\ref{sec:quantum-compiler} exhibit sufficient continuity in the security parameter to give rise to algebraic strategies of the players with an asymptotically idealised notion of no-signalling communication.
\begin{itemize}
    \item Section~\ref{subsec:algebras} provides the necessary preliminaries on universal C*-algebras studied in the bipartite case.
    \item In Section~\ref{subsec:universal-algebra-multipartite}, the theory to describe multipartite strategies in a C*-algebraic framework is developed. This includes the construction of universal C*-algebras of sequential measurements, and of the maps connecting them, eventually creating a hierarchy of universal C*-algebras for a growing number of players.
    \item Section~\ref{subsec:asy-eff-strat} studies limits and accumulation points of sequences of algebraic strategies. Crucially, it establishes the (perfect) no-signalling properties that these accumulation points inherit from the approximate no-signalling of the strategies studied in Section~\ref{sec:quantum-compiler}.
\end{itemize}

Section~\ref{sec:asymptotic-bound} identifies the missing piece to map the obtained asymptotic correlations back to the nonlocal picture, and closes this gap by applying the generalised chain rule for Radon-Nikodym derivatives of completely positive maps on C*-algebras.
\begin{itemize}
    \item Section~\ref{subsec:preliminaries-sequential-algebraic} provides the necessary preliminaries on sequential algebraic correlations.
    \item Section~\ref{subsec:operationally-no-sig-commuting-operator} shows that the correlations generated by sequential operationally-no-signalling algebraic strategies are exactly those generated by commuting-operator strategies, by applying the chain rule developed in Section \ref{sec:chain-rule-rn}.
    \item Finally, Section~\ref{subsec:connection-compiled} puts all of these pieces together and (almost) closes the loop by showing that the correlations asymptotically generated by QPT provers in KLVY-compiled nonlocal games correspond to those achievable by commuting-operator strategies in the original game.
\end{itemize}

Appendix~\ref{app:homomorphism} discusses the subtleties in purifying the prover's operations in the compiled game.

In Appendix~\ref{app:alternative-proof}, we present an alternative proof to the one presented in the main text of this paper.
The alternative proof differs from the proof in the main manuscript by avoiding the need to purify the prover's operations.
Instead of dealing with PVMs, it uses POVMs, and instead of $*$-homomorphisms describing purified quantum channels, it uses CP-maps which are the natural maps taking POVMs to POVMs.

\newpage
\section{A chain rule for the Radon-Nikodym Theorem}%
\label{sec:chain-rule-rn}

C$^*$-algebras sit at the crossroads of algebra, analysis, and geometry, providing the natural mathematical framework to describe quantum mechanics. With a single mathematical structure they encode observables as self‑adjoint elements, states as positive linear functionals, dynamics as $*$-automorphisms, and composition via tensor products, capturing exactly the operational features of quantum statistics.
Moreover, abstract C$^*$-algebraic tools such as the GNS construction and Stinespring dilation turn questions about quantum correlations into questions about representations and extensions, giving powerful criteria for characterizing quantum correlations.

\subsection{Preliminaries on C*-algebras}%
\label{subsec:prelims-c-star-algebras}

\begin{definition}[C$^*$-algebra]
    A C$^*$-algebra $\mathscr{A}$ is a Banach algebra over $\mathbb{C}$ with involution (denoted as $*$) that satisfies the $\text{C}^*$-identity
    \begin{align*}
        \| \mathfrak{a} \mathfrak{a}^* \| = \| \mathfrak{a} \| \| \mathfrak{a}^* \|, \quad \forall \mathfrak{a} \in \mathscr{A}.
    \end{align*}
\end{definition}

A \cs-algebra is separable if it contains a countable dense subset. A (two-sided) ideal $I$ in $\mathscr{A}$ is a linear subspace $I \subseteq \mathscr{A}$ such that $\mathfrak{b}\mathfrak{a}, \mathfrak{a}\mathfrak{b} \in I$ for all $\mathfrak{b} \in I$ and $\mathfrak{a} \in \mathscr{A}$.  We declare two elements $\mathfrak{a}, \mathfrak{b}\in \mathscr{A}$ to be equivalent when their difference lies in $I$. The set of equivalence classes, written $\mathscr{A}/I={[\mathfrak{a}]: \mathfrak{a}\in \mathscr{A}}$, inherits natural algebraic operations. When $I$ is closed and two‑sided, the norm $\left\|[\mathfrak{a}]\right\| = \inf\{\left\|\mathfrak{a} + \mathfrak{b}\right\|: \mathfrak{b} \in I\}$ turns $\mathscr{A}/I$ into a \cs-algebra, called the quotient of $\mathscr{A}$ by $I$. The canonical quotient map is the surjective $*$‑homomorphism $\pi_I:\mathscr{A}\to \mathscr{A}/I$ given by $\pi(\mathfrak{a}):=\mathfrak{a}+I=[\mathfrak{a}]$.

Every C$^*$‑algebra admits a concrete realisation as bounded operators on a Hilbert space via a representation.
\begin{definition}[Representation $(\pi, \mathcal{H})$ of $\mathscr{A}$]
    A representation of a $C^*$-algebra $\mathscr{A}$ is a pair $(\pi, \mathcal{H})$ of a Hilbert space $\mathcal{H}$ and a $*$-homomorphism $\pi: \mathscr{A}\to B(\mathcal{H})$.
    If $\mathscr{A}$ is unital, then it is required that $\pi(\mathfrak{1})= \mathds{1}_{\mathcal{H}}$.
    A representation is cyclic if there is a vector $\ket{e} \in \mathcal{H}$ such that $\{ \pi(\mathfrak{a}) \ket{e} |\forall  \mathfrak{a} \in \mathscr{A} \}$ is dense in $\mathcal{H}$.
\end{definition}

A state on a \cs-algebra provides the algebraic counterpart of a quantum system's preparation. It is a linear functional $\phi: \mathscr{A} \to \mathbb{C}$ that is positive, meaning $\phi(a^*a) \geq 0$ for every $a \in \mathscr{A}$, and normalized, i.e. $\phi(1_\mathscr{A}) = 1$.

The Gelfand-Naimark-Segal (GNS) construction gives an explicit recipe to build a cyclic representation of every $C^*$-algebra starting from a state.
\begin{theorem}[Gelfand-Naimark-Segal (GNS) construction]\label{th:GNS}
    Let $\mathscr{A}$ be a unital $C^*$-algebra, and $\phi$ a state on $\mathscr{A}$. Then, there is a cyclic representation $(\pi_\phi, \mathcal{H}_\phi)$ with unit cyclic vector $\ket{\Omega_\phi}$ such that for all $\mathfrak{a}\in \mathscr{A}$
    \begin{equation*}
        \phi(\mathfrak{a}) = \bra{\Omega_\phi} \pi_\phi(\mathfrak{a}) \ket{\Omega_\phi}.
    \end{equation*}
\end{theorem}

This very foundational result underlines a deep connection between C*-algebras and Hilbert spaces, given that C*-algebras can always be concretely realised as subalgebras of the space of bounded operators on some Hilbert space.

It is a natural follow-up question to investigate how the representation of different positive functionals are related.
The Radon–Nikodym Theorem has this spirit, in essence, it states that the Hilbert space from a representation of a dominant functional is always sufficiently large to serve for the representation of a dominated functional.
The two representations are related through the introduction of a bounded self-adjoint operator, which, in reminiscence of measure theory, is usually called the \emph{Radon-Nikodym derivative}.

\begin{theorem}[Radon-Nikodym Theorem for positive linear functionals \cite{conway}]\label{th:rn}
    Let $\psi$ and $\phi$ be positive linear functionals on $\mathscr{A}$ such that $\psi \leq \phi$.
    Let $(\mathcal{H}_\phi, \pi_\phi, \Omega_\phi)$ be the GNS triplet associated to $\phi$.
    Then there exists a unique self-adjoint operator $D \in \pi_\phi(\mathscr{A})' \subseteq B(\mathcal{H}_\phi)$ with $0 \leq D \leq \mathds{1}$, such that
    \begin{equation*}
        \psi(\mathfrak{a}) = \bra{\Omega_\phi} D \pi_\phi(\mathfrak{a}) \ket{\Omega_\phi},
        \qquad \forall \mathfrak{a} \in \mathscr{A}.
    \end{equation*}
\end{theorem}

A special case of Theorem~\ref{th:rn} is given by positive decompositions of functionals.
Lemma~\ref{lem:RN-sums} gives a joint Hilbert space representation of all functionals in such a decomposition.
We will later see the physical meaning of this decomposition as a joint purification of different assemblages realizing the same state.

\begin{lemma}[Theorem 3.3 of \cite{Rag03radon}]\label{lem:RN-sums}
    Consider a positive linear functional $\phi \in \mathsf{P}(\mathscr{A}) $ with the GNS triplet $(\mathcal{H}_\phi, \pi_\phi, \Omega_\phi)$. For any finite decomposition $\phi = \sum_i \phi_i$ with $\phi_i \in \mathsf{P}(\mathscr{A})$ there exist unique positive operators $D_i \in \pi_\phi(\mathscr{A})' \subseteq B(\mathcal{H}_\phi)$ that satisfy $\sum_i D_i = \mathds{1}_{\mathcal{H}_\phi}$, such that $\phi_i(\mathfrak{a})=  \bra{\Omega_\phi} D_i \pi_\phi(\mathfrak{a}) \ket{\Omega_\phi}$.
\end{lemma}

The equivalent of the GNS construction for completely positive maps is the Stinespring dilation theorem.
Before getting started with further generalisations, it is insightful to see its explicit proof, as it will be useful for understanding our main result. Here, we only provide the statement of the theorem and the main proof ideas; more details can be found in \cite{conway} and \cite{Rag03radon}.

\begin{theorem}[Stinespring dilation theorem]\label{thm:stinespring}
    Let $\mathscr{A}$ be a unital C$^*$-algebra, $\mathcal{H}$ a Hilbert space, and let $T$ be a completely positive map $T : \mathscr{A} \to B(\mathcal{H})$. Then there exists
\begin{enumerate}
    \item a Hilbert space $\mathcal{K}$, %
    \item a unital $^*$-homomorphism $\pi : \mathscr{A} \to B(\mathcal{K})$,
    \item a bounded operator $V : \mathcal{H} \to \mathcal{K}$, with $\|V\|^2 = \|T(\mathfrak{1})\|$,
\end{enumerate}
such that
\begin{equation*}
    T(\mathfrak{a}) = V^* \pi(\mathfrak{a}) V, \qquad \forall \mathfrak{a} \in\mathscr{A}.
\end{equation*}
If $T$ is unital, $V$ can be chosen to be an isometry, \emph{i.e.}, $V^*V = \mathds{1}_\mathcal{H}$.
We will refer to the tuple $(\mathcal{K},\pi,V)$ as the \emph{Stinespring dilation} of $T$.
The Stinespring dilation can always be chosen to be \emph{minimal}, \emph{i.e.}, such that $\mathcal{K} = \overline{\pi(\mathscr{A})V\mathcal{H}}$. Minimal Stinespring dilations are unique up to unitary equivalence.
\end{theorem}
\begin{proof}[Proof idea]
    In order to construct the Hilbert space $\mathcal{K}$, consider the algebraic tensor product $\mathscr{A} \odot \mathcal{H}$, consisting of all formal sums $\sum_j \mathfrak{a}_j \odot h_j$, with $\mathfrak{a}_j \in \mathscr{A}$ and $h_j \in \mathcal{H}$.
    Define the following sesquilinear form on $\mathscr{A} \odot \mathcal{H}$ by
    \begin{equation*}
        \left\langle \sum_j \mathfrak{a}_j \odot h_j, \sum_i \mathfrak{b}_i \odot g_i \right\rangle
        = \sum_{i,j} \left\langle T(\mathfrak{b}_i^* \mathfrak{a}_j)h_j, g_i \right\rangle
    \end{equation*}
    which is positive semi-definite because the map $T$ is completely positive.
    However it is not an inner product because it is not positive definite, meaning that it might take zero value also for some non-null element.
    
    Let $\mathcal{L}= \{ v \in \mathscr{A}\odot\mathcal{H} | \langle v |v\rangle = 0\}$ be a linear subspace in $\mathscr{A}\odot\mathcal{H}$.
    Set then $\mathcal{K}_0 = (\mathscr{A}\odot\mathcal{H})/\mathcal{L}$; on this space we now have a well-defined inner product given by
    \begin{equation*}
        \langle u + \mathcal{L},v + \mathcal{L}\rangle = \langle u,v\rangle .
    \end{equation*}
    Let $\mathcal{K}$ be the completion of $\mathcal{K}_0$ with respect to the norm induced by this inner product.

    Now, define the map $\pi:\mathscr{A} \to \mathsf{B}(\mathcal{K})$ by
    \begin{equation*}
        \pi(\mathfrak{a}) \left[ \sum_j \mathfrak{a}_j \odot h_j + \mathcal{L}\right] = \sum_j \mathfrak{a}\cdot \mathfrak{a}_j \odot h_j + \mathcal{L}
    \end{equation*}
    which is a well defined bounded operator on $\mathcal{K}_0$, and extend it to a bounded operator on $\mathcal{K}$.
    It is not difficult to see that $\pi$ is a valid $*$-homomorphism.

    Finally, define the operator $V:\mathcal{H} \to \mathcal{K}$ as
    \begin{equation*}
        V h = \mathfrak{1}\odot h + \mathcal{L}.
    \end{equation*}
    Note that $\| V h \|^2 = \langle \mathfrak{1}\odot h, \mathfrak{1}\odot h\rangle = \langle T(\mathfrak{1})h, h\rangle \leq \|T(\mathfrak{1})\| \cdot \| h \|^2$, thus $V$ is a bounded operator. Furthermore $V$ is an isometry if and only if $T$ is unital, in which case equality holds.
    It is then straightforward to verify that for all $\mathfrak{a}\in\mathscr{A}$
    \begin{align*}
        T(\mathfrak{a}) = V^* \pi(\mathfrak{a}) V,
    \end{align*}
    and that the thus constructed dilation is indeed minimal.
    These techniques were originally developed for the GNS construction, where the same problem occurs, and then extended to completely positive maps.
\end{proof}

What was shown before for positive linear functionals can now be extended to completely positive maps.
The following Theorem~\ref{thm:rn-for-cp-maps} can be understood as a very general way of jointly purifying quantum channels and instruments.

\begin{theorem}[Radon-Nikodym Theorem for completely positive maps~\cite{BS86radon,Rag03radon}]\label{thm:rn-for-cp-maps} %
    Consider two completely positive maps $S,R \in \mathsf{CP}(\mathscr{A},\mathcal{H})$ and let $(\mathcal{H}_R, \pi_R, V_R)$ be a minimal Stinespring dilation of $R$. Then $S \leq R$ if and only if there exists a unique self-adjoint operator $D \in \pi_R(\mathscr{A})' \subseteq \mathsf{B}(\mathcal{H}_R)$, with $0 \leq D \leq \mathds{1}$, such that
    \begin{equation*}
        S(\mathfrak{a}) = V_R^* \pi_R(\mathfrak{a}) D V_R = V_R^* D \pi_R(\mathfrak{a}) V_R, \qquad \forall \mathfrak{a} \in \mathscr{A}.
    \end{equation*}
\end{theorem}

We are interested in the specific case of finite decompositions of CP maps, which physically represent quantum transformations; once again, every addend is going to be maximised by the sum, and can be represented by the same Stinespring dilation that does not depend on any label, at the cost of introducing a Radon-Nikodym derivative which does.

\begin{lemma}[From Theorem 3.3 of \cite{Rag03radon}]\label{lem:RN2-sums}
    Consider a completely positive map $R \in \mathsf{CP}(\mathscr{A},\mathcal{H}) $ with the minimal Stinespring dilation $(\mathcal{H}_R, \pi_R, V_R)$. For any finite decomposition $R = \sum_i R_i$ with $R_i \in \mathsf{CP}(\mathscr{A},\mathcal{H})$ there exist unique positive operators $D_i \in \pi_R(\mathscr{A})' \subseteq B(\mathcal{H}_R)$ that satisfy $\sum_i D_i = \mathds{1}_{\mathcal{H}_R}$, such that $R_i(\mathfrak{a})=  V_R^* D_i \pi_R(\mathfrak{a}) V_R$.
\end{lemma}

\subsection{The chain rule for the composition of two CP maps}%
\label{subsec:chain-rule-rn-2-maps}

A natural follow-up question, relevant especially in the field of quantum information, is the behaviour of Radon-Nikodym derivatives under composition of CP maps.
Knowing the Radon-Nikodym derivatives of two CP maps with respect to two dominant maps, can one directly relate a composition of the derivatives to the derivative of the composition of the two dominated maps with respect to the composition of the two dominant maps?

As a first attempt, let us sequentially apply Theorem \ref{th:rn} for the states and Theorem \ref{thm:rn-for-cp-maps} for the transformations 
\begin{equation*}
    \phi_{a|x}(T_{b|y}(\mathfrak{c}_{c|z})) = \bra{\Omega_\phi} A_{a|x} V_T^* B_{b|y} \pi_T(\mathfrak{c}_{c|z})V_T\ket{\Omega_\phi}.
\end{equation*}
The Radon-Nikodym derivatives $A_{a|x} \in \mathsf{B}(\mathcal{H}_\phi)$ and $B_{b|y} \in \mathsf{B}(\mathcal{H}_T)$ can be interpreted as POVMs, together with $\pi_T(\mathfrak{c}_{c|z})\in \mathsf{B}(\mathcal{H}_T)$.
We also have the following commutation relations between them 
\begin{align*}
    [B_{b|y}, \pi_T(\mathfrak{c}_{c|z})] =0, \qquad [A_{a|x}, V_T^* B_{b|y} \pi_T(\mathfrak{c}_{c|z})V_T]=0.
\end{align*}
However, to prove that they all pairwise commute is not trivial. What makes this hard is that the POVMs live in different Hilbert spaces, connected by the isometry $V_T$.
Lemma~\ref{lem:extended-derivative} shows a way to lift operators from $\mathcal{H}_\phi$ to $\mathcal{H}_T$, while keeping desirable properties; this is possible by exploiting the very specific structure of the Stinespring isometry $V_T$.

Before proving the main result of this section, the new chain rule for Radon-Nikodym derivatives of CP maps, we require a technical result that allows us in certain situations to lift bounded operators through Stinespring dilations of CP maps.
This will later be used to lift Radon-Nikodym derivatives of previous CP maps through the Stinespring dilations of later CP maps in the chain.

\begin{lemma}[Lifting Lemma]\label{lem:extended-derivative}
    Let $(\mathcal{K}, \pi, V)$ be a Stinespring dilation of a completely positive map $T: \mathscr{A} \to \mathsf{B}(\mathcal{H})$.
    Then, there exists a *-homomorphism
    \begin{align*}
        \mathsf{B}(\mathcal{H}) \supseteq T(\mathscr{A})' \cap (T(\mathscr{A})^\ast)' &\to \pi(\mathscr{A})' \subseteq \mathsf{B}(\mathcal{K}), \\
        M &\mapsto \overline{M},
    \end{align*}
    lifting operators $M$ on $\mathcal{H}$ to operators $\overline{M}$ on $\mathcal{K}$ with the following properties.
    \begin{enumerate}
        \item The lift is consistent with the dilation: $V M = \overline{M} V$, and $M V^\ast = V^\ast \overline{M}$.
        \item The lift is isometric: $\| \overline{M} \| = \| M \|$.
        \item The lift is unital: $\overline{\mathds{1}_\mathcal{H}} = \mathds{1}_\mathcal{K}$.
        \item The lift is completely positive. Also, $\overline{M} \geq 0$ if and only if $M \geq 0$.
    \end{enumerate}
\end{lemma}
This Lemma captures the fact that the tensor-product structure in finite-dimensional dilations is really generalised with commuting operators in the infinite-dimensional case. As such, everything that lived in the original space $\mathcal{H}$ before the dilation, can be naturally lifted to an object (operation/state) in the new, larger space given by the dilation which commutes with all representations of the C*-algebra $\mathscr{A}$.

\begin{proof}[Proof of Lemma~\ref{lem:extended-derivative}]
Remember that the Hilbert space $\mathcal{K}$ is the closure of $\mathcal{K}_0 = (\mathscr{A} \odot \mathcal{H})/\mathcal{L}$.
Consider an operator $M \in T(\mathscr{A})' \cap (T(\mathscr{A})^\ast)' \subseteq \mathsf{B}(\mathcal{H})$. We lift it to an operator $\overline{M}$ on $\mathcal{K}$ by defining it on the dense subset $\mathcal{K}_0$ as
    \begin{align*}
        \overline{M} \left[ \sum_j \mathfrak{a}_j \odot h_j + \mathcal{L}\right] = \sum_j \mathfrak{a}_j \odot M h_j + \mathcal{L}
    \end{align*}
whose closure correctly gives us a bounded operator on $\mathcal{K}$.
For the operator $\overline{M}$ to be properly defined on the algebraic quotient $(\mathscr{A} \odot \mathcal{H})/\mathcal{L}$, we need to show that it leaves invariant the ideal $\mathcal{L}$, defined as
\begin{align}
    \mathcal{L} = \left.\left\{ v = \sum_j \mathfrak{a}_j \odot h_j \in \mathscr{A}\odot \mathcal{H} \;\right|\; \langle v, v \rangle = \sum_{j,k} \left\langle h_j, T(\mathfrak{a}_j^* \mathfrak{a}_k) h_k \right\rangle = 0 \right\}.
\end{align}
To this end, and retracing the steps of the original proof of Stinespring's dilation \cite{Stinespring}, we prove that
\begin{align}
    \left \langle \overline{M}v, \overline{M}v \right \rangle \leq \| M \|^2 \left\langle v , v \right\rangle, \quad \forall v \in \mathscr{A} \odot \mathcal{H}, \;\forall M,M^* \in T(\mathscr{A})' \subseteq \mathsf{B}(\mathcal{H}), \label{eq:extension_bounded}
\end{align}
which directly implies the invariance of $\mathcal{L}$.
In order to prove Equation~\ref{eq:extension_bounded} by contradiction, assume that it were false. Then, there exist $M, M^\ast \in T(\mathscr{A})'$ with $\|M\| < 1$ and $v \in \mathscr{A} \odot \mathcal{H}$ with $\langle v,v \rangle \leq 1$ such that
\begin{align*}
    \langle \overline{M}v, \overline{M}v \rangle \geq 1.
\end{align*}
By the commutativity of $M$ with $T(\mathscr{A})$, it follows that
\begin{align*}
    &1 \leq \langle \overline{M}v, \overline{M}v \rangle
    = \sum_{j,k} \left\langle M h_j, T(\mathfrak{a}_j^* \mathfrak{a}_k) M h_k \right\rangle
    = \sum_{j,k} \left\langle M h_j, M T(\mathfrak{a}_j^* \mathfrak{a}_k) h_k \right\rangle \\
    &= \sum_{j,k} \left\langle M^\ast M h_j, T(\mathfrak{a}_j^* \mathfrak{a}_k) h_k \right\rangle
    = \langle \overline{M^\ast M} v, v \rangle
    \leq \langle \overline{M^\ast M} v, \overline{M^\ast M} v \rangle,
\end{align*}
using the Cauchy-Schwarz inequality in the last step. Making use of the additional commutativity of $M^\ast$ with $T(\mathscr{A})$, we iteratively deduce in the same manner that
\begin{align*}
    \langle \overline{(M^\ast M)^{2^n}} v, v \rangle \geq 1,
\end{align*}
for all $n\in\mathbb{N}$. This however forms a contradiction with the assumption that $\|M\| < 1$ which by itself implies that
\begin{align*}
    \langle \overline{(M^\ast M)^{2^n}} v, v \rangle \to 0, \text{ for } n \to \infty,
\end{align*}
and thus concludes the proof of Equation~\ref{eq:extension_bounded}. The operator $\overline{M}$ is thus well-defined on the Hilbert space $\mathcal{K}$.
The calculation
\begin{align*}
    \pi(\mathfrak{a}) \overline{M} \left[ \sum_j \mathfrak{a}_j \odot h_j + \mathcal{L}\right] 
    &= \pi(\mathfrak{a}) \left[ \sum_j \mathfrak{a}_j \odot M h_j + \mathcal{L}\right] \\
    &= \left(\mathfrak{a} \cdot \sum_j \mathfrak{a}_j \right)\odot M h_j + \mathcal{L}\\
    &= \overline{M} \left[ \left(\mathfrak{a} \cdot \sum_j \mathfrak{a}_j \right)\odot h_j + \mathcal{L}\right] \\
    &= \overline{M} \pi(\mathfrak{a}) \left[ \sum_j \mathfrak{a}_j \odot h_j + \mathcal{L}\right]
\end{align*}
implies by the density of $\mathcal{K}_0$ in $\mathcal{K}$ that $[\overline{M}, \pi(\mathfrak{a})]=0, \forall \mathfrak{a}\in \mathscr{A}$. Hence, $\overline{M} \in \pi(\mathscr{A})'$ as required. Straightforward calculations show that
\begin{align*}
    \overline{\lambda (M_1 + M_2)} = \lambda (\overline{M_1} + \overline{M_2}), \quad \text{ and } \quad \overline{M_1 \cdot M_2} = \overline{M_1} \cdot \overline{M_2}
\end{align*}
for all $M_1,M_2 \in T(\mathscr{A})' \cap (T(\mathscr{A})^\ast)'$ and all $\lambda \in \mathbb{C}$, making the lifting map indeed an algebra homomorphism.

By the construction of $\overline{M}$, it follows for all $v = \sum_j \mathfrak{a}_j \odot h_j + \mathcal{L} \in \mathcal{K}_0$ and $w = \sum_k \mathfrak{b}_j \odot g_j + \mathcal{L} \in \mathcal{K}_0$ that
\begin{align*}
    \langle v, \overline{M} w \rangle
    &= \sum_{j,k} \left\langle h_j, T(\mathfrak{a}_j^* \mathfrak{b}_k) M g_k \right\rangle
    = \sum_{j,k} \left\langle h_j, M T(\mathfrak{a}_j^* \mathfrak{b}_k) g_k \right\rangle \\
    &= \sum_{j,k} \left\langle M^\ast h_j, T(\mathfrak{a}_j^* \mathfrak{b}_k) g_k \right\rangle
    = \langle \overline{M^\ast} v, w \rangle,
\end{align*}
where we used that $M$ commutes with $T(\mathscr{A})$. Therefore, by density of $\mathcal{K}_0$ in $\mathcal{K}$ it also follows that $(\overline{M})^\ast = \overline{M^*}$. This shows that the lifting map is even a *-homomorphism. As a direct consequence, $\overline{M}$ is self-adjoint if and only if $M$ is self-adjoint.

All other required properties of the lift are easily shown.
\begin{enumerate}
    \item We want to prove $V M = \overline{M} V$. Consider their action on $\mathcal{H}$ as follows
    \begin{align*}
        &V M \left[ \sum_j h_j \right] = V \left[ \sum_j M h_j \right] = \mathfrak{1} \odot \sum_j M h_j, \\
        &\overline{M} V \left[ \sum_j h_j \right] = \overline{M}\left[ \mathfrak{1} \odot \sum_j h_j \right] = \mathfrak{1} \odot \sum_j M h_j,
    \end{align*}
    and the same holds for the adjoint version of this statement.
    
    \item Equation~\ref{eq:extension_bounded} implies that $\| \overline{M} \| \leq \| M \|$. Since $\overline{M}$ is a lift of $M$, it follows that $\| \overline{M} \| = \| M \|$.        

    \item Unitality is obvious by the construction of $\overline{M}$.

    \item Complete positivity of the lifting map follows from the fact that all *-homomorphisms are completely positive.
\end{enumerate}
\end{proof}

The lifting also preserves the POVM structure.
\begin{corollary}\label{cor:extension_povm}
    Under the assumptions of Lemma~\ref{lem:extended-derivative}, if $\{M_{a|x}\}$ is a set of POVMs labeled by $x$ and with outcomes $a$, then so is $\{\overline{M_{a|x}}\}$.
\end{corollary}
\begin{proof}
    Since the lifting map is a $^*$-homomorphism (and thus completely positive), and all $M_{a|x}$ are Hermitian, positive operators, it also holds that
    all $\overline{M_{a|x}}$ are Hermitian, positive operators.
    It further follows for all labels $x$ that
    \begin{align*}
        \sum_a \overline{M_{a|x}} = \overline{\sum_a M_{a|x}} = \overline{\mathds{1}_\mathcal{H}} = \mathds{1}_\mathcal{K},
    \end{align*}
    which shows that $\{\overline{M_{a|x}}\}$ describes POVMs on the Hilbert space $\mathcal{K}$.
\end{proof}

Building on top of this Lifting Lemma \ref{lem:extended-derivative}, we can prove a generalisation of the Radon-Nikodym Theorem for CP maps which looks like a chain rule, in the sense that standardises sequentially applications of Radon-Nikodym theorem.
To our knowledge, this theorem is not known in the literature; considering the centrality of the Radon-Nikodym Theorem for many fields in mathematics and physics, it is of independent interest.
To favour readability, we first present a theorem that considers only two sequential applications.

\begin{theorem}[Chain rule for Radon-Nikodym derivatives] \label{th:chain-2-RN}
Consider unital C$^*$-algebras $\mathscr{A}_2$, $\mathscr{A}_1$, and $\mathscr{A}_0 = \mathsf{B}(\mathcal{H})$.
Consider completely positive maps $S^{(i)}_{x_i}, R^{(i)} : \mathscr{A}_i \to \mathscr{A}_{i-1}$ for $i=1,2$, such that $R^{(i)}$ is unital and $S^{(i)}_{x_i} \leq R^{(i)}$ for all $i$ and $x_i$.

Then, there exists a Hilbert space $\mathcal{K}$, a $^*$-representation $\pi : \mathscr{A}_2 \to \mathsf{B}(\mathcal{K})$, an isometry $V : \mathcal{H} \to \mathcal{K}$, %
commuting self-adjoint operators $F_{x_1}, F_{x_2} \in \pi(\mathscr{A}_2)' \subseteq \mathsf{B}(\mathcal{K})$ such that $0 \leq F_{x_1},F_{x_2} \leq \mathds{1}_\mathcal{K}$ and
    \begin{align}
        &S^{(1)}_{x_1} \circ S^{(2)}_{x_2} (\mathfrak{a}_2) = V^* F_{x_1}F_{x_2} \pi(\mathfrak{a}_2) V, \qquad \forall \mathfrak{a}_2 \in \mathscr{A}_2, \label{eq:chainRN2} \\
        &[F_{x_1},F_{x_2}] = [F_{x_2}, \pi(\mathfrak{a}_2)]=  [F_{x_1},\pi(\mathfrak{a}_2)] = 0, \qquad \forall \mathfrak{a}_2 \in \mathscr{A}_2. \label{eq:chainRN2_comm}
    \end{align}
\end{theorem}

\begin{proof}%
Consider $(\mathcal{H}_1, \pi_1, V_1)$ a minimal Stinespring dilation of the map $R^{(1)}$.
Since $S^{(1)}_{x_1} \leq R^{(1)}$, we can apply the Radon-Nikodym Theorem for CP maps (theorem~\ref{thm:rn-for-cp-maps}) to obtain the derivative $D_{x_1}^{(1)} \in \pi_1(\mathscr{A}_1)' \subseteq \mathsf{B}(\mathcal{H}_1)$,  satisfying 
\begin{align*}
    S^{(1)}_{x_1} (\mathfrak{a}_1) 
    &= V_{1}^* D_{x_1}^{(1)} \left[ \pi_{1} (\mathfrak{a}_1) \right] V_{1}, \qquad \forall \mathfrak{a}_1 \in \mathscr{A}_1,
\end{align*}
and thus in particular, setting $\mathfrak{a}_1 = S^{(2)}_{x_2}(\mathfrak{a}_2)$,
\begin{align}\label{eq:chainRN-step1}
    S^{(1)}_{x_1} \circ S^{(2)}_{x_2}(\mathfrak{a}_2) 
    &= V_{1}^* D_{x_1}^{(1)} \left[ \pi_{1} \circ S^{(2)}_{x_2}(\mathfrak{a}_2) \right] V_{1}, \qquad \forall \mathfrak{a}_2 \in \mathscr{A}_2.
\end{align}
Consider now the completely positive map $\pi_{1} \circ S^{(2)}_{x_2} \in \mathsf{CP}(\mathscr{A}_2, \mathcal{H}_1)$.
Note that $S^{(2)}_{x_2} \leq R^{(2)}$ implies that $\pi_1 \circ S^{(2)}_{x_2} \leq \pi_1 \circ R^{(2)}$.
Let hence $(\mathcal{H}_2, \pi_2, V_2)$ be a minimal Stinespring dilation of $\pi_1 \circ R^{(2)} \in \mathsf{CP}(\mathscr{A}_2, \mathcal{H}_1)$, and apply again the Radon-Nikodym Theorem for CP maps (theorem~\ref{thm:rn-for-cp-maps}), introducing another derivative $D_{x_2}^{(2)} \in \pi_2(\mathscr{A}_2)' \subseteq \mathsf{B}(\mathcal{H}_2)$ such that
    \begin{align}\label{eq:2RN-beforeext}
       S^{(1)}_{x_1} \circ S^{(2)}_{x_2}(\mathfrak{a}_2) = V_{1}^* D^{(1)}_{x_1} 
       \left[ V_{2}^* D^{(2)}_{x_2} \pi_{2}(\mathfrak{a}_2) V_{2} \right]
       V_{1}, \qquad \forall \mathfrak{a}_2 \in \mathscr{A}_2.
    \end{align}
Now we can use Lemma \ref{lem:extended-derivative} to construct the operator $\overline{D^{(1)}_{x_1}}$, which is the lift of $D^{(1)}_{x_1}$ from $\mathsf{B}(\mathcal{H}_1)$ to $\mathsf{B}(\mathcal{H}_2)$, such that
\begin{equation}\label{eq:2RN-afterext}
    S^{(1)}_{x_1} \circ S^{(2)}_{x_2}(\mathfrak{a}_2)  = V_{1}^* V_{2}^* \overline{D^{(1)}_{x_1}} D^{(2)}_{x_2} \pi_{2}(\mathfrak{a}_2) V_{2} V_{1}, \qquad \forall \mathfrak{a}_2 \in \mathscr{A}_2.
\end{equation}
Set the Hilbert space $\mathcal{K} = \mathcal{H}_2$, the representation $\pi = \pi_2$, the isometry $V = V_2 \circ V_1$ and the self-adjoint operators $F_{x_1} = \overline{D^{(1)}_{x_1}}, F_{x_2} = D^{(2)}_{x_2}$. These have exactly the desired form presented in Eq.~\ref{eq:chainRN2}.
It remains to prove that the commutation relations of Eq.~\ref{eq:chainRN2_comm} are also satisfied.
First note that by construction through theorem~\ref{thm:rn-for-cp-maps}, $D^{(2)}_{x_2}$ commutes with $\pi_2(\mathscr{A}_2)$.
By Lemma~\ref{lem:extended-derivative}, also $\overline{D^{(1)}_{x_1}}$ commutes with $\pi_2(\mathscr{A}_2)$.
The only thing left to prove is the commutation of $\overline{D^{(1)}_{x_1}}$ and $D^{(2)}_{x_2}$.
Consider now Eq. \ref{eq:2RN-beforeext} and focus on the relative order of $D^{(1)}_{x_1}, D^{(2)}_{x_2}$ and $\pi_2(\mathfrak{a}_2)$; at this stage we can use the commutation relations coming from both applications of the Radon-Nikodym Theorem for CP maps to have four equivalent ways of writing the same expression for all $\mathfrak{a}_2 \in \mathscr{A}_2$ :
\begin{align*}
    S^{(1)}_{x_1} \circ S^{(2)}_{x_2}(\mathfrak{a}_2) 
    &= V_{1}^* D^{(1)}_{x_1}  \left[ V_{2}^* D^{(2)}_{x_2} \pi_{2}(\mathfrak{a}_2) V_{2} \right]V_{1}
    = V_{1}^* D^{(1)}_{x_1} \left[ V_{2}^* \pi_{2}(\mathfrak{a}_2) D^{(2)}_{x_2} V_{2} \right] V_{1} \\
    &= V_{1}^* \left[ V_{2}^* D^{(2)}_{x_2} \pi_{2}(\mathfrak{a}_2) V_{2} \right] D^{(1)}_{x_1}V_{1} 
    = V_{1}^* \left[ V_{2}^* \pi_{2}(\mathfrak{a}_2) D^{(2)}_{x_2} V_{2} \right] D^{(1)}_{x_1}V_{1}.
\end{align*}
Note that $D^{(1)}_{x_1}$ is always on the left in the first line and on the right in the second line; hence respectively apply the first and the second relation given in property 1 of Lemma \ref{lem:extended-derivative} to lift $D^{(1)}_{x_1}$ to the operator $\overline{D^{(1)}_{x_1}}$ on $\mathcal{K}$.
Finally, using the fact that $\overline{D^{(1)}_{x_1}} \in \pi_2(\mathscr{A}_2)'$, we obtain all possible orderings of the three operators $M_1 = \overline{D^{(1)}_{x_1}}, M_2 = D^{(2)}_{x_2}$, and $M_3 = \pi_2(\mathfrak{a}_2)$:
\begin{align*}
    &S^{(1)}_{x_1} \circ S^{(2)}_{x_2}(\mathfrak{a}_2) 
    = V^* M_{\sigma(1)} M_{\sigma(2)} M_{\sigma(3)} V, \qquad \forall \mathfrak{a}_2 \in \mathscr{A}_2,
\end{align*}
for all permutations $\sigma \in S_3$.
To prove that this commutation relation holds on the full Hilbert space $\mathcal{K}$, and not only on the subspace $\overline{V \mathcal{H}}$, obtain the following equation for all $\mathfrak{a_2,b,c} \in \mathscr{A}_2$ and $\mathfrak{d,e} \in \mathscr{A}_1$ by applying the same reasoning:
\begin{align*}
    &S^{(1)}_{x_1} \circ \left( \mathfrak{d}^* \cdot S^{(2)}_{x_2}(\mathfrak{b}^* \cdot \mathfrak{a}_2 \cdot \mathfrak{c}) \cdot  \mathfrak{e} \right)
    = V_{1}^* \pi_1^*(\mathfrak{d}) D^{(1)}_{x_1} \left[ \pi_{1} \circ S^{(2)}_{x_2}(\mathfrak{b}^* \cdot \mathfrak{a}_2 \cdot \mathfrak{c}) \right]
       \pi_1(\mathfrak{e}) V_{1}\\
    &\qquad = \left( V_{1}^* \pi_1^*(\mathfrak{d}) \right)
    \left( V_{2}^* \pi_2^*(\mathfrak{b}) \right)
    M_{\sigma(1)} M_{\sigma(2)} M_{\sigma(3)}
    \left( \pi_2(\mathfrak{c}) V_{2}\right)
    \left( \pi_1(\mathfrak{e}) V_{1}\right),  \qquad \forall \mathfrak{a}_2 \in \mathscr{A}_2,
\end{align*}
for all permutations $\sigma \in S_3$,
where we used the usual commutation relations and the property of $^*$-homomorphisms that $\pi(a b) = \pi(a) \pi(b)$.
Using the minimality conditions of the Stinespring dilations, we know that $\mathcal{H}_i = \overline{\pi_i(\mathscr{A}_i)V_i\mathcal{H}_{i-1}}$, and more precisely $\mathcal{K} = \mathcal{H}_2 = \overline{\pi_2(\mathscr{A}_2)V_2\pi_1(\mathscr{A}_1)V_1\mathcal{H}}$.
It thus follows that
\begin{align*}
    M_{\sigma(1)} M_{\sigma(2)} M_{\sigma(3)} = M_{\sigma'(1)} M_{\sigma'(2)} M_{\sigma'(3)}
\end{align*}
for any permutations $\sigma, \sigma' \in S_3$, on the entire Hilbert space $\mathcal{K}$.
Letting $\mathfrak{a}_2 = \mathfrak{1}$ yields the final missing commutation relation between $\overline{D^{(1)}_{x_1}}$ and $D^{(2)}_{x_2}$, which concludes the proof.
\end{proof}

\begin{remark}
    For simplicity, Theorem~\ref{th:chain-2-RN} is stated for unital CP maps $R^{(i)}$. Its statement is still true in the case that the CP maps are not unital, since the Stinespring dilation (Theorem~\ref{thm:stinespring}) and the Radon-Nikodym derivative (Theorem~\ref{thm:rn-for-cp-maps}) do not require the CP maps to be unital. In this case, the only modification to the statement of Theorem~\ref{th:chain-2-RN} is that $V$ is a bounded operator with $\| V \|^2 \leq \| R^{(1)}(\mathfrak{1}_{\mathscr{A}_1})\| \cdot \| R^{(2)}(\mathfrak{1}_{\mathscr{A}_2})\|$, but not necessarily isometric.
\end{remark}

Once again,
we are interested in the case of different finite decompositions of the same CP maps.

\begin{lemma}\label{lemma:chain-2-rn-sums}
Consider unital C$^*$-algebras $\mathscr{A}_2$, $\mathscr{A}_1$, and $\mathscr{A}_0 = \mathsf{B}(\mathcal{H})$.
Consider completely positive unital maps $R^{(i)} : \mathscr{A}_i \to \mathscr{A}_{i-1}$ for $i=1,2$, and finite decompositions $R^{(i)} = \sum_{k_i} R^{(i)}_{k_i}$ such that $R^{(i)}_{k_i} \in \mathsf{CP}(\mathscr{A}_i,\mathscr{A}_{i-1})$.

Then, there exists a Hilbert space $\mathcal{K}$, a representation $\pi : \mathscr{A}_2 \to \mathsf{B}(\mathcal{K})$, an isometry $V : \mathcal{H} \to \mathcal{K}$,
and self-adjoint, positive operators $F^{(i)}_{k_i} \in \pi(\mathscr{A}_2)' \subseteq \mathsf{B}(\mathcal{K})$ such that $\sum_{k_i} F^{(i)}_{k_i} = \mathds{1}_\mathcal{K}$ and
    \begin{align*}
        &R^{(1)}_{k_1} \circ R^{(2)}_{k_2} (\mathfrak{a}_2) = V^* F^{(1)}_{k_1} F^{(2)}_{k_2} \pi(\mathfrak{a}_2) V, \qquad \forall k_1, k_2, \forall \mathfrak{a}_2 \in \mathscr{A}_2, \\
        &[F^{(1)}_{k_1},F^{(2)}_{k_2}] = [F^{(2)}_{k_2}, \pi(\mathfrak{a}_2)]=  [F^{(1)}_{k_1},\pi(\mathfrak{a}_2)] = 0, \qquad \forall k_1, k_2, \forall \mathfrak{a}_2 \in \mathscr{A}_2.
    \end{align*}
\end{lemma}
\begin{proof}
    Apply Theorem~\ref{th:chain-2-RN} to obtain the operators $F^{(i)}_{k_i}$. The only statement which is not immediately implied by theorem~\ref{th:chain-2-RN} is that $\sum_{k_i} F^{(i)}_{k_i} = \mathds{1}_\mathcal{K}$ for $i=1,2$.
    This is however a direct consequence of the fact that $F^{(i)}_{k_i}$ can be constructed through Lemma~\ref{lem:RN2-sums} and Lemma~\ref{lem:extended-derivative} together with the observation from Corollary~\ref{cor:extension_povm}.
\end{proof}

\begin{remark}\label{remark:rn-2-chain-functional}
    Theorem~\ref{th:chain-2-RN} and Lemma~\ref{lemma:chain-2-rn-sums} can be considered for the special case that $R^{(1)}$ is a functional, \emph{i.e.}, $\mathcal{H} = \mathbb{C}$ and thus $\mathscr{A}_0 = \mathbb{C}$. In this case, the isometry $V$ simply becomes a vector $\ket{\Omega} \in \mathcal{K}$ with $\|\ket{\Omega}\| = 1$.
\end{remark}

We will conclude this section with a natural generalisation of the chain rule for Radon-Nikodym derivatives.
This differs from Theorem~\ref{th:chain-2-RN} only by a weaker constraint on the maps: instead of requiring that the maps $S_{x_2}^{(2)}$ are all bounded by the same map, we require it for their representations under $\pi_1$.

\begin{theorem}[Generalised chain rule for Radon-Nikodym derivatives] \label{th:gen-chainRN-2}
    Consider unital C$^*$-algebras $\mathscr{A}_2$, $\mathscr{A}_1$, and $\mathscr{A}_0 = \mathsf{B}(\mathcal{H})$.
    For $i=1,2$, consider CP maps $S^{(i)}_{x_i} : \mathscr{A}_i \to \mathscr{A}_{i-1}$, and unital CP maps $R^{(i)}: \mathscr{A}_i \to \mathsf{B}(\mathcal{H}_{i-1})$, with $\mathcal{H}_0 = \mathcal{H}$ and $(\mathcal{H}_1, \pi_1, V_1)$ the minimal Stinespring dilation of $R^{(1)}$.

Consider the following conditions:
\begin{align}
    & S^{(1)}_{x_1} \leq R^{(1)} &\forall x_1,\\
    & \pi_1 \circ S^{(2)}_{x_2} \leq R^{(2)} &\forall x_2. \label{eq:gen-chainRN-2}
\end{align}

Then, there exists a Hilbert space $\mathcal{K}$, a $^*$-representation  $\pi : \mathscr{A}_2 \to \mathsf{B}(\mathcal{K})$, an isometry $V : \mathcal{H} \to \mathcal{K}$, %
commuting self-adjoint operators $F_{x_1}, F_{x_2} \in \pi(\mathscr{A}_2)' \subseteq \mathsf{B}(\mathcal{K})$ such that $0 \leq F_{x_1},F_{x_2} \leq \mathds{1}_\mathcal{K}$ and
    \begin{align*}
        &S^{(1)}_{x_1} \circ S^{(2)}_{x_2} (\mathfrak{a}_2) = V^* F_{x_1}F_{x_2} \pi(\mathfrak{a}_2) V, \qquad \forall \mathfrak{a}_2 \in \mathscr{A}_2, \\
        &[F_{x_1},F_{x_2}] = [F_{x_2}, \pi(\mathfrak{a}_2)]=  [F_{x_1},\pi(\mathfrak{a}_2)] = 0, \qquad \forall \mathfrak{a}_2 \in \mathscr{A}_2.
    \end{align*}
\end{theorem}
\begin{proof}
    Consider the proof of Theorem \ref{th:chain-2-RN} until Eq. \ref{eq:chainRN-step1}.
    While before we used the fact that $S^{(2)}_{x_2} \leq R^{(2)}$ implies that $\pi_1 \circ S^{(2)}_{x_2} \leq \pi_1 \circ R^{(2)}$, here we can directly use the constraint of Eq. \ref{eq:gen-chainRN-2}.
    Then the proof follows exactly as in Theorem \ref{th:chain-2-RN}, except that $(\mathcal{H}_2, \pi_2, V_2)$ is a minimal Stinespring dilation of $R^{(2)}$ directly.
\end{proof}

\subsection{The general chain rule}%
\label{subsec:chain-rule-rn-general}

All the proofs presented before can be extended from the composition of 2 CP maps to the composition of an arbitrary finite number of CP maps by induction.

\begin{theorem}[Chain rule of Radon-Nikodym for completely positive maps, $\mathsf{k}$ players]\label{th:chained-rn-k}
Let $\mathsf{k} \in \mathbb{N}$.
Consider unital C$^*$-algebras $\mathscr{A}_i$ for $i \in [\mathsf{k}-1]$, and $\mathscr{A}_0 = \mathsf{B}(\mathcal{H})$.
Consider completely positive maps $S^{(i)}_{x_i}, R^{(i)} : \mathscr{A}_i \to \mathscr{A}_{i-1}$, such that $R^{(i)}$ is unital and $S^{(i)}_{x_i} \leq R^{(i)}$ for all $i$ and $x_i$.

Then, there exists a Hilbert space $\mathcal{K}$, a $^*$-representation $\pi : \mathscr{A}_{\mathsf{k}-1} \to \mathsf{B}(\mathcal{K})$, an isometry $V : \mathcal{H} \to \mathcal{K}$, %
commuting self-adjoint operators $F_{x_i} \in \pi(\mathscr{A}_{\mathsf{k}-1})' \subseteq \mathsf{B}(\mathcal{K})$ such that $0 \leq F_{x_i} \leq \mathds{1}_\mathcal{K}$ and $\forall \mathfrak{a}_{\mathsf{k}-1} \in \mathscr{A}_{\mathsf{k}-1}$
    \begin{align}
        &S^{(1)}_{x_1} \circ \dots \circ S^{(\mathsf{k}-1)}_{x_{\mathsf{k}-1}} (\mathfrak{a}_{\mathsf{k}-1}) = V^* F_{x_1}\dots F_{x_{\mathsf{k}-1}} \pi(\mathfrak{a}_{\mathsf{k}-1}) V, & \\
        &[F_{x_i},F_{x_j}] = 0 &\forall i, j \in [\mathsf{k}-1], i \neq j, \\
        &[F_{x_i}, \pi(\mathfrak{a}_{\mathsf{k}-1})]= 0 &\forall i \in [\mathsf{k}-1].
    \end{align}
    
\end{theorem}
\begin{proof}
    We already proved the theorem for $\mathsf{k}=3$ (Theorem \ref{th:chain-2-RN}). 
    Now the induction step: suppose that the theorem is true for $\mathsf{k}$, we need to show that it is also true for $\mathsf{k}+1$.
    The statement for $\mathsf{k}$ players is the following decomposition :
    \begin{equation*}
        S^{(1)}_{x_1} \circ \dots \circ S^{(\mathsf{k}-1)}_{x_{\mathsf{k}-1}} (\mathfrak{a}_{\mathsf{k}-1}) = V^* F_{x_1}\dots F_{x_{\mathsf{k}-1}} \pi(\mathfrak{a}_{\mathsf{k}-1}) V
    \end{equation*}
    where all operators pairwise commute.
    To consider an extra player, we need to consider an additional algebra $\mathscr{A}_\mathsf{k}$, and maps $S^{(\mathsf{k})}_{x_{\mathsf{k}}}, R^{(\mathsf{k})} : \mathscr{A}_\mathsf{k} \to\mathscr{A}_{\mathsf{k}-1}$ such that $S^{(\mathsf{k})}_{x_{\mathsf{k}}} \leq R^{(\mathsf{k})}$.
    Consider now the $\mathsf{k}+1$-partite strategy, defined for all $\mathfrak{a}_{\mathsf{k}} \in \mathscr{A}_{\mathsf{k}}$
        \begin{equation*}
        S^{(1)}_{x_1} \circ \dots \circ S^{(\mathsf{k}-1)}_{x_{\mathsf{k}-1}}\circ S^{(\mathsf{k})}_{x_{\mathsf{k}}} (\mathfrak{a}_{\mathsf{k}}).
    \end{equation*}
We know how to decompose the action of the first $\mathsf{k}-1$ maps, obtaining pairwise commuting operators
    \begin{align*}
        S^{(1)}_{x_1} \circ \dots \circ S^{(\mathsf{k}-1)}_{x_{\mathsf{k}-1}} \left( S^{(\mathsf{k})}_{x_{\mathsf{k}}} (\mathfrak{a}_{\mathsf{k}}) \right)
        &= V^* F_{x_1}\dots F_{x_{\mathsf{k}-1}} \pi \left( S^{(\mathsf{k})}_{x_{\mathsf{k}}} (\mathfrak{a}_{\mathsf{k}}) \right) V\\
        &= V^* F_{x_1}\dots F_{x_{\mathsf{k}-1}} 
        \left[V_\mathsf{k}^* D^{(\mathsf{k})}_{x_{\mathsf{k}}}
        \pi_\mathsf{k} (\mathfrak{a}_{\mathsf{k}})
        V_\mathsf{k} \right]V,
    \end{align*}
    and in the second line we used the Radon-Nikodym Theorem for $\pi\circ S^{(\mathsf{k})}_{x_{\mathsf{k}}} \leq \pi \circ R^{(\mathsf{k})}$, where $(\mathcal{H}_\mathsf{k}, \pi_\mathsf{k},V_\mathsf{k})$ is the Stinespring dilation of $\pi \circ R^{(\mathsf{k})}$.

    The lifting lemma (Lemma~\ref{lem:extended-derivative}) allows to lift all of the operators $F_{x_i}$ to the new dilated space $\mathcal{H}_\mathsf{k}$:
    \begin{align*}
        S^{(1)}_{x_1} \circ \dots \circ S^{(\mathsf{k}-1)}_{x_{\mathsf{k}-1}}\circ S^{(\mathsf{k})}_{x_{\mathsf{k}}} (\mathfrak{a}_{\mathsf{k}}) 
        &= (V_\mathsf{k} V)^*\overline{F_{x_1}}\dots \overline{F_{x_{\mathsf{k}-1}}} D^{(\mathsf{k})}_{x_{\mathsf{k}}}
        \pi_\mathsf{k} (\mathfrak{a}_{\mathsf{k}})(V_\mathsf{k}V)\\
        &= W^* M_{x_1} \dots M_{x_\mathsf{k}} \pi'(\mathfrak{a}_\mathsf{k}) W,
    \end{align*}
    where in the second line we regrouped the isometries and renamed the operators, to have the standard form of the statement we want to show.
    Using exactly the same permutation argument in the end of Theorem \ref{th:chain-2-RN}, we can show all pairwise commutation relations:
    \begin{align*}
        &[M_{x_i},M_{x_j}] = 0 &\forall i \neq j \in [\mathsf{k}],\\
        &[M_{x_i}, \pi'(\mathfrak{a}_{\mathsf{k}})]= 0 &\forall i \in [\mathsf{k}].
    \end{align*}
    This concludes the proof by induction.
\end{proof}

\begin{theorem}[Generalised chain rule for Radon-Nikodym derivatives, $\mathsf{k}$ players]\label{th:gen-chained-rn-k}
Let $\mathsf{k} \in \mathbb{N}$.
For $i \in [\mathsf{k}-1]$, consider $\mathsf{k}-1$ unital C$^*$-algebras $\mathscr{A}_i$, and $\mathscr{A}_0 = \mathsf{B}(\mathcal{H})$.
Consider CP maps $S^{(i)}_{x_i} : \mathscr{A}_i \to \mathscr{A}_{i-1}$ and unital CP maps $ R^{(i)} : \mathscr{A}_i \to \mathsf{B}(\mathcal{H}_{i-1})$, with $\mathcal{H}_0 = \mathcal{H}$ and $(\mathcal{H}_i, \pi_i, V_i)$ the minimal Stinespring dilation of $R^{(i)}$.

Consider the following conditions:
\begin{align}
    & S^{(1)}_{x_1} \leq R^{(1)} &\forall x_1,\\
    & \pi_{i-1} \circ S^{(i)}_{x_i} \leq R^{(i)} &\forall i \in [2,\mathsf{k}-1], \forall x_i.
\end{align}

Then, there exists a Hilbert space $\mathcal{K}$, a $^*$-representation $\pi : \mathscr{A}_{\mathsf{k}-1} \to \mathsf{B}(\mathcal{K})$, an isometry $V : \mathcal{H} \to \mathcal{K}$, %
commuting self-adjoint operators $F_{x_i} \in \pi(\mathscr{A}_{\mathsf{k}-1})' \subseteq \mathsf{B}(\mathcal{K})$ such that $0 \leq F_{x_i} \leq \mathds{1}_\mathcal{K}$ and $\forall \mathfrak{a}_{\mathsf{k}-1} \in \mathscr{A}_{\mathsf{k}-1}$
    \begin{align}
        &S^{(1)}_{x_1} \circ \dots \circ S^{(\mathsf{k}-1)}_{x_{\mathsf{k}-1}} (\mathfrak{a}_{\mathsf{k}-1}) = V^* F_{x_1}\dots F_{x_{\mathsf{k}-1}} \pi(\mathfrak{a}_{\mathsf{k}-1}) V, & \\
        &[F_{x_i},F_{x_j}] = 0 &\forall i \neq j \in [\mathsf{k}-1],\\
        &[F_{x_i}, \pi(\mathfrak{a}_{\mathsf{k}-1})]= 0 &\forall i \in [\mathsf{k}-1].
    \end{align}
\end{theorem}
\begin{proof}
    It suffices to combine the proof of Theorem \ref{th:chained-rn-k} and the reasoning of Theorem \ref{th:gen-chainRN-2}.
\end{proof}
 
\newpage
\section{Compiled nonlocal games}\label{sec:compiled-non-local-games}

\subsection{Preliminaries on multipartite nonlocal games}\label{subsec:multi-nl}

Non-local games are a prevalent tool to study the intrinsic difference between classical and quantum information \cite{Brunner_2014}.
Bipartite games are the simplest and best understood setting, and it is often not trivial to extend proofs and techniques from two to more players.
In this paper, we contribute to deepen the understanding of multipartite quantum correlations.
Let us start by introducing the appropriate notation to characterise games with $\mathsf{k}$ players, labelled by $i \in [\mathsf{k}]$.
For every player $P_i$, we will consider a fixed finite set of input labels $\mathcal{I}_i$ and output labels $\mathcal{O}_i$. Often we will identify $P_i=(\mathcal{I}_i, \mathcal{O}_i)$ as the collection of input and output labels of the party $i$.

\begin{definition}[$\mathsf{k}$-players game, $\mathcal{G}$]\label{def:game}
A $\mathsf{k}$-players game is defined by the collection of all the input-output labels of each player $\{P_i\}_\mathsf{k}$, the sampling distribution of the inputs $q: \bigtimes_i \mathcal{I}_i \to \mathbb{R}_{+}$ and the game predicate (or winning condition) $V: \bigtimes_i(\mathcal{I}_i\times\mathcal{O}_i) \to \{0,1\}$.
In a compact way, the game is identified by the tuple $\mathcal{G}= (\{P_i\}_\mathsf{k}, q,V)$. We call a strategy or correlation of the game the joint probability of the outputs given the inputs
\begin{equation*}
    p(\Vec{a}|\Vec{x}) = p(a_1, \dots a_\mathsf{k}| x_1, \dots x_\mathsf{k}),
\end{equation*}
for which the score of the game can be computed
\begin{equation*}
        \omega(\mathcal{G},p) = \sum_i \sum_{x_i \in \mathcal{I}_i} \sum_{a_i \in \mathcal{O}_i} q(\Vec{x}) V(\Vec{a}| \Vec{x}) p(\Vec{a}|\Vec{x}).
\end{equation*} 
\end{definition}

We are generally interested in fixing a game $\mathcal{G}$, and studying the optimal score with strategies constrained to be in some set $p\in \mathcal{S} $
\begin{equation*}
    \omega_{\mathcal{S}}(\mathcal{G}) = \sup_{p\in\mathcal{S}} \omega(\mathcal{G},p).
\end{equation*}
Different sets of strategies can be characterised in many ways; some examples are the physical theory that we use to model the behaviour of the players, the way in which the inputs are distributed, and the computational power available to the players.
In the following we will characterise some of the sets which are relevant for this work.

A game is insightful if, considering different constraints, there is a gap between the optimal scores. Then, the score of the game can be used as a certification of these sets of correlations.
Non-local games are the first class of games developed to detect the difference between classical $\mathcal{C}$ and quantum strategies $\mathcal{Q}$. 
Early on researchers understood that a distinguishing feature of quantum mechanics is non-local effects.
The simplest way to detect these is to distribute a simple task between non-communicating players: quantum players can then use non-local effects (like entanglement) to have an advantage over their classical adversaries.
To enforce the no-communication, a non-local game is structured as follows.
At the beginning the players are allowed to communicate and agree on a shared strategy, but only until the verifier samples the questions $\Vec{x} =\{x_i\}_{[\mathsf{k}]}$.
Then the following interaction starts:
    \begin{enumerate}
        \item the verifier sends separately to each prover a question $x_i$;
        \item every prover sends back the answer $a_i$.
    \end{enumerate}

In theories respecting special relativity, a natural way to impose that the players are not communicating during the interaction with the verifier, is to model them as space-like separated agents. In this framework, no player can signal their input to someone else; tracing out their output, the information about their input must be completely lost.
This identifies the set of non-signalling correlations.
\begin{definition}[Non-signalling strategy, $p \in \mathcal{C}_{ns}^\mathsf{k}$]
   A strategy $p(\Vec{a}|\Vec{x})$ is said to be non-signalling if
    \begin{equation*}
        \sum_{a_i \in \mathcal{O}_i} 
        p(\Vec{a}_{[a_i]}|\Vec{x}_{[x_i]}) - \sum_{a_i' \in \mathcal{O}_i} p(\Vec{a}_{[a_i']}|\Vec{x}_{[x_i']}) = 0 \qquad \forall i\in [k],
    \end{equation*}
    where $\Vec{v}_{[v_i]}$ and $\Vec{v}_{[v_i']}$ are equal vectors, up to the $i$-th element that is respectively $v_i$ and $v_i'$.
    The optimal score of a game $\mathcal{G}$ with respect to the set of non-signalling correlations is $\omega_{ns}(\mathcal{G})$.
\end{definition}

If we add more structure and fix quantum theory as the underlying physical model, the no-communication constraint mathematically translates to modelling every agent in a different Hilbert space.

\begin{definition}[Quantum strategy, $p \in \mathcal{C}_{q}^\mathsf{k}$]%
    \label{def:nl-2}
    A quantum strategy is characterised by modelling the $\mathsf{k}$ players in different Hilbert spaces. It is specified by
    \begin{enumerate}
        \item a Hilbert space for each player $\left\{\mathcal{H}_n\right\}_{n\in [\mathsf{k}]}$, such that the global space is $\mathcal{H} = \otimes_{n=1}^\mathsf{k} \mathcal{H}_n$;
        \item POVMs for every player $\left\{ M^{(n)}_{a_n|x_n}\right\}_{n\in [\mathsf{k}]}$, respectively defined on $\mathcal{H}_n$;
        \item a shared state $\ket{\psi} \in \mathcal{H}$.
    \end{enumerate}
    These elements produce what we call a quantum correlation:
    \begin{equation*}
        p(\Vec{a}|\Vec{x}) = \bra{\psi} \bigotimes_{n=1}^\mathsf{k} M^{(n)}_{a_n|x_n}  \ket{\psi},
    \end{equation*}
    and the optimal score of a game $\mathcal{G}$ with respect to this set of correlations is $\omega_q(\mathcal{G})$.
\end{definition}

Another option to model space-like separated agents, more common in the field of quantum field theory, is to consider commuting operations on a single Hilbert space.
\begin{definition}[Quantum commuting operator strategy, $p \in \mathcal{C}_{qc}^\mathsf{k}$]\label{def:co-2}\label{def:co-k}
A quantum commuting operator strategy models the action of different players as commuting operations. More precisely, it is characterised by
        \begin{enumerate}
        \item a single Hilbert space $\mathcal{H}$;
        \item POVMs on $\mathcal{H}$ for every player $\left\{ M^{(n)}_{a_n|x_n}\right\}_{n=1}^\mathsf{k}$, s.t. $\left[ M^{(i)}_{a_i|x_i}, M^{(j)}_{a_j|x_j} \right] = 0$ $\forall i \neq j \in [\mathsf{k]}$;
        \item a shared state $\ket{\psi} \in \mathcal{H}$.
    \end{enumerate}
    These elements produce what we call a quantum commuting operator correlation :
    \begin{equation*}
        p(\Vec{a}|\Vec{x}) 
        = \bra{\psi} \prod_{n=1}^\mathsf{k} M^{(n)}_{a_n|x_n} \ket{\psi},
    \end{equation*}
    and the optimal score of a game $\mathcal{G}$ with respect to this set of correlations is $\omega_{qc}(\mathcal{G})$.
\end{definition}

It is trivial to show that both $\mathcal{C}_{q}^\mathsf{k}$ and $ \mathcal{C}_{qc}^\mathsf{k}$ are strictly included in $\mathcal{C}_{ns}^\mathsf{k}$, and there are games whose optimal scores show a separation.
The relation between $\mathcal{C}_{q}^\mathsf{k}$ and $ \mathcal{C}_{qc}^\mathsf{k}$ is more subtle. In finite dimensions the two sets coincide, but for many years it was an open question whether they were the same in infinite dimensions. This was finally settled in \cite{ji2022mipre}, where the authors showed that quantum correlations are a strict subset of commuting operator strategies.

\subsection{The KLVY compiler}%
\label{subsec:KLVYcompiler}

The space-like separation of players constitutes a big challenge for experiments.
Furthermore, the characterization of single devices and proofs of quantumness in the single-player setting remain important goals from a theoretical perspective.
The authors in \cite{KLVY22Quantum} address the issue by introducing a framework that relaxes the no-communication constraint.
The central idea is to simulate space-like separation with an encryption scheme, at the price of having to introduce computational assumptions.
They propose a compiler which is taking any $\mathsf{k}$-player non-local game, and compiles it into a single-prover interactive game.

In the compiled protocol, the single prover receives the inputs one-by-one, but cannot extract useful information about them, given that all but the last input are sent in an encrypted fashion.
Nevertheless, using homomorphic encryption, the single prover can still act on the encrypted data without having to decrypt them; in this way, the completeness of the compilation protocol is provided, which means that provers in the compiled game can perform at least as good as in the original non-local game.
On the other hand, the hope is that the security of the cryptography provides sufficient privacy of the encrypted data to prevent the prover in the compiled protocol to perform substantially better than non-communicating provers in the original game.

The introduction of computationally secure cryptography makes it necessary to impose bounds on the complexity of the prover, or \emph{adversary}, in the compiled game, as computationally unbounded provers would be able to break the security of the employed cryptography, thereby being able to extract the plaintext inputs and to circumvent the non-signalling restrictions that the cryptography was meant to impose.
In this way, unbounded adversaries would be able to create correlations across the several rounds of interaction that violate no-signalling and which are therefore impossible in the multi-player space-like separated setting.
For this reason, provers in the compiled protocol are assumed to be implementing only efficient operations in the following.
Roughly speaking, an operation is said to be efficient if it can be implemented with time and resources polynomial in the security parameter $\mathsf{poly}(\lambda)$.

\begin{definition}[Efficient algorithms]\label{def:efficient-algorithms}
    A deterministic polynomial time (PT) algorithm is a uniform family of classical circuits, of polynomial size with respect to the input size and the security parameter $\lambda$.
    A probabilistic polynomial time (PPT) algorithm is allowed to additionally access a polynomial amount of random coins.
    
    A quantum polynomial time (QPT) algorithm is a uniform family of quantum circuits of size polynomial in the input size and the security parameter $\lambda$.
\end{definition}

In the following we will summarise the framework of \cite{KLVY22Quantum} and outline the definitions that are relevant to this work. For a more complete description of the protocol and the QFHE scheme used, we refer to the original paper.

\begin{definition}[Quantum fully homomorphic encryption scheme, \cite{KLVY22Quantum} Def.~2.3]
    A quantum fully homomorphic encryption (QFHE) scheme with security parameter $\lambda \in \mathds{N}$ is a tuple of algorithms $(\gen, \enc, \eval, \dec)$ defined as follows:
    \begin{itemize}
        \item ${\gen}$ is a PPT algorithm that takes as input $1^\lambda$ and outputs a classical secret key ${\sk}$ of size $\mathrm{poly}(\lambda)$ bits;
        \item ${\enc}$ is a PPT algorithm that takes as input a secret key ${\sk}$ and a classical input $x$, and outputs a ciphertext ${\encx}$;
        \item ${\eval}$ is a QPT algorithm that homomorphically applies a quantum circuit on a quantum state and an encrypted input ${\encx}$, and returns an encrypted output $\enca$;
        \item ${\dec}$ is a PT algorithm that takes as input a secret key $\sk$ and a classical ciphertext ${\enca}$, and outputs a classical output $a$. 
    \end{itemize}
    Additionally, we also require the two following properties
    \begin{enumerate}
        \item correctness with auxiliary input, which means that $\eval$ behaves correctly with entanglement;
        \item IND-CPA security against a QPT adversary, which means that efficient adversaries without the secret key cannot decrypt messages.
    \end{enumerate}
\end{definition}
In the following, we will always omit that the algorithm $\enc$ and $\dec$, also have the secret key $\sk$ as an input.

Let us report the definition of the security of the scheme against QPT adversaries, which is particularly relevant for our work.
\begin{definition}
    [IND-CPA security against QPT adversary]\label{def:IND-CPA}
    Consider a QFHE scheme with security parameter $\lambda$, and a quantum efficient (QPT) adversary with black-box access to an encryption oracle $\mathcal{A}^{\enc(\cdot)}$.
    The scheme is said to have IND-CPA security if, for all two messages $m_0$ and $m_1$ chosen by the adversary, he
    cannot distinguish between $\enc(m_0)$ and $\enc(m_1)$ with more than a negligible probability in $\lambda$.
    More formally,
    \begin{equation*}
    \left|\Pr\left[\mathcal{A}^{\enc(\cdot)}(\mathsf{m}_0) = 1 \ \middle\vert
        \mathsf{m}_0 = \enc(m_0)
        \right]
    -\Pr\left[\mathcal{A}^{\enc(\cdot)}(\mathsf{m}_1) = 1 \ \middle\vert
        \mathsf{m}_1 = \enc(m_1)
        \right]\right|
    \le \textsf{negl}(\lambda).
    \end{equation*}
\end{definition}

We can now properly define the \cite{KLVY22Quantum} protocol, that compiles every $\mathsf{k}$-players non-local game into a computationally bounded single-prover game.

The key realisation of \cite{KLVY22Quantum} is that to preserve the classical soundness of the original non-local game, the time-order in which the inputs are distributed matters.
In other words, if all of the encrypted inputs are received at the same time, even an efficient classical prover can reach the non-signalling optimal score of the original non-local game \cite{KRR14}.
The \cite{KLVY22Quantum} compiler fixes a very specific time structure of the input distributions: an encrypted input is given if and only if an encrypted answer to the previous question is already committed. Fig. \ref{fig:compiler} shows a pictorial representation of the non-local game and its compiled version.
A full description of the compiler is given by Definition~\ref{def:compilation-protocol}.

\begin{figure}[ht]
  \centering
  \begin{minipage}{0.65\textwidth}
    \centering
\begin{tikzpicture}[scale=0.3]
    \fill[gray!30] (0,0) -- (24,0) -- (24,1) -- (0,1) -- cycle;
    \fill[gray!30] (0,0) -- (1,0) -- (1,10) -- (0,10) -- cycle;
    \fill[gray!30] (0,9) -- (24,9) -- (24,10) -- (0,10) -- cycle;

    \fill[orange!30] (4,3) rectangle (8,7);
    \fill[orange!30] (12,3) rectangle (16,7);
    \fill[orange!30] (20,3) rectangle (24,7);

    \draw[->, thick] (6,9) -- (6,7) node[right,midway] {$x$};
    \draw[->, thick] (14,9) -- (14,7)node[right,midway] {$y$};
    \draw[->, thick] (22,9) -- (22,7) node[right,midway] {$z$};
    
    \draw[->, thick] (6,3) -- (6,1)node[right,midway] {$a$};
    \draw[->, thick] (14,3) -- (14,1)node[right,midway] {$b$};
    \draw[->, thick] (22,3) -- (22,1)node[right,midway] {$c$};

    \fill[orange!30] (2,11) -- (20,11) -- (20,13) -- (2,13) -- cycle;
    \fill[orange!30] (2,3) -- (4,3) -- (4,11) -- (2,11) -- cycle;
    \fill[orange!30] (10,3) -- (12,3) -- (12,11) -- (10,11) -- cycle;
    \fill[orange!30] (18,3) -- (20,3) -- (20,11) -- (18,11) -- cycle;

    \node at (5,5){Alice};
    \node at (13,5) {Bob};
    \node at (21,5) {Charlie};

\end{tikzpicture}   \end{minipage}
  \hfill
  \begin{minipage}{0.3\textwidth}
    \centering

\begin{tikzpicture}[scale=0.3]
    \fill[gray!30] (8.5,0) -- (9.5,0) --(9.5,28) -- (8.5,28) -- cycle;
    \fill[gray!30] (3,0) -- (8.5,0) --(8.5,1) -- (3,1) -- cycle;
    \fill[gray!30] (3,9) -- (8.5,9) --(8.5,10) -- (3,10) -- cycle;
    \fill[gray!30] (3,18) -- (8.5,18) --(8.5,19) -- (3,19) -- cycle;
    \fill[gray!30] (3,27) -- (8.5,27) --(8.5,28) -- (3,28) -- cycle;

    \fill[orange!30] (2,3) rectangle (6,7);
    \fill[orange!30] (2,12) rectangle (6,16);
    \fill[orange!30] (2,21) rectangle (6,25);

    \fill[orange!30] (0,3) -- (2,3) -- (2,25) -- (0,25) -- cycle;

    \node at (3,5){Charlie};
    \node at (3,14) {Bob};
    \node at (3,23) {Alice};

    \draw[->, thick] (4,27) -- (4,25) node[right,midway] {$\mathsf{Enc}(x)$};
    \draw[->, thick] (4,18) -- (4,16)node[right,midway] {$\mathsf{Enc}(y)$};
    \draw[->, thick] (4,9) -- (4,7) node[right,midway] {$z$};
    
    \draw[->, thick] (4,21) -- (4,19)node[right,midway] {$\mathsf{Enc}(a)$};
    \draw[->, thick] (4,12) -- (4,10)node[right,midway] {$\mathsf{Enc}(b)$};
    \draw[->, thick] (4,3) -- (4,1)node[right,midway] {$c$};

    \draw[->, very thin] (12,28) -- (12,0)node[right,midway] {$t$};
    
\end{tikzpicture}
   \end{minipage}
    \caption{On the left, a graphical representation of a $3$-players non-local game. On the right, the game obtained through \cite{KLVY22Quantum} compilation. Time is flowing downwards. The players are represented in orange, and the verifier is in gray.}
    \label{fig:compiler}
\end{figure}
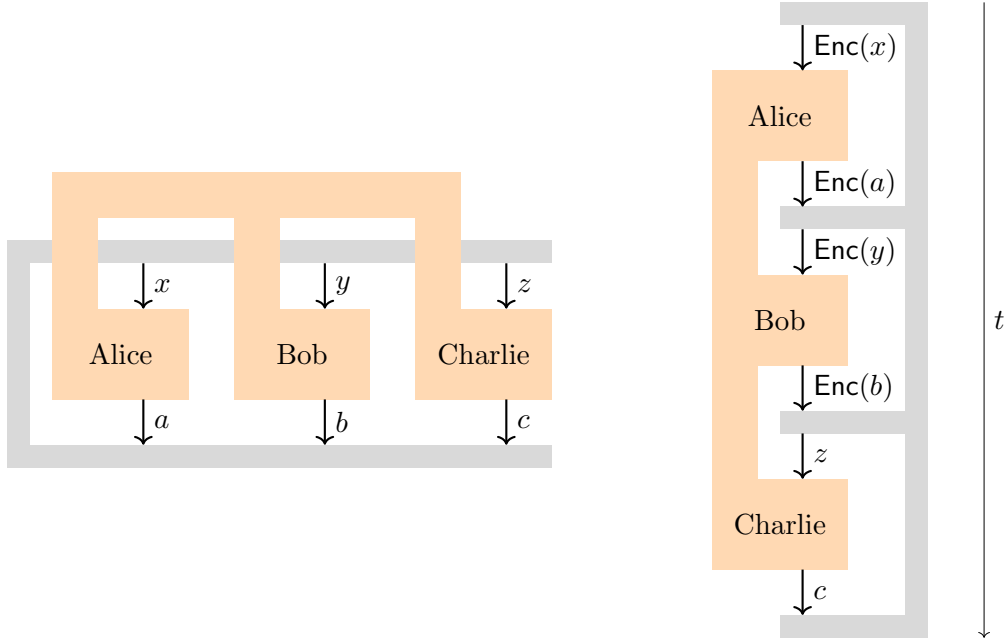

\begin{definition}[Compilation protocol of \cite{KLVY22Quantum}]\label{def:compilation-protocol}
    Consider a $\mathsf{k}$-player non-local game $\mathcal{G}$ and a QFHE scheme with security parameter $\lambda\in\mathbb{N}$.
    The compiled game $\mathcal{G}_\lambda$ is defined in the following way: first the verifier samples the questions $\{x_i\}_{[\mathsf{k}]}$; then, verifier and prover execute the following interactive protocol:
    \begin{enumerate}
        \item the verifier sends $\encx_1= \enc(x_1)$ to the prover;
        \item the prover sends back to the verifier an encrypted answer $\enca_1=\enc(a_1)$;
        \item repeat point 1. and 2. for $i\in[2,\mathsf{k}-1]$;
        \item the verifier sends $x_\mathsf{k}$ in the clear;
        \item the prover sends back $a_\mathsf{k}$ in the clear.
    \end{enumerate}
    The verifier then decrypts the (encrypted) answers of the prover, and evaluates the predicate of the non-local game $\mathcal{G}$ on the decrypted labels.
\end{definition}

\begin{definition}[Quantum strategy in the compiled game]\label{def:QPT-str-k-player}
Consider a compiled $\mathsf{k}$-player game $\mathcal{G}^\lambda$ according to Definition~\ref{def:compilation-protocol}.
An efficient quantum (QPT) strategy is then given by a a family $\{W_{i,\lambda}\}_{i\in[\mathsf{k}],\lambda\in\mathbb{N}}$ of QPT algorithms, that behave as follows:
    \begin{enumerate}
        \item on input of a ciphertext $\encx_1$, the prover applies $W_{1,\lambda}$ and obtains a classical output $\enca_1$ along with an internal (quantum) state. The prover sends $\enca_1$ to the verifier;
        \item for $i\in [2 ,\mathsf{k}-1]$, on input of a ciphertext $\encx_i$, the prover analogously applies the algorithm $W_{i,\lambda}$ to update its internal state and produce the classical output $\enca_i$, returning the latter to the verifier;
        \item on input of the plaintext $x_\mathsf{k}$, the prover applies the algorithm $W_{\mathsf{k},\lambda}$ to produce the classical output $a_\mathsf{k}$, that she sends back to the verifier. 
    \end{enumerate}
\end{definition}

\begin{definition}[Compiled correlations and score]\label{def:compiled-correlations-score}
    Given a $\mathsf{k}$-players compiled game $\mathcal{G}_\lambda$ and a prover strategy $\{W_{i,\lambda}\}_{i\in[\mathsf{k}],\lambda\in\mathbb{N}}$, if the correlations produced by the compiled game are given by
    \begin{equation*}
        p_\lambda(\enca_1,\dots, \enca_{\mathsf{k}-1},a_{\mathsf{k}}| \encx_1,\dots, \encx_{\mathsf{k}-1},x_{\mathsf{k}}),
    \end{equation*}
    then the associated decrypted correlations are
    \begin{align*}
        p_\lambda(\Vec{a}|\Vec{x}) =  
        &\mathop{\mathds{E}}_{\encx_1:\dec(\encx_1) = x_1} \sum_{\enca_1:\dec(\enca_1)=a_1} \cdots \\
        &\mathop{\mathds{E}}_{\encx_{\mathsf{k}-1}:\dec(\encx_{\mathsf{k}-1}) = x_{\mathsf{k}-1}} \sum_{\enca_{\mathsf{k}-1}:\dec(\enca_{\mathsf{k}-1})=a_{\mathsf{k}-1}}
        p_\lambda(\enca_1,\dots, \enca_{\mathsf{k}-1},a_{\mathsf{k}}| \encx_1,\dots, \encx_{\mathsf{k}-1},x_{\mathsf{k}}).
    \end{align*}
    The score of the game is then given by
    \begin{equation*}
        \omega(\mathcal{G}_\lambda, p_\lambda) = 
        \sum_i \sum_{x_i \in \mathcal{I}_i} \sum_{a_i \in \mathcal{O}_i} q(\Vec{x}) V(\Vec{a}| \Vec{x}) p_\lambda(\Vec{a}|\Vec{x}).
    \end{equation*}
\end{definition}

Similarly to before, we will fix a game and consider different set of correlations, to compare the optimal scores with respect to different constraints
\begin{equation*}
    \omega_\mathcal{S} = \sup_{p_\lambda \in \mathcal{S}} \omega(\mathcal{G}_\lambda, p_\lambda).
\end{equation*}

When considering compiled games, instead of considering correlations generated by arbitrary (unbounded) provers, it is more meaningful to consider the set of efficient correlations, \emph{i.e.}, correlations that can be generated by an efficient prover.
While for every fixed security parameter $\lambda$ one can define the corresponding compiled game, efficiency is an asymptotic property of prover's strategies in the security parameter $\lambda\in\mathbb{N}$.
Let us denote $\mathcal{G}_{comp}=\{\mathcal{G}_\lambda\}_\lambda$ for the sequence of compiled games for different security parameters $\lambda\in\mathbb{N}$ and $p_{comp} = \{p_\lambda\}_\lambda$ for efficient strategies in the compiled game for all $\lambda \in \mathbb{N}$.
To capture the asymptotic notion of efficiency, it is natural to consider the sequence of scores $\omega(\mathcal{G}_{comp}, p_{comp}) = \{ \omega(\mathcal{G}_\lambda,p_\lambda) \}_\lambda$.
One can then define the asymptotic score achieved by a certain strategy $p_{comp} = \{p_\lambda\}_\lambda$ as
\begin{equation*}
    \limsup_{\lambda \to \infty} \omega(\mathcal{G}_\lambda, p_\lambda),
\end{equation*}
and the optimal asymptotic core achievable by efficient strategies as
\begin{equation*}
    \omega_\mathcal{E}(\mathcal{G}_{comp}) = \sup_{\{p_\lambda\}_\lambda \in \mathcal{E}} \limsup_{\lambda \to \infty} \omega(\mathcal{G}_\lambda, p_\lambda),
\end{equation*}
where $\mathcal{E}$ denotes the set of all efficient strategies.

\subsubsection{Known completeness and soundness results}
This compiler preserves many desirable properties of the original non-local game.
Let us denote by $\omega_c(\mathcal{G}_{comp})$ ($\omega_q(\mathcal{G}_{comp})$) the optimal score that an efficient classical (quantum) single prover can achieve, following the compilation protocol of Def.~\ref{def:compilation-protocol}.
\begin{enumerate}
    \item \textit{Classical soundness for all games} \cite{KLVY22Quantum}.\\
    An efficient, classical prover can achieve \textbf{at most} a score in the compiled game which is negligibly close to the optimal classical value of the original non-local game, \emph{i.e.}, for every efficient classical strategy $\{p_\lambda\}_\lambda$ it holds that
    \begin{equation*}
        \omega(\mathcal{G}_{\lambda},p_\lambda) \leq \omega_c(\mathcal{G}) + \negl(\lambda).
    \end{equation*}
    Asymptotically, this implies that $\omega_c(\mathcal{G}_{comp}) \leq \omega_c(\mathcal{G})$.
    \item \textit{Quantum completeness for all games} \cite{KLVY22Quantum}.\\
    An efficient quantum prover can achieve \textbf{at least} a score in the compiled game which is negligibly close to the optimal quantum value of the original non-local game, \emph{i.e.}, there exists an efficient quantum strategy $\{p_\lambda\}_\lambda$ such that
    \begin{equation*}
        \omega(\mathcal{G}_{\lambda},p_\lambda) \geq \omega_q(\mathcal{G}) - \negl(\lambda).
    \end{equation*}
    Asymptotically, this implies that $\omega_q(\mathcal{G}_{comp}) \geq \omega_q(\mathcal{G})$.
    Intuitively, honest entangled provers for $\mathcal{G}$ can be simulated by a single quantum prover applying homomorphic operations, so no correlations are lost in the compilation.
    \item \textit{Quantum soundness for some bipartite games} \cite{NZ2023Bounding, CMMN2024Computational, baroni2024quantumboundscompiledxor, mehta2024selftestingcompiledsettingtiltedchsh}.\\
    For certain games, even a quantum prover cannot exceed the original quantum bound, as long as it remains computationally bounded. A celebrated example is the CHSH game: Natarajan and Zhang show that the compiled CHSH game inherits the Tsirelson bound~\cite{NZ2023Bounding}. More generally, Cui \emph{et al.}~\cite{CMMN2024Computational} and Baroni \emph{et al.}~\cite{baroni2024quantumboundscompiledxor} proved that any two-player XOR game has its Tsirelson bound preserved under compilation. Mehta, Paddock, and Wooltorton show an analogous result for tilted CHSH inequalities~\cite{mehta2024selftestingcompiledsettingtiltedchsh}.
    Asymptotically, for these specific games $\mathcal{G}$, we already know that $\omega_q(\mathcal{G}_{comp}) \leq \omega_q(\mathcal{G})$.
    \item \textit{Asymptotic quantum soundness for bipartite games} \cite{KMPSW24bound}.\\
    For all bipartite games, the optimal score achievable by computationally bounded quantum provers in the compiled game is asymptotically upper-bounded by the commuting-operator quantum value of the original nonlocal game in the limit of the security parameter tending towards infinity:
        \begin{equation*}
        \omega_q(\mathcal{G}_{comp}) =
        \sup_{\{p_\lambda\}_\lambda \in \mathcal{E}_q} \limsup_{\lambda \to \infty} \omega(\mathcal{G}_\lambda, p_\lambda)
        \leq \omega_{qc}(\mathcal{G}).
    \end{equation*}
\end{enumerate}

As opposed to the stronger concrete bounds that are known for the general classical soundness of the compiler, and for the quantum soundness when compiling games from specific classes, the asymptotic quantum soundness result which holds for all bipartite games does not allow for security statements in the finite-$\lambda$ regime.
Nonetheless, the very general result of \cite{KMPSW24bound} represents an extremely important step towards understanding the soundness of compiled games, as it confirms the intuition that the \cite{KLVY22Quantum} compiler asymptotically preserves quantum soundness for all bipartite games.
However, it only holds when restricted to at most two players.
In this work, we show that quantum soundness is indeed asymptotically preserved for all multipartite games.

\newpage
\section{Efficient quantum strategies for compiled games}\label{sec:quantum-compiler}

The correlations produced by a quantum computationally bounded single-prover are always almost-no-signalling
\begin{equation*}
    \left| \sum_{a_i \in \mathcal{O}_i} p_\lambda(\Vec{a}_{[a_i]}|\Vec{x}_{[x_i]}) - \sum_{a_i'\in \mathcal{O}_i} p_\lambda(\Vec{a}_{[a_i']}|\Vec{x}_{[x_i']})
    \right| \leq \textsf{negl}(\lambda),
\end{equation*}
and are perfectly non-signalling in the limit of $\lambda \to \infty$, because signalling in the single-prover framework means being able to decrypt messages.
We will show that imposing quantum physics as the underlying physical theory further restricts the set of admissible correlations.
More precisely, in this section we identify the additional constraints that naturally arise from the security properties of the encryption scheme on a QPT prover.

For pedagogical reasons, we often start considering a bipartite scenario involving Alice and Bob, and then extend it to a tripartite setting by introducing Charlie. The extension to more than three parties is then straightforward.

\subsection{Modelling efficient quantum strategies}\label{subsec:eff-quantum-str}

In this section we want to characterise the correlations that a quantum-polynomial-time (QPT) prover can produce in the compiled game, by making concrete the constraints following from the security of the encryption.
They will reflect the sequential structure of the compilation protocol, and be expressed in terms of efficient states, measurement and transformations.

Recall the definition of QPT algorithm of Def. \ref{def:efficient-algorithms}. More pragmatically, there are only three basic actions that an efficient quantum prover can do: 
\begin{itemize}
    \item creating $\ket{0}$ and appending them as ancillas;
    \item performing efficiently implementable unitary gates $U$;
    \item measuring in the computational basis $\Pi$.
\end{itemize}
Erasing information is also an efficient action; however this irreversible operation obviously cannot help the quantum prover, hence these operations can be neglected without loss of generality.

Thus, states $\rho$ that can be created by a QPT algorithm can be modelled as the preparation of (at most polynomially) many $\ket{0}$ states, followed by an efficient unitary gate $U$. This is always a pure state.

An efficient measurement can always be thought of as a projective measurement in an efficient basis, which means an efficiently implementable unitary $U$ followed by a measurement in the computational basis $\Pi_a$. This is always a projective measurement $M_a= U^* \Pi_a U$, indeed
\begin{equation*}
    M_a M_{a'}= U^* \Pi_a U U^* \Pi_{a'} U= U^* \Pi_a \Pi_{a'} U = \delta_{a,a'} M_a.
\end{equation*}
This fact was already pointed out in \cite{NZ2023Bounding}.

In a sequential setting, it is important to consider not only the classical outcome of the measurement, but also the post-measurement quantum state.
This is properly characterised by quantum instruments \cite{taranto2025higherorderquantumoperations}, a collection of completely positive (CP) maps $\Tilde{I}_a \in \mathsf{CP}(\mathcal{H}_i,\mathcal{H}_o)$, such that the sum over the classical output $a$ is also trace preserving (TP), $\sum_a \Tilde{I}_a = \Tilde{I} \in \mathsf{CPTP}(\mathcal{H}_i,\mathcal{H}_o)$.
Start with a state $\rho$ and apply the quantum instrument, the non-normalised post-measurement state is $\rho_a=\Tilde{I}_a(\rho)$, and the outcome $a$ occurs with probability $p(a|\rho,\Tilde{I}_a)= \tr(\rho_a)$.
An efficient quantum instrument can be modelled without loss of generality as follows: the prover receives a state $\rho$ and can append ancillas, applies an efficient unitary $U$, then decides to measure some of the registers in the computational basis $\Pi_a$, and finally applies a unitary $V_a$ on the remaining registers, conditioned to the outcome of the measurement. This defines a set of maps $\{\Tilde{I}_a\}_a$
\begin{align}\label{eq:instrument}
    \Tilde{I}_a(\rho) = \Tr_1\left[(\Pi_a \otimes V_a) U (\rho \otimes \ket{0}\bra{0}) U^*  (\Pi_a \otimes V_a^*)\right]
\end{align}
which are completely positive by construction.
The adjoint of this map, taking measurements to measurements, has the following structure:
\begin{align}\label{eq:adj-insrument}
    \Tilde{I}_a^*(M) = U (\Pi_a \otimes V_a M V_a^*) U^*.
\end{align}
Interestingly, these maps are not only CP, but also preserve multiplicative structure, \emph{i.e.}, they are algebraic homomorphisms. In some sense, we can see this as a purification of the adjoint of the map given by the quantum instrument.

\begin{lemma}\label{lem:homomorphism}
    Consider a quantum instrument $\Tilde{I}_a$ of the form of Eq.~\ref{eq:instrument}. Its adjoint is a *-homomorphism, \emph{i.e.}, it preserves multiplicative structure:
    \begin{equation*}
        \Tilde{I}_a^*(M_1^* M_2) = \left(\Tilde{I}_a^*(M_1)\right)^* \Tilde{I}_a^*(M_2).
    \end{equation*}
\end{lemma}
\begin{proof}
A simple calculation proves the claim:
\begin{align*}
    \left(\Tilde{I}_a^*(M_1)\right)^* I_a^*(M_2)
    &= U (\Pi_a \otimes V_a M_1^* V_a^*) U^* U (\Pi_a \otimes V_a M_2 V_a^*) U^*\\
    &= U (\Pi_a\Pi_a \otimes V_a M_1^* V_a^*V_a M_2 V_a^*) U^*\\
    &=  U (\Pi_a\otimes V_a M_1^* M_2 V_a^*) U^* = \Tilde{I}_a^*(M_1^* M_2).
\end{align*}
Note that we are dealing with two different types of adjoints in the previous equations. While $\Tilde{I}_a^*$ denotes the adjoint of the map $\Tilde{I}_a$ as a CP map between the spaces $\mathsf{B}(\mathcal{H}_i)$ and $\mathsf{B}(\mathcal{H}_o)$, $M^*$ denotes the adjoint of $M$ as an operator in $\mathsf{B}(\mathcal{H}_o)$.
\end{proof}

\begin{figure}[ht]
    \centering

    \begin{minipage}{0.2\textwidth}
        \centering
\begin{quantikz}
\lstick{$\ket{0}$} &\gate{U}  &\rstick{$\ket{\psi}$}
\end{quantikz}     \end{minipage}
    \hspace{0.05\textwidth}
    \begin{minipage}{0.2\textwidth}
        \centering
\begin{quantikz}
\lstick{$\rho$}  &\gate{U}  &\meter{a} &\setwiretype{n}
\end{quantikz}     \end{minipage}
    \hspace{0.05\textwidth}
    \begin{minipage}{0.4\textwidth}
        \centering
\begin{quantikz}
\lstick{$\rho$} &&& \gate[2]{U} & \meter{a} &\setwiretype{n} \\
\setwiretype{n}            && \lstick{$\ket{0}$} & \setwiretype{q} & \gate{V_a} &\rstick{$\sigma$}
\end{quantikz}     \end{minipage}

    \caption{Efficient state preparation, efficient measurement, and efficient quantum instrument. Every wire represents a register consisting of a polynomial number of qubits.}
    \label{fig:efficient-operations}
\end{figure}
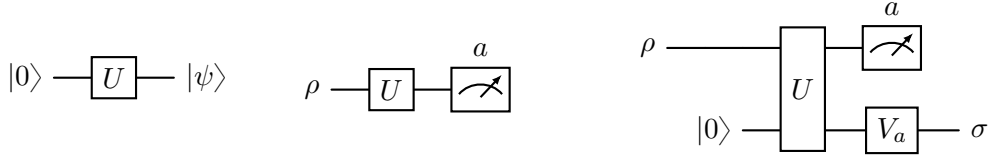

To sum up, for an operation to be efficient we need to know how to physically implement it (in polynomial time), meaning that we have access to one of its purifications.
In Fig. \ref{fig:efficient-operations} there is a graphical representation of the efficient operations we described above.

Let us consider the concrete case of a $2$-players compiled game $\mathcal{G}_\lambda$; the most general efficient compiled quantum strategy is modelled in a Hilbert space $\mathcal{H}^\lambda$ of dimension $d$ as follows:
\begin{enumerate}
    \item the prover prepares an efficiently preparable state $U^\lambda\ket{0}^{\otimes d} = \ket{\psi^\lambda}$, such that $ \ket{\psi^\lambda}\bra{\psi^\lambda}= \rho^\lambda \in \mathsf{S}(\mathcal{H}^\lambda)$;
    \item when she receives the encrypted input $\encx =\enc(x)$, she applies an efficient quantum instrument $\Tilde{A}_{\enca|\encx}^\lambda$, obtaining a classical output $\enca=\enc(a)$, that she sends back to the verifier, and a sub-normalised post-measurement state $\Tilde{A}_{\enca|\encx}^\lambda(\rho^\lambda)$;
    \item finally, she receives $y$ and applies an efficient measurement $B_{b|y}^\lambda$, obtaining a classical output $b$ that she sends back to the prover.
\end{enumerate}

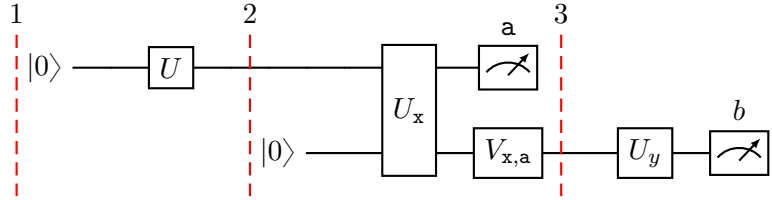
\begin{figure}[htbp]
\centering
\begin{quantikz}
\setwiretype{n}  &\slice{1}  &&\lstick{$\ket{0}$}& \setwiretype{q} & \gate{U} & \slice{2}  &&&& \gate[2]{U_\encx} & \meter{\enca} &\setwiretype{n} 
&&&\\
\setwiretype{n}  &&&&   &&      && \lstick{$\ket{0}$}& \setwiretype{q} & &  \gate{V_{\encx,\enca}}\slice{3} &
& \gate{U_y}& \meter{b} &\setwiretype{n}
\end{quantikz}
     \caption{Circuit representation of a general efficient compiled quantum correlation for bipartite games. (1) is the efficient state preparation of $\rho^\lambda$, (2) is the efficient quantum instrument $\Tilde{A}_{\enca|\encx}^\lambda$ and (3) is the efficient measurement $B_{b|y}^\lambda$.}
    \label{fig:efficient-2-circuit}
\end{figure}

The correlations written with respect to the decrypted variables are :
\begin{align*}
    p_\lambda(a,b|x,y) 
    &= \mathop{\mathds{E}}_{\encx:\dec(\encx) = x} \sum_{\enca:\dec(\enca)=a} \Tr \left[B_{b|y}^\lambda \Tilde{A}_{\enca|\encx}^\lambda(\rho^\lambda)\right]\\
    &= \Tr \left[ B_{b|y}^\lambda \Tilde{A}_{a|x}^\lambda(\rho^\lambda) \right] 
    = \Tr \left[ B_{b|y}^\lambda \rho^\lambda_{a|x} \right].
\end{align*}

In the first line the correlations are expressed with the efficient quantum instrument $\Tilde{A}^{\lambda}_{\enca|\encx}$ applied by the prover; the maps $\Tilde{A}^{\lambda,*}_{\enca|\encx}$ are the adjoint of efficient quantum instruments, thus invoking Lemma~\ref{lem:homomorphism} they are homomorphisms.

Note that in the second line, we move the average and the sum on the encrypted labels inside the definition of $\Tilde{A}_{a|x}^\lambda(\cdot)$, allowing us to work solely with the decrypted labels. These maps are not physically implemented by the prover, because she does not have access to the decryption gate; rather they are a mathematical tool to characterise the decrypted correlations.
The maps $\Tilde{A}^{\lambda,*}_{a|x}$ can be purified to homomorphisms, more details can be found in Appendix~\ref{app:homomorphism}.

The action of the efficient instrument on the initial state can be seen as an efficient preparation of the sub-normalised states $\rho_{a|x}^\lambda$; hence a compiled $2$-players game can always be modelled as a prepare and measure experiment.
More formally, a compiled quantum efficient (QPT) strategy is therefore characterised by:
\begin{enumerate}
    \item a Hilbert space $\mathcal{H}^\lambda$;
    \item QPT-measurable PVMs $B_{b|y}^\lambda \in \mathsf{B}(\mathcal{H}^\lambda)$;
    \item QPT-preparable sub-normalised states $\rho_{a|x}^\lambda \in \mathsf{B}(\mathcal{H}^\lambda)$, such that $\sum_a\rho_{a|x}^\lambda = \rho_{x}^\lambda\in \mathsf{S}(\mathcal{H}^\lambda)$ are states.
\end{enumerate}

What changes if we consider the compilation of a $3$-players game?
We simply have to consider an additional efficient quantum instrument $\Tilde{B}^\lambda_{\encb|\ency}$ that acts on the sub-normalised state produced by the first part of the strategy:
\begin{align*}
    p_\lambda(a,b,c|x,y,z) 
    &= \mathop{\mathds{E}}_{\encx:\dec(\encx) = x} \sum_{\enca:\dec(\enca)=a} \mathop{\mathds{E}}_{\ency:\dec(\ency) = y} \sum_{\encb:\dec(\encb)=b} \Tr \left[ C_{c|z}^\lambda \Tilde{B}_{\encb|\ency}^\lambda \circ \Tilde{A}_{\enca|\encx}^\lambda \left(\rho^\lambda\right) \right]\\
    &= \Tr \left[ C_{c|z}^\lambda \Tilde{B}_{b|y}^\lambda \circ \Tilde{A}_{x|a}^\lambda(\rho^\lambda) \right]\\
    &    = \Tr \left[ C_{c|z}^\lambda \Tilde{B}_{b|y}^\lambda(\rho^\lambda_{a|x}) \right]
    = \Tr \left[  \Tilde{B}_{b|y}^{\lambda,*}(C_{c|z}^\lambda) \rho^\lambda_{a|x} \right].
\end{align*}
Formally, a quantum efficient strategy for a compiled tripartite game $\mathcal{G}^\lambda$ is therefore characterised by:
\begin{enumerate}
    \item a Hilbert space $\mathcal{H}^\lambda$;
    \item QPT-measurable PVMs $C_{c|z}^\lambda \in \mathsf{B}(\mathcal{H}^\lambda)$;
    \item QPT-implementable quantum instruments $\Tilde{B}_{b|y}^\lambda \in \mathsf{CP}(\mathcal{H}^\lambda)$, such that $\Tilde{B}_{y}^\lambda = \sum_b \Tilde{B}_{b|y}^\lambda \in \mathsf{CPTP}(\mathcal{H}^\lambda)$ and $\Tilde{B}_{b|y}^{\lambda,*}$ are *-homomorphisms;
    \item QPT-preparable sub-normalised states $\rho_{a|x}^\lambda \in \mathsf{B}(\mathcal{H}^\lambda)$, such that $\rho_x^\lambda = \sum_a \rho_{a|x}^\lambda \in \mathsf{S}(\mathcal{H}^\lambda)$ are states.
\end{enumerate}
Fig.~\ref{fig:efficient-3-circuit} shows a circuit representation of this strategy.
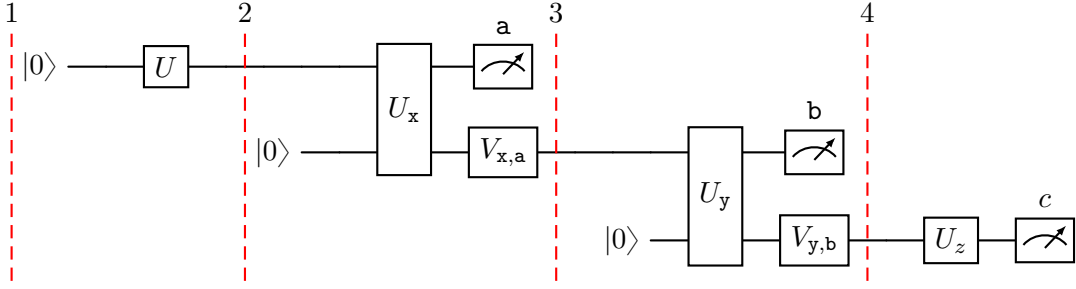
\begin{figure}[htbp]
    \centering
\begin{quantikz}
\setwiretype{n}  &\slice{1} &&\lstick{$\ket{0}$}& \setwiretype{q} & \gate{U} & \slice{2} &&&& \gate[2]{U_\encx} & \meter{\enca} &\setwiretype{n} 
&&&&&
& &&\\
\setwiretype{n}       &&&&&&     && \lstick{$\ket{0}$}& \setwiretype{q} & &  \gate{V_{\encx,\enca}}\slice{3} &
&&& \gate[2]{U_\ency}& \meter{\encb} &\setwiretype{n}
& &&\\
\setwiretype{n} &&&&&& &&&&&&
&& \lstick{$\ket{0}$}& \setwiretype{q} & \gate{V_{\ency,\encb}}\slice{4} &
&\gate{U_z} & \meter{c} &\setwiretype{n} 
\end{quantikz}     \caption{Circuit representation of a general efficient compiled quantum correlation for tripartite games. (1) is the efficient state preparation of $\rho^\lambda$, (2) is the efficient quantum instrument $\Tilde{A}_{\enca|\encx}^\lambda$, (3) is the efficient quantum instrument $\Tilde{B}_{\encb|\ency}^\lambda$  and (4) is the efficient measurement $C_{c|z}^\lambda$.}
    \label{fig:efficient-3-circuit}
\end{figure}

From an operational perspective, the three fundamental primitives underlying these correlations are preparation, transformation, and measurements.
To generalise it to $\mathsf{k}>3$ parties, no new types of operations are required. Indeed, we can model the additional players as a concatenation of transformations, obtaining the following correlations
\begin{equation*}
    p_\lambda(\Vec{a}|\Vec{x}) = \Tr \left[ M^{(\mathsf{k}),\lambda}_{a_\mathsf{k}|x_\mathsf{k}}\cdot I^{(\mathsf{k}-1),\lambda}_{a_{\mathsf{k}-1}|x_{\mathsf{k}-1}}\circ \dots \circ I^{(2),\lambda}_{a_2|x_2} \left(\rho_{a_1|x_1}^\lambda\right) \right]
\end{equation*}
which can be interpreted as a state preparation, a measurement and $\mathsf{k}-2$ instruments in between.

Hence, for a compiled $\mathsf{k}$-partite game $\mathcal{G}^\lambda$, a QPT strategy is characterised by
\begin{enumerate}
    \item a Hilbert space $\mathcal{H}^\lambda$;
    \item QPT-measurable PVMs $M^{(\mathsf{k}),\lambda}_{a_\mathsf{k}|x_\mathsf{k}} \in \mathsf{B}(\mathcal{H}^\lambda)$;
    \item $\mathsf{k}-2$ QPT-implementable quantum instruments $I^{(\mathsf{i}),\lambda}_{a_{i}|x_{i}} \in \mathsf{CP}(\mathcal{H}^\lambda)$, for $i\in[2,\mathsf{k}-1]$, such that $I^{(\mathsf{i}),\lambda}_{x_{i}}= \sum_{a_i}I^{(\mathsf{i}),\lambda}_{a_{i}|x_{i}} \in \mathsf{CPTP}(\mathcal{H}^\lambda)$
    and $I^{(\mathsf{i}),\lambda,*}_{a_{i}|x_{i}}$ is a *-homomorphism;
    \item QPT-preparable sub-normalised states $\rho_{a_1|x_1}^\lambda \in \mathsf{B}(\mathcal{H}^\lambda)$, such that $\rho_{x_1}^\lambda = \sum_{a_1} \rho_{a_1|x_1}^\lambda \in \mathsf{S}(\mathcal{H}^\lambda)$ are states.
\end{enumerate}

\subsection{Constraints from IND-CPA security}\label{subsec:constraints-quantum-IND-CPA}

It is well established that the security of the encryption scheme imposes that all the correlations produced by the compiled prover must be almost-no-signalling.
In this section we will characterise even more the set of compiled quantum QPT correlations.
We will first recap the case of $2$ players introduced in \cite{KMPSW24bound}; we will then extend this to $3$ players. The extension requires new techniques for block-encoding estimations, which might be of independent interest. The extension from $3$ to more players follows quite directly.

\subsubsection{Two players}
Consider $2$-player QPT quantum strategies as considered in Section~\ref{subsec:eff-quantum-str}; they are almost-no-signalling, meaning that no efficient measurement $B_{b|y}^\lambda$ can distinguish $\rho_x^\lambda$ and$\rho_{x'}^\lambda$ with more than negligible probability:
\begin{align*}
    \left| \Tr\left[ B_{b|y}^\lambda \rho_x^\lambda \right] - \Tr\left[ B_{b|y}^\lambda \rho_{x'}^\lambda \right] \right| \leq \textsf{negl}(\lambda).
\end{align*}

In \cite{KMPSW24bound}, the authors show that the IND-CPA security is also imposing constraints on non-physical correlations.
\begin{theorem}[Proposition 4.6 of \cite{KMPSW24bound}]\label{th:poly-bipartite}
    Consider any QPT strategy for a $2$-player compiled game, as in Section~\ref{subsec:eff-quantum-str}. For all polynomials $P(\{B_{b|y}\})$ in the non-commuting variables $\{B_{b|y}\}_{b,y}$, the following statement is true for all $x,x' \in \mathcal{I}_1$
\begin{align}\label{eq:poly-bipartite}
    \left| \Tr\left[  P(\{B_{b|y}^\lambda \}) \rho_x^\lambda \right] - \Tr\left[ P(\{B_{b|y}^\lambda\})  \rho_{x'}^\lambda \right] \right| \leq \textsf{negl}(\lambda).
\end{align}
\end{theorem}

Polynomials of observables in general have no physical meaning; not being unitary, they cannot be seen as physical evolution.
To make sense of these objects we need to introduce block-encodings, whose basic idea is to give a physical meaning to all operators, by interpreting them as a unitary evolution followed by a post-selection. For a pedagogical introduction to these techniques, we refer to \cite{Martyn_2021}.

\begin{definition}[Block-encoding]
    Consider an operator $\mathcal{O} \in \mathsf{B}(\mathcal{H})$.
    The unitary $U$ is said to be the block-encoding of $\mathcal{O}$ with scale factor $t\geq1$ and $p$ ancillas, if we can write
    \begin{align*}
        t \mathcal{O} = \left(\bra{0}^{\otimes p} \otimes \mathds{1}_\mathcal{H} \right) U \left(\ket{0}^{\otimes p} \otimes \mathds{1}_\mathcal{H} \right).
    \end{align*}
    We say that a block-encoding is QPT implementable if the unitary $U$ is QPT implementable.
\end{definition}

Block-encodable operators behave nicely under linear combinations.
\begin{lemma}[Polynomials of block-encodings, \cite{Gily_n_2019} (Lemma 52 and 53)] \label{lem:poly-QPT-be}
    Sums and products of operators with block-encodings also have a block-encoding. If the block-encodings of the factors are QPT, then the block-encoding of the polynomial is also QPT.
\end{lemma}

By definition efficient measurements $B_{b|y}^\lambda$ have a trivial QPT block-encoding, i.e. we know how to efficiently implement them with a circuit.
The polynomials $P(\{B_{b|y}^\lambda\})$ will also have a QPT block-encoding, and can be physically implemented with a circuit by the compiled prover.
Hence IND-CPA security should imply that these polynomials cannot distinguish between $\rho_x^\lambda$ and $\rho_{x'}^\lambda$.
Not so directly! We need to make sure that these operations are efficient, and in general the post-selection will not be efficient at all. In other words, a QPT block encoding does not implies a QPT implementation.

Luckily though, it implies efficient estimation. Indeed there is an efficient algorithm that can estimate the trace of an operator with a QPT block-encoding on every state, up to a $\epsilon = \mathsf{poly}^{-1}(\lambda)$ error. These techniques were originally developed for energy estimation in \cite{Rall_2021}, and adapted to compiled games in \cite{NZ2023Bounding}.

\begin{lemma}[Efficient energy estimator, Lemma 14 of \cite{NZ2023Bounding}]\label{lemma:energy_estimator_only_states}
        For every Hermitian operator $\mathcal{M}^\lambda \in \mathsf{B}(\mathcal{H})$ with $\sup_\lambda\| \mathcal{M}^\lambda \|_\infty = O(1)$ with a QPT block-encoding and for any $\mathsf{poly}(\lambda)$, there is an efficient measurement $\{M_\beta^\lambda\}_\beta$ such that
    \begin{equation*}
        \left|  \Tr(\mathcal{M}^\lambda \rho) - \sum_\beta \beta \Tr(M_\beta^\lambda \rho)\right| \leq \frac{1}{\mathsf{poly}(\lambda)} \qquad \forall \rho \in \mathsf{S}(\mathcal{H}).
    \end{equation*}
\end{lemma}
In the following, we will omit the superscript $\lambda$.
This Lemma is referred to as efficient energy estimator because it was originally formulated in \cite{Rall_2021} for Hamiltonians with a QPT block-encoding.
The result can be generalised to all operators $\mathcal{M}$ with bounded operator norm and a QPT block-encoding and all operators $\rho$ with bounded trace norm; this will be relevant for the extension to three and more parties.

\begin{lemma}[Efficient estimate of a block-encodable operator]\label{lemma:eff-be-op-bip}
    For every operator $\mathcal{M} \in \mathsf{B}(\mathcal{H})$ with $\| \mathcal{M} \|_\infty = O(1)$ with a QPT block-encoding and for any $\mathsf{poly}(\lambda)$, there is an efficient measurement $\{M_\beta\}_\beta$ such that
    \begin{equation*}
        \left|  \Tr(\mathcal{M} \rho) - \sum_\beta \beta \Tr(M_\beta \rho)\right| \leq \frac{1}{\mathsf{poly}(\lambda)} \qquad \forall \rho \in \mathsf{B}(\mathcal{H}) \text{ s.t. } \| \rho \|_1 = O(1).
    \end{equation*}
\end{lemma}
\begin{proof}
    We will prove this statement by reducing it to Lemma~\ref{lemma:energy_estimator_only_states}. Note that the difference between the statements of the two lemmas is that Lemma~\ref{lemma:energy_estimator_only_states} considers only states $\rho$, \emph{i.e.}, positive Hermitian trace-one operators, while Lemma~\ref{lemma:eff-be-op-bip} generalises the statement to all operators $\rho$ with uniformly bounded trace norm, and that $\mathcal{M}$ is no longer required to be Hermitian.

    Let henceforth $\rho \in \mathsf{B}(\mathcal{H})$ with $\| \rho \|_1 = O(1)$, and $\mathcal{M}\in \mathsf{B}(\mathcal{H})$ with $\|\mathcal{M}\|_\infty = O(1)$ admitting a QPT block-encoding. Using the Toeplitz decomposition, both $\rho$ and $\mathcal{M}$ can be decomposed into their Hermitian and skew-Hermitian parts as
    \begin{align*}
        &\rho = \rho_H + i\rho_S, &&\rho_H = \frac{\rho + \rho^*}{2}, &&\rho_S = \frac{\rho - \rho^*}{2i}, \\
        &\mathcal{M} = \mathcal{M}_H + i\mathcal{M}_S, &&\mathcal{M}_H = \frac{\mathcal{M} + \mathcal{M}^*}{2}, &&\mathcal{M}_S = \frac{\mathcal{M} - \mathcal{M}^*}{2i},
    \end{align*}
    where $\rho_H,\rho_S,\mathcal{M}_H,\mathcal{M}_S$ are Hermitian matrices and
    \begin{align*}
        \| \rho_H \|_1, \| \rho_S \|_1 \leq \| \rho \|_1, \quad \| \mathcal{M}_H \|_\infty, \| \mathcal{M}_S \|_\infty \leq \| \mathcal{M} \|_\infty.
    \end{align*}
    Note that by Lemma~\ref{lem:poly-QPT-be}, $\mathcal{M}_H$ and $\mathcal{M}_S$ have QPT block-encodings.
    Furthermore, $\rho_H$ and $\rho_S$ can be written as the difference of positive operators defined as follows
    \begin{align*}
        &\rho_H = \rho_H^+ - \rho_H^-, &&\rho_H^+ = \frac{|\rho_H| + \rho_H}{2}, &&\rho_H^- = \frac{|\rho_H| - \rho_H}{2}, \\
        &\rho_S = \rho_S^+ - \rho_S^-, &&\rho_S^+ = \frac{|\rho_S| + \rho_S}{2}, &&\rho_S^- = \frac{|\rho_S| - \rho_S}{2},
    \end{align*}
    where we define the absolute value $|\sigma| = \sqrt{\sigma^2}$ for Hermitian operators $\sigma$. The operators $\rho_H^+, \rho_H^-, \rho_S^+, \rho_S^-$ are Hermitian and positive, and
    \begin{align*}
        \Tr( \rho_H^+ ), \Tr( \rho_H^- ), \Tr( \rho_S^+ ), \Tr( \rho_S^- ) \leq \| \rho \|_1.
    \end{align*}
    Assume without loss of generality that all of these four traces are non-zero, and define the following states
    \begin{align*}
        \tilde{\rho}_H^+ = \frac{\rho_H^+}{\Tr( \rho_H^+ )}, \quad
        \tilde{\rho}_H^- = \frac{\rho_H^-}{\Tr( \rho_H^- )}, \quad
        \tilde{\rho}_S^+ = \frac{\rho_S^+}{\Tr( \rho_S^+ )}, \quad
        \tilde{\rho}_S^- = \frac{\rho_S^-}{\Tr( \rho_S^- )}.
    \end{align*}
    We can then apply Lemma~\ref{lemma:energy_estimator_only_states} to both operators $\mathcal{M}_H, \mathcal{M}_S$ to obtain efficient measurements $\{M^k_{\beta_k}\}_{\beta_k}$ for $k \in \{H,S\}$ such that
    \begin{align*}
        \left|  \Tr(\mathcal{M}_k \sigma) - \sum_{\beta_k} \beta_k \Tr(M^k_{\beta_k} \sigma)\right| \leq \frac{1}{\mathsf{poly}(\lambda)}, \qquad \forall \sigma \in \mathsf{S}(\mathcal{H}).
    \end{align*}
    Decomposing $\rho$ into a linear combination of the four states $\tilde{\rho}_H^+, \tilde{\rho}_H^-, \tilde{\rho}_S^+, \tilde{\rho}_S^-$ as shown above, and setting $i^H = 1$, $i^S=i$, we finally arrive at
    \begin{align*}
        &\left|  \Tr(\mathcal{M} \rho) - \sum_{k\in\{H,S\}} i^{k} \sum_{\beta_k} \beta_k \Tr(M^k_{\beta_k} \rho) \right|
        \leq  \sum_{k\in\{H,S\}} \left| \Tr(\mathcal{M}_k \rho) - \sum_{\beta_k} \beta_k \Tr(M^k_{\beta_k} \rho) \right| \\
        &\leq  \sum_{k\in\{H,S\}} \sum_{l\in\{H,S\}} \sum_{m\in\{+,-\}} \| \rho \|_1 \left| \Tr(\mathcal{M}_k \rho_l^m) - \sum_{\beta_k} \beta_k \Tr(M^k_{\beta_k} \rho_l^m) \right|
        \leq \frac{8\| \rho \|_1}{\mathsf{poly}(\lambda)} \leq \frac{O(1)}{\mathsf{poly}(\lambda)},
    \end{align*}
    which concludes the proof of the claim.
\end{proof}

With these important results connecting polynomials of observables, QPT block-encodings and their efficient estimations, we have the necessary tools to prove Theorem \ref{th:poly-bipartite}.

\begin{proof}[Proof of Theorem \ref{th:poly-bipartite}]
The measurements $B_{b|y}$ have a trivial QPT block-encoding; Lemma~\ref{lem:poly-QPT-be} implies that $P(\{B_{b|y}^\lambda \})$ also have a QPT block-encoding, but it is not hermitian in general.
From Lemma~\ref{lemma:eff-be-op-bip} we then know that there is an efficient measurement $\{ M_\beta^{b,y} \}_\beta$ that can efficiently estimate the trace of $P(\{B_{b|y}^\lambda \})$ on every state :
\begin{equation*}
    \left| \Tr\left[ P(\{B_{b|y}^\lambda \}) \rho \right] - 
    \sum_\beta \beta \Tr\left[ M_\beta^{b,y} \rho \right]
    \right| \leq \frac{1}{\mathsf{poly}(\lambda)} \qquad \forall \rho\in\mathsf{S}(\mathcal{H}).
\end{equation*}

By rearranging the terms and invoking the triangular inequality, we obtain the following bound
\begin{align*}
 \Bigg| \Tr\left[ \rho_x^\lambda P(\{B_{b|y}^\lambda \}) \right]- &\Tr\left[ \rho_{x'}^\lambda P(\{B_{b|y}^\lambda\}) \right] \Bigg|\\
    = \Bigg| &\Tr\left[\rho_x^\lambda P(\{B_{b|y}^\lambda \}) \right] - \sum_\beta \beta \Tr\left[ M_\beta^{b,y} \rho_x^\lambda \right] + \sum_\beta \beta \Tr\left[ M_\beta^{b,y} \rho_x^\lambda \right]\\ 
    - &\Tr\left[ \rho_{x'}^\lambda P(\{B_{b|y}^\lambda \})\right]  + \sum_\beta \beta \Tr\left[ M_\beta^{b,y} \rho_{x'}^\lambda \right] - \sum_\beta \beta \Tr\left[ M_\beta^{b,y} \rho_{x'}^\lambda \right] \Bigg| \\
    &\leq \frac{2}{\mathsf{poly}(\lambda)} + \left|
        \sum_\beta \beta \Tr\left[ M_\beta^{b,y} \rho_{x}^\lambda \right] -
        \sum_\beta \beta \Tr\left[ M_\beta^{b,y} \rho_{x'}^\lambda \right]
        \right|\\
    & \leq \frac{2}{\mathsf{poly}(\lambda)} + \textsf{negl}(\lambda)
\end{align*}
where in the second-to-last line we have QPT implementable states and measurements, hence we can use the IND-CPA security statement (Def.~\ref{def:IND-CPA}). 
Since this has to be true for any $\mathsf{poly}(\lambda)$, we can choose a polynomial such that the final quantity is negligible.
\end{proof}

\subsubsection{Three players}
In this section we want to characterise the constraints deriving from IND-CPA for $3$-players compiled games, following the same approach as in Theorem \ref{th:poly-bipartite}.
We will start by showing a technical Lemma about diagonal non-unitary operators, which might be of independent interest.

\begin{lemma}[Efficient implementation of a diagonal non-unitary operator]\label{lem:qpt-impl-diagonals}
    Consider a diagonal matrix $D$, with efficiently computable and uniformly bounded entries: $|\gamma_i| \leq C \;\forall i$.
    Consider $\sigma \in \mathsf{S}(\mathcal{H}^\lambda)$ an efficiently preparable quantum state, $\mathcal{U}\in \mathsf{B}(\mathcal{H}^\lambda)$ an efficiently implementable unitary and $M\in\mathsf{B}(\mathcal{H}^\lambda)$ an efficiently implementable measurement.
    Then the following trace
    \begin{equation*}
        \Tr\left[
        (\mathcal{U} D \mathcal{U}^*) \sigma (\mathcal{U} D^* \mathcal{U}^*) M
        \right]
        \qquad \forall M\in \mathsf{B}(\mathcal{H}).
    \end{equation*}
    can be efficiently computed.

    Furthermore, if $M$ is not an efficient implementable operator, but just a bounded operator $\| M \|_\infty \leq O(1)$ with a QPT block encoding, the quantity above can be efficiently estimated.
\end{lemma}
\begin{proof}
    We need to show that the operator
\begin{align}
    \rho = \mathcal{U} D \mathcal{U}^* \sigma \mathcal{U} D^* \mathcal{U}^*,
\end{align}
describes a physically meaningful object which can be efficiently prepared.
First note that this can be rewritten as
\begin{align}
    \rho = \mathcal{E} \circ \mathcal{F} \circ \mathcal{E}^{-1} ( \sigma ),
\end{align}
where $\mathcal{E}(\cdot) = \mathcal{U} \cdot \mathcal{U}^*$ and $\mathcal{F}(\cdot) = D \cdot D^\ast$ are bounded, positive maps. This implies that $\rho$ is a Hermitian, positive operator with uniformly bounded trace:
\begin{align}
    \Tr \left( \rho \right) \leq \|\mathcal{U}\|_\infty^4 \|D\|_\infty^2 \Tr(\sigma) \leq C^2.
\end{align}
Instead of directly preparing $\rho$, we argue that the sub-normalised state $\rho/C^2$ can be prepared efficiently. For any operator $M$, we can then write
\begin{align}
    \Tr \left( M \rho \right) = C^2 \Tr \left( M \rho/C^2 \right).
\end{align}
Since the renormalisation factor $C^2$ is uniformly bounded, any approximation of the trace on the right-hand side degrades at most by a constant multiplicative factor when turned into an approximation of the left-hand expression.

The unitary quantum channels $\mathcal{E}$ and $\mathcal{E}^{-1}$ have efficient implementations, thus it suffices to argue that the renormalised channel $\tilde{\mathcal{F}}(\cdot) = \mathcal{F}(\cdot)/C^2 = \frac{D}{C} \cdot \frac{D^\ast}{C}$ can also be efficiently implemented.

Denoting $\tilde{D} = D/C$ and its renormalised eigenvalues $\tilde{\gamma_i} = \gamma_i/C$, we first note that there exists an efficient classical circuit evaluating the map $i \mapsto \tilde{\gamma_i}$.
As $\tilde{D}$ is a contraction, we can construct its unitary 1-dilation $W$ explicitly as
\begin{align}
    W =
    \begin{bmatrix}
        \tilde{D} & \sqrt{I - \tilde{D}^\ast \tilde{D}} \\
        \sqrt{I - \tilde{D}^\ast \tilde{D}} & \tilde{D}^\ast
    \end{bmatrix},
\end{align}
following Egerváry's dilation theorem \cite{Egervary54}. A straightforward calculation verifies that $W$ is indeed a unitary satisfying $\tilde{D} = (\bra{0} \otimes \mathds{1}) W (\ket{0} \otimes \mathds{1})$. Consequently, it suffices now to argue that the unitary channel described by $W$ can be efficiently implemented, since the preparation of the single-qubit ancilla state in the computational basis, along with its computational basis measurement and post-selection of the 0-outcome, are all efficient operations. Note, that the post-selection does not require renormalisation in this case and remains therefore efficient even though the trace of the resulting sub-normalised state may be arbitrarily small.
Now note that the quantum channels given by the $2\times 2$-unitaries $W_i$ where
\begin{align}
    W_i =
    \begin{bmatrix}
        \tilde{\gamma_i} & \sqrt{1 - \tilde{\gamma_i}^\ast \tilde{\gamma_i}} \\
        \sqrt{1 - \tilde{\gamma_i}^\ast \tilde{\gamma_i}} & \tilde{\gamma_i}^\ast
    \end{bmatrix}
\end{align}
are efficiently implementable for all $i$, and so are their controlled versions
\begin{align}
    C_i W_i : |j\rangle |\varphi\rangle \mapsto
    \left\{
    \begin{matrix}
        |i\rangle W_i |\varphi\rangle, &\text{ if } j=i, \\
        |j\rangle |\varphi\rangle, &\text{ if } j\neq i,
    \end{matrix}
    \right.
\end{align}
where the control is to be understood as applying $W_i$ only if the first register is in the state $i$, and $|\varphi\rangle$ is any pure single-qubit state.
Since, as argued previously, the map $i \mapsto \tilde{\gamma_i}$ has an efficient classical implementation, there further exists an efficient quantum circuit implementing the map
\begin{align}
    W = \left( \prod_i C_i W_i \right) : |j\rangle |\varphi\rangle \mapsto |j\rangle W_j |\varphi\rangle
\end{align}
for any state $|j\rangle$ in the computational basis and any single-qubit state $|\varphi\rangle$, which concludes the first part of the proof.
Similar questions and similar techniques can also be found in \cite{Schlimgen_2022}.

If $M$ is not an efficient implementable operator, but a bounded operator $\| M \|_\infty \leq O(1)$ with a QPT block encoding, we can use the construction above to efficiently implement $\rho/C^2$, and then we can apply Lemma~\ref{lemma:energy_estimator_only_states} to efficiently estimate the trace.
\end{proof}

With the tools developed so far, we can now formulate strong constraints that any tripartite compiled quantum QPT strategies have to satisfy, not to contradict IND-CPA.
This will be an important tool to show the asymptotic bound in Section~\ref{sec:asymptotic-bound}.

\begin{theorem}\label{th:3pl-ind-cpa}
    Consider a tripartite quantum compiled strategy, as in Section~\ref{subsec:eff-quantum-str}. For all operators $\mathcal{L^\lambda, R^\lambda} \in \mathsf{B}(\mathcal{H}^\lambda)$ with $\sup_\lambda \|\mathcal{L}^\lambda\|_\infty=O(1)$ and $\sup_\lambda \|\mathcal{R}^\lambda\|_\infty=O(1)$ that have a QPT block-encoding, and for all non-commuting polynomials $P(\{C_{c|z}^\lambda \})$, the following statement is true
    \begin{align*}
    \left|\sum_b \Tr\left[P(\{C_{c|z}^\lambda\}) \Tilde{B}_{b|y}^{\lambda} \left(\mathcal{R}^\lambda \rho_{a|x}^\lambda\mathcal{L}^{\lambda,*} \right) \right] - \sum_{b'} \Tr\left[P(\{C_{c|z}^\lambda\}) \Tilde{B}_{b'|y'}^{\lambda} \left(\mathcal{R}^\lambda \rho_{a|x}^\lambda\mathcal{L}^{\lambda,*} \right) \right] \right| \leq \textsf{negl}(\lambda).
\end{align*}
\end{theorem}
\begin{proof}
Consider the adjoint of the quantum instrument, we will prove the following equivalent statement:
\begin{align*}
    \left| \Tr\left[ \mathcal{L}^{\lambda,*} \Tilde{B}_{y}^{\lambda,*}\left(P(\{C_{c|z}^\lambda \})\right) \mathcal{R}^\lambda \rho_{a|x}^\lambda \right] - \Tr\left[ \mathcal{L}^{\lambda,*} \Tilde{B}_{y'}^{\lambda,*}\left(P(\{C_{c|z}^\lambda \})\right) \mathcal{R}^\lambda \rho_{a|x}^\lambda \right] \right| \leq \textsf{negl}(\lambda).
\end{align*}
To simplify the notation in the proof, we drop the label $\lambda$, and we use the shortcut notation
 $\rho^\lambda_{a|x} = \rho$ and
$\Tilde{B}_y^{\lambda,*}(P\{C^\lambda_{c|z}\}) = B_y$.
The proof will follow three steps:
\begin{enumerate}
    \item decompose bounded maps into completely positive maps;
    \item sequentially apply Lemma~\ref{lemma:eff-be-op-bip};
    \item use Lemma~\ref{lem:qpt-impl-diagonals} to efficiently estimate the results.
\end{enumerate}

1. Every bounded map $\mathcal{L}^* B \mathcal{R}$ can be decomposed in completely positive maps $\mathcal{V}^* B \mathcal{V}$. Consider the following operators:
\begin{align*}
    & \mathcal{V}_1 = \mathcal{L} -i \mathcal{R}, && \mathcal{V}_3 = \mathcal{L} + i \mathcal{R},\\
    & \mathcal{V}_2 = \mathcal{L} - \mathcal{R}, &&\mathcal{V}_4 = \mathcal{L} + \mathcal{R}.
\end{align*}
A simple computation leads to the following decomposition 
\begin{equation*}
    \mathcal{L}^* B \mathcal{R} = \frac{1}{4} \sum_{j=1}^4 i^j \mathcal{V}_j^* B \mathcal{V}_j
\end{equation*}
which is known in the literature as Wittstock’s decomposition theorem. Using the triangular inequality we can verify that
    \begin{equation*}
        \sum_{j=1}^4 \left| 
        i^j \left[ \Tr\left(\sigma \mathcal{V}_j^* B_y \mathcal{V}_j \right) -
        \Tr\left(\sigma \mathcal{V}_j^* B_{y'} \mathcal{V}_j \right) \right]
        \right|
        \leq 
          \sum_{j=1}^4\left| 
        \Tr\left(\sigma \mathcal{V}_j^* B_y \mathcal{V}_j \right) -
        \Tr\left(\sigma \mathcal{V}_j^* B_{y'} \mathcal{V}_j \right) 
        \right|.
    \end{equation*}
Therefore, to prove theorem \ref{th:3pl-ind-cpa} it is sufficient to prove that 
    \begin{equation*}
    \left| 
        \Tr\left(\sigma \mathcal{V}_j^* B_y \mathcal{V}_j \right) -
        \Tr\left(\sigma \mathcal{V}_j^* B_{y'} \mathcal{V}_j \right) 
        \right|
        \leq \eta(\lambda) \qquad \forall j.
    \end{equation*}

2. First apply Lemma \ref{lemma:eff-be-op-bip} considering $\mathcal{M}= \mathcal{V}$ and $\rho = \sigma \mathcal{V}^* B$; the theorem applies because $\| \mathcal{V}\|_\infty = O(1)$ and $\|\rho\|_1 = \|\sigma \mathcal{V}^* B\|_1 \leq \|B\|_\infty \|\mathcal{V}^*\|_\infty \|\sigma\|_1 = O(1)$ thanks to Holder's inequality :
\begin{align*}
    \left| \Tr\left( \sigma \mathcal{V}^* B \mathcal{V}\right) - \Tr\left( \sigma \mathcal{V}^* B \left(\sum_i\gamma_i N_{\gamma_i}\right) \right)
    \right| \leq \frac{1}{\mathsf{poly}(\lambda)}.
\end{align*}

Then, apply the theorem again, considering $\mathcal{M}= \mathcal{V}^*$ and $\rho = B \left(\sum_i\gamma_i N_{\gamma_i}\right) \sigma$.
Recall that if $\mathcal{M}$ has a QPT block encoding $U$, then $\mathcal{M}^*$ also has a QPT block encoding $U^*$.
Equivalently, if $\{N_{\gamma_i}\}_{\gamma_i}$ is the QPT-measurable POVM associated to $\mathcal{V}$ via Lemma \ref{lemma:eff-be-op-bip}, the QPT-measurable POVM associated to $\mathcal{V}^*$ will be $\{N_{\gamma_i}^*\}_{\gamma_i^*}$. We obtain the following:
\begin{align*}
    \left| \Tr\left(B \left(\sum_i\gamma_i N_{\gamma_i}\right) \sigma \mathcal{V}^*  \right) - 
    \Tr\left(B \left(\sum_i\gamma_i N_{\gamma_i}\right) \sigma  \left(\sum_j\gamma_j^* N_{\gamma_j}^*\right)  \right)
    \right| \leq \frac{1}{\mathsf{poly}(\lambda)}.
\end{align*}
The two statements above taken together imply 
\begin{align*}
    \left| \Tr\left( \sigma \mathcal{V}^* B_y \mathcal{V}\right) - \Tr\left(\sigma  \left(\sum_j\gamma_j^* N_{\gamma_j}^*\right) B_y\left( \sum_i\gamma_i N_{\gamma_i}\right)  \right)
    \right| \leq \frac{2}{\mathsf{poly}(\lambda)}.
\end{align*}
Denote the second trace of the equation above as $\Delta_y$. Then we can show that
\begin{align*}
     | \Tr\left(\sigma \mathcal{V}^* B_y \mathcal{V} \right) &-
         \Tr\left(\sigma \mathcal{V}^* B_{y'} \mathcal{V} \right) 
        | \\
        &= \left| \Tr\left(\sigma \mathcal{V}^* B_y \mathcal{V} \right) - \Delta_y + \Delta_y -
        \Tr\left(\sigma \mathcal{V}^* B_{y'} \mathcal{V} \right)  + \Delta_{y'} - \Delta_{y'}
        \right| \\
        &\leq \frac{4}{\mathsf{poly}(\lambda)} + \left|
        \Delta_y - \Delta_{y'}
        \right|.
\end{align*}

3. It remains to demonstrate that $\Delta_y$ can be efficiently estimated.
Recall that without loss of generality, we can model any efficient measurement $\{N_{\gamma_i}\}_{\gamma_i}$ as a QPT unitary $\mathcal{U}$ followed by a measurement in the computational basis, $N_{\gamma_i} = \mathcal{U} \ket{i}\bra{i} \mathcal{U}^*$. Then we can write
\begin{equation*}
    \sum_i \gamma_i N_{\gamma_i} = \mathcal{U} \left(\sum_i \gamma_i \ket{i}\bra{i} \right) \mathcal{U}^* = \mathcal{U} D \mathcal{U}^*
\end{equation*}
and consequently
\begin{equation*}
    \Delta_y = \Tr\left[
    \mathcal{U} D \mathcal{U}^* \sigma \mathcal{U} D^* \mathcal{U}^* B_y
    \right].
\end{equation*}
Recall that $B_y = \Tilde{B}_y^{\lambda,*}(P\{C^\lambda_{c|z}\})$ is not a physical operation, but has a QPT implementable block-encoding. We can apply the result of Lemma~\ref{lem:qpt-impl-diagonals}, to establish that $\Delta_y$ can be efficiently estimated up to inverse polynomial precision; this is true for every polynomial. Hence, by the IND-CPA security of the scheme the term $| \Delta_y - \Delta_{y'}|$ is negligible.
This concludes the proof of the theorem.
\end{proof}

\begin{remark}
    Note that Theorem~\ref{th:3pl-ind-cpa} is formulated for the original (unpurified) description of the prover's strategy in the compiled game, before the purification described in Section~\ref{subsec:eff-quantum-str} and Appendix~\ref{app:homomorphism}.
    The statement of the theorem holds equivalently for the purified strategy, \emph{i.e.}, pure states $\rho_{a|x}^\lambda$, purified instruments $\tilde{B}_{b|y}^\lambda$, and purified measurements (PVMs) $C_{c|z}^\lambda$, with one small modification:
    the operators $\mathcal{L}^\lambda$ and $\mathcal{R}^\lambda$ are restricted to live in the subalgebra of operators generated by the operators $\tilde{B}_{b|y}^{\lambda,\ast}(C_{c|z}^\lambda)$.
    This restriction is sufficient to prevent the operators $\mathcal{L}^\lambda, \mathcal{R}^\lambda$ from acting non-trivially on any wires that the unpurified prover would not have had access to.
    The claim then directly reduces to the statement of Theorem~\ref{th:3pl-ind-cpa} by tracing out the additional wires added during the purification as early as possible, as they are never non-trivially acted on anymore later on in the experiment.
\end{remark}

\subsubsection{More than three players}

Consider the quantum compiled correlations for $\mathsf{k}$ players described in Section~\ref{subsec:eff-quantum-str}. What follows is a generalisation of Theorem~\ref{th:3pl-ind-cpa} for such correlations.
We use $\mathop{\bigcirc}_{i=1}^n f_i = f_n \circ \dots \circ f_1$ as a compact notation for the composition of maps.
\begin{theorem}\label{th:kpl-ind-cpa}
    Consider a $\mathsf{k}$-partite quantum compiled strategy, as in Section~\ref{subsec:eff-quantum-str}.
    Consider operators $\mathcal{L^\lambda, R^\lambda} \in \mathsf{B}(\mathcal{H}^\lambda)$ with $\sup_\lambda \|\mathcal{L}^\lambda\|_\infty=O(1)$ and $\sup_\lambda \|\mathcal{R}^\lambda\|_\infty=O(1)$ that have a QPT block-encoding, and non-commuting polynomials $P(\{ \cdot \})$.
    The following statement is true for all $i \in [2, \mathsf{k}-1]$ :
\begin{align*}
    \Bigg|
    &\sum_{a_i} \Tr\left[P\left(\left\{ \left(\mathop{\bigcirc}_{j=i+1}^{\mathsf{k}-1}I^{(j),\lambda}_{a_j|x_j}\right)^* ( M_{a_\mathsf{k}|x_\mathsf{k}}^{(\mathsf{k}),\lambda}) \right\}\right) \Tilde{I}_{a_i|x_i}^{(i),\lambda} \left(\mathcal{R}^\lambda \left( \mathop{\bigcirc}_{m=2}^{i-1} I^{(m),\lambda}_{a_m|x_m} \right)(\rho_{a_1|x_1}^\lambda) \mathcal{L}^{\lambda,*} \right) \right] \\
    &- \sum_{a_i'} \Tr\left[P\left(\left\{ \left(\mathop{\bigcirc}_{j=i+1}^{\mathsf{k}-1}I^{(j),\lambda}_{a_j|x_j}\right)^* ( M_{a_\mathsf{k}|x_\mathsf{k}}^{(\mathsf{k}),\lambda}) \right\}\right) \Tilde{I}_{a_i'|x_i'}^{(i),\lambda} \left(\mathcal{R}^\lambda \left( \mathop{\bigcirc}_{m=2}^{i-1} I^{(m),\lambda}_{a_m|x_m} \right)(\rho_{a_1|x_1}^\lambda) \mathcal{L}^{\lambda,*} \right) \right]  
    \Bigg| \\
    &\leq \textsf{negl}(\lambda).
\end{align*}
\end{theorem}
\begin{proof}
We can directly apply the proof of Theorem \ref{th:3pl-ind-cpa}, where for all $i\in[2,\mathsf{k}-1]$ we consider the sub-normalised states 
\begin{equation*}
    \rho_{a|x}^\lambda = \rho_{a_1, \dots, a_{i-1}|x_1, \dots, x_{i-1}}^\lambda = \left( \mathop{\bigcirc}_{m=2}^{i-1} I^{(m),\lambda}_{a_m|x_m} \right)(\rho_{a_1|x_1}^\lambda) 
\end{equation*}
and measurements
\begin{equation*}    C_{c|z}^\lambda = M_{a_{i+1},\dots, a_\mathsf{k}|x_{i+1},\dots, x_\mathsf{k} }^\lambda =\left(\mathop{\bigcirc}_{j=i+1}^{\mathsf{k}-1}I^{(j),\lambda}_{a_j|x_j}\right)^* ( M_{a_\mathsf{k}|x_\mathsf{k}}^{(\mathsf{k}),\lambda}).
\end{equation*}
\end{proof}
 
\newpage
\section{Efficient algebraic strategies for compiled games}
\label{sec:algebraic-strategies}

In \cite{KMPSW24bound} the authors show that for all bipartite games, the constraints arising from security of the QFHE scheme - in the limit of the security parameter going to infinity - impose strong conditions on the (asymptotic) states generated by the first round of the compiled single prover.
In particular, when summed over the output label $a$, the asymptotic states coincide for different inputs $x$.
In this sense, the states are \emph{no-signalling}.
We aim to find a similar statement for all $\mathsf{k}$-players games.

The first step is to properly define this limit, however for different values of $\lambda$, the strategies are defined on different Hilbert spaces $\mathcal{H}^\lambda$.
The solution adopted by \cite{KMPSW24bound} is to consider a single unital C$^*$-algebra $\mathscr{A}$, that doesn't depend on the security parameter. The measurements are positive elements of this algebra $\mathfrak{b}_{b|y} \in \mathscr{A}_+$, such that they sum to the identity of the algebra $\sum_b \mathfrak{b}_{b|y} = \mathfrak{1}$; the sub-normalised states $\rho^\lambda_{a|x}$ are sub-normalised algebraic states, i.e. positive functionals $\phi^\lambda_{a|x}:\mathscr{A}\to\mathbb{C}$ with bounded trace $\|\phi^\lambda_{a|x} \|\leq 1$ such that $\|\sum_a \phi^\lambda_{a|x} \| = 1$. Algebraic correlations are formulated as
\begin{equation*}
    p_\lambda(a,b|x,y) = \phi^\lambda_{a|x} (\mathfrak{b}_{b|y}).
\end{equation*}

For every fixed $\lambda$ the strategy is characterised by a different algebraic state, and taking the limit $\lambda \to \infty$ corresponds to taking the limit of a sequence of normalised functionals on a C$^*$-algebra $\mathscr{A}$.
The space of these states is compact, hence the limiting object exists. 
In this section we will make this intuition more concrete, and extend it to the case of more than $2$ players.

\subsection{Preliminaries on universal C*-algebras and algebraic strategies}\label{subsec:algebras}

Readers unfamiliar with the C*-algebraic framework might wonder why this level of abstraction is even necessary. The advantage is quite subtle, and lies on the fact that powerful structures emerge more naturally.
A key example are universal C$^*$-algebras, which are C$^*$-algebras that are described solely by the means of generators and relations. They serve the purpose to capture abstract properties shared by a class of C$^*$-algebras, and will allow us to translate the analysis of the limit on many different $\mathcal{H}^\lambda$ - that are still C$^*$-algebras - to a single, concretely defined space.

A universal C$^*$-algebra does not necessarily exist for any given set of generators and relations. To show its existence, two necessary and sufficient conditions must be implied by the choice of relations: (i) there exists a concrete realisation of the generators as bounded operators on Hilbert space that satisfy the given relations, and (ii) the relations must effectively impose a uniform bound on the norms of all generators.
For a detailed discussion of the construction and existence of universal C$^*$-algebras, see~\cite{blackadar2006operator}.

\begin{definition}[Universal C$^*$-algebra]\label{def:universal_c_star_algebra}
    Let $G = \{ x_i \; | \; i \in I \}$ be a set of generators, $P(G)$ be the free $^*$-algebra of non-commutative polynomials in elements of $G$ and their adjoints, and $R$ be a set of relations, often of the form
    \begin{align*}
        \| p(x_{i_1}, \dots, x_{i_n}, x^*_{i_1}, \dots, x^*_{i_n}) \| \leq \eta,
    \end{align*}
    where $p \in P(G)$ and $\eta \geq 0$.
    A representation of $(G|R)$ into a C$^*$-algebra $\mathscr{A}$ is a $^*$-representation $\pi : P(G) \to \mathscr{A}$ such that $\pi(G) \subseteq \mathscr{A}$ satisfies the relations in $R$.
    Now define
    \begin{align*}
        \|\mathfrak{a}\| = \sup \{ \|\pi(\mathfrak{a})\| \; | \; \pi \text{ is a representation of } (G|R) \},
    \end{align*}
    for all $\mathfrak{a} \in P(G)$. If $\|\mathfrak{a}\| < \infty$ for all $\mathfrak{a} \in P(G)$, we call
    \begin{align*}
        C^*(G|R) = \overline{P(G)/\{\|\mathfrak{a}\|=0\}}^{\|\cdot\|}
    \end{align*}
    the universal C$^*$-algebra with generators $G$ and relations $R$.
    In the following, we will always consider unital universal C$^*$-algebras which can be realised by adding an additional generator $\mathds{1}$ with the relations $\mathds{1} = \mathds{1}^* = \mathds{1}^2$ and $\mathds{1} x_i = x_i \mathds{1} = x_i$ for all $i \in I$.
\end{definition}

The previous Definition~\ref{def:universal_c_star_algebra} of universal C$^*$-algebras quite straightforwardly implies the following Lemma~\ref{lemma:universality_c_star_algebra}, which explains why C$^*$-algebras constructed in this way are called \emph{universal}.

\begin{lemma}[Universality property]\label{lemma:universality_c_star_algebra}
    Let $G = \{x_i \;|\; i \in I \}$ be a set of generators and $R$ be a set of relations such that $C^*(G|R)$ exists. Let $\mathscr{A}$ be a C$^*$-algebra containing elements $G' = \{x_i' \;|\; i \in I \}$ that satisfy the relations $R$. Then there exists a unique $^*$-homomorphism $\theta$ such that
    \begin{align*}
        \theta : C^*(G|R) \to \mathscr{A}, \quad x_i \mapsto x_i', \quad \forall i \in I.
    \end{align*}
\end{lemma}

This leaves open the question whether universal C$^*$-algebras exist in the specific cases relevant to our setting.
We start with relations that capture the essential properties of PVMs, a situation that was discussed before, e.g. in \cite{Paddock_2023}.

\begin{lemma}[Universal C$^*$-algebra of PVMs, $\mathscr{A}_{B}^{PVM}$]\label{lem:universal_povm_algebra}
    Let $I_B,O_B$ be finite sets. Let $G = \{ \mathfrak{e}_{b|y} \; | \; y\in I_B, b\in O_B \}$ be a set of generators, and consider the relations $R$ given as
    \begin{enumerate}
        \item $\mathfrak{e}_{b|y} = \mathfrak{e}_{b|y}^*$ \hfill $ \forall y\in I_B, b\in O_B$,
        \item $0 \leq \mathfrak{e}_{b|y} \leq \mathfrak{1}$ \hfill $\forall y\in I_B, b\in O_B$,
        \item $\sum_{b\in O_B} \mathfrak{e}_{b|y} = \mathfrak{1}$ \hfill $\forall y\in I_B$,
        \item $\mathfrak{e}_{b|y}\mathfrak{e}_{b'|y} = \delta_{b,b'}\mathfrak{e}_{b|y}$ \hfill $\forall y\in I_B, b,b'\in O_B$.
    \end{enumerate}
    Then, $C^*(G,R)$ exists and we call it the universal C$^*$-algebra $\mathscr{A}^\text{PVM}_{B}$ of PVMs with inputs $I_B$ and outcomes $O_B$.
    To simplify the notation, we will often refer to it as  $\mathscr{A}_{B}$.
\end{lemma}

Note that in \cite{KMPSW24bound} the measurements were modelled through the universal C$^*$-algebra of POVMs $\mathscr{A}^{POVM}$, of which $\mathscr{A}^{PVM}$ is a quotient algebra.
However, they also comment that projective measurements are enough, when considering the efficiency constraint (Def. 4.4 in \cite{KMPSW24bound}).

The universal property ensures that every PVM $B_{b|y}^\lambda \in \mathsf{B}(\mathcal{H}^\lambda)$ can be written as a representation of a generator of the universal C$^*$-algebra
\begin{equation}\label{eq:universal-property-POVM}
    \theta_B^\lambda(\mathfrak{e}_{b|y}) = B_{b|y}^\lambda.
\end{equation}
Let us now rephrase the bipartite compiled strategies in the framework of C$^*$-algebras.
\begin{definition}[C$^*$-algebraic compiled strategy for $\mathsf{k}=2$, $\mathcal{A}_\lambda^2$]
Consider the compiled game $\mathcal{G}_\lambda$. In the algebraic framework, a general efficient strategy is characterised by
    \begin{enumerate}
        \item a universal C$^*$ algebra of PVMs $\mathscr{A}_B$;
        \item PVMs, that are positive elements of the algebra $\mathfrak{b}_{b|y} \in \mathscr{A}_B^+$, that sum up to identity $\sum_b \mathfrak{b}_{b|y} = \mathfrak{1}$;
        \item sub-normalised efficient states, which are positive functional on the algebra $\phi^\lambda_{a|x}:\mathscr{A}_B\to\mathbb{C}$, such that $\phi^\lambda_{x} = \sum_a \phi^\lambda_{a|x}$ is a normalised state.
    \end{enumerate}
    These elements produce a bipartite C$^*$-algebraic compiled strategy :
    \begin{equation*}
        p_\lambda(a,b|x,y)= \phi^\lambda_{a|x}(\mathfrak{b}_{b|y}).
    \end{equation*}
\end{definition}

Using the universal property of the C$^*$-algebra of PVMs, we can define the states 
\begin{equation}\label{eq:def-alg-state}
    \phi^\lambda_{a|x}(\cdot ) = \left[ \sigma^\lambda_{a|x} \theta^\lambda(\cdot) \right],
\end{equation}
such that the algebraic correlations match the quantum compiled strategies defined in Section \ref{subsec:eff-quantum-str}:
\begin{equation}\label{eq:2-alg-corr}
    \phi^\lambda_{a|x}(\mathfrak{b}_{b|y})
    = \Tr\left[ \sigma^\lambda_{a|x} \theta^\lambda(\mathfrak{b}_{b|y}) \right]
    = \Tr\left[ \sigma^\lambda_{a|x} B^\lambda_{b|y} \right].
\end{equation}
This equality is represented with a commuting diagram in Fig. \ref{fig:univAlgebras-2players}.

\begin{figure}[htbp]
    \centering
\begin{tikzpicture}
    \node (ABC) at (2.5,2) {\( \mathscr{A}_B \)};
    \node (C) at (5.5,0) {\( \mathbb{C} \)};
    \node (BH2) at (2.5,0) {\( \mathsf{B}(\mathcal{H}^\lambda) \)};

    \draw[->] (ABC) -- (C) node[midway, above,  yshift=5pt] {\( \phi_{a|x}^\lambda \)};
    \draw[->] (ABC) -- (BH2) node[midway, left] {\( \theta^\lambda \)};
    \draw[->] (BH2) -- (C) node[midway,below] {\( \mathsf{Tr}(\sigma_{a|x}^\lambda \cdot) \)};

\end{tikzpicture}
     \caption{Commutative diagram that describes Eq. \ref{eq:2-alg-corr}. Above there is the universal C$^*$-algebra of PVMs with labels $(I_B,O_B)$, below the real experiment living in a concrete Hilbert space $\mathcal{H}^\lambda$. To obtain the correlation we can either directly apply the algebraic state $\phi_{a|x}^\lambda$, or use the universal property of the universal algebra, and then multiply with the physical state $\sigma_{a|x}^\lambda$ and take the trace.}
    \label{fig:univAlgebras-2players}
\end{figure}

\subsection{Universal C*-algebras for multipartite strategies}%
\label{subsec:universal-algebra-multipartite}

To extend this formalism to more than $2$ players, we construct a new universal C$^*$-algebra, that captures the behaviour of sequential projective measurements.
More precisely, it describes a measurement whose input and output labels have a natural bipartition, that we call $B$ and $C$ for obvious reasons.
We call it sequential $B \to C$ because we impose one-way no-signalling from $C$ to $B$, while in the other direction signalling is allowed.
Up to our knowledge, this construction was never studied in the literature, and it might be of independent interest.

\begin{lemma}[Universal C$^*$-algebra of sequential PVMs]\label{lem:universal_sequential_povm_algebra}
    Let $I_B, O_B, I_C, O_C$ be finite sets. Let $G = \{ \mathfrak{f}_{b,c|y,z} \; | \; y\in I_B, z \in I_C, b\in O_B, c\in O_C  \}$ be a set of generators, and consider the relations $R$ given as
    \begin{enumerate}
        \item $\mathfrak{f}_{b,c|y,z} = \mathfrak{f}_{b,c|y,z}^*$ \hfill $\forall y\in I_B, z \in I_C, b\in O_B, c\in O_C$,
        \item $0 \leq \mathfrak{f}_{b,c|y,z} \leq 1$ \hfill $\forall y\in I_B, z \in I_C, b\in O_B, c\in O_C$,
        \item $\mathfrak{f}_{b,c|y,z}  \mathfrak{f}_{b',c'|y,z}= \delta_{b,b'} \delta_{c,c'} \mathfrak{f}_{b,c|y,z}$ \hfill $\forall y\in I_B, z \in I_C, b,b'\in O_B, c,c'\in O_C$,
        \item $\sum_{b\in O_B, c\in O_C} \mathfrak{f}_{b,c|y,z} = \mathds{1}$ \hfill $\forall y\in I_B, z \in I_C$,
        \item $\sum_{c \in O_C} \mathfrak{f}_{b,c|y,z} = \sum_{c \in O_C} \mathfrak{f}_{b,c|y,z'}$ \hfill $\forall y\in I_B, b\in O_B, z,z' \in I_C$.
    \end{enumerate}
    Then, $C^*(G,R)$ exists and we call it the universal C$^*$-algebra of sequential PVMs $\mathscr{A}^\text{PVM}_{B\to C}$.
    To simplify the notation, we will often refer to this as $\mathscr{A}_{B \to C}$.
\end{lemma}
\begin{proof}
    To see that there exists a representation of $(G,R)$, we construct a concrete realisation of these generators and relations as bounded operators on Hilbert space. Consider any PVMs $\{M_{c|z}\} \subseteq \mathsf{B}(\mathcal{H})$ on some Hilbert space $\mathcal{H}$ with inputs $I_C$ and outputs $O_C$, and efficient quantum instruments $B_{b|y} : \mathsf{B}(\mathcal{K}) \to \mathsf{B}(\mathcal{H})$ with inputs $I_B$ and outputs $O_B$. Recall that $B_{b|y}^*$ is a *-homomorphism; a simple calculation then shows that the operators $\{B_{b|y}^*(M_{c|z})\} \in \mathsf{B}(\mathcal{K})$ give rise to a representation of $(G,R)$. 

    Clearly, the relations further impose a uniform bound on the norms of all generators, as they are all positive and bounded by the identity. This implies the finiteness of $\| \cdot \|$ on $P(G)$ and therefore the existence of $C^*(G|R)$.
\end{proof}

Given a set of PVMs $B^{\lambda,*}_{b|y} (C_{c|z}^\lambda)= M_{bc|yz}^\lambda \in \mathsf{B}(\mathcal{H}^\lambda)$ that respects the relations of Lemma~\ref{lem:universal_sequential_povm_algebra}, the universality property ensures us that we can write
\begin{equation}\label{eq:universal-property-seqPOVM-gen}
    \theta_{BC}^\lambda (\mathfrak{f}_{bc|yz}) = M_{bc|yz}^\lambda = B^{\lambda,*}_{b|y} (C_{c|z}^\lambda).
\end{equation}

Now consider the subset of elements of the universal C$^*$-algebra $\mathscr{A}_{B\to C}$ from Lemma~\ref{lem:universal_sequential_povm_algebra} obtained by fixing the labels $b$ and $y$; the elements of this set clearly form a subalgebra and satisfy the relations of the universal C$^*$-algebra $\mathscr{A}_C$ from Lemma~\ref{lem:universal_povm_algebra}.
The only non-trivial condition to check is that the sum over $c$ gives the identity.
To see this, note that with the notation $\mathfrak{f}_{b|y} = \sum_{c\in O_C} \mathfrak{f}_{b,c|y,z}$ it holds that
\begin{align*}
    \mathfrak{f}_{b|y} \mathfrak{f}_{b,c|y,z} = \sum_{c'\in O_C} \mathfrak{f}_{b,c'|y,z} \mathfrak{f}_{b,c|y,z} = \mathfrak{f}_{b,c|y,z},
\end{align*}
for all $c$ and $z$, and therefore the projector $\mathfrak{f}_{b|y}$ in $\mathscr{A}_{B\to C}$ acts as the identity in the subalgebra.
Hence, its universality property implies the existence of a *-homomorphism $T_{b|y}$ such that
\begin{equation*}
    T_{b|y}(\mathfrak{e}_{c|z}) = \mathfrak{f}_{bc|yz}.
\end{equation*}
This map is completely positive; furthermore it is sub-unital and $\sum_b T_{b|y} = T_y$ is unital.

Note that this map, together with $B_{b|y}^{\lambda,*}$, intertwines the representations of the two universal algebras on the generators:
\begin{align}\label{eq:comm-univ-rep}
    \theta_{BC}^\lambda(\mathfrak{f}_{bc|yz})
    = \theta_{BC}^\lambda \circ T_{b|y}(\mathfrak{e}_{c|z})
    =B_{b|y}^{\lambda,*}\circ \theta_C^\lambda(\mathfrak{e}_{c|z}),
\end{align}
and given that the maps are *-homomorphisms, the commutation extends to all of the elements of the algebras. Fig.~\ref{fig:univAlgebras-3players} provides a pictorial representation of these commutation relations.

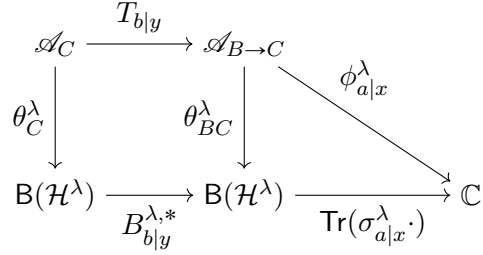
\begin{figure}[htbp]
    \centering
\begin{tikzpicture}
    \node (AC) at (0,2) {\( \mathscr{A}_C \)};
    \node (ABC) at (2.5,2) {\( \mathscr{A}_{B \to C} \)};
    \node (C) at (5.5,0) {\( \mathbb{C} \)};
    \node (BH1) at (0,0) {\( \mathsf{B}(\mathcal{H}^\lambda) \)};
    \node (BH2) at (2.5,0) {\( \mathsf{B}(\mathcal{H}^\lambda) \)};

    \draw[->] (AC) -- (ABC) node[midway, above] {\( T_{b|y} \)};
    \draw[->] (ABC) -- (C) node[midway, above,  yshift=5pt] {\( \phi_{a|x}^\lambda \)};
    \draw[->] (AC) -- (BH1) node[midway, left] {\( \theta^\lambda_{C} \)};
    \draw[->] (ABC) -- (BH2) node[midway, left] {\( \theta^\lambda_{BC} \)};
    \draw[->] (BH1) -- (BH2) node[midway,below] {\( B^{\lambda,*}_{b|y} \)};
    \draw[->] (BH2) -- (C) node[midway,below] {\( \mathsf{Tr}(\sigma_{a|x}^\lambda \cdot) \)};

\end{tikzpicture}
     \caption{Commutative diagram that characterises the algebraic and the quantum tripartite compiled efficient strategies. The upper part of the diagram is purely algebraic, the part below is a concrete realisation in a Hilbert space $\mathcal{H}^\lambda$.
    All directed paths starting from $\mathscr{A}_C$ ending up in $\mathbb{C}$ are equivalent.}
    \label{fig:univAlgebras-3players}
\end{figure}

As a physical interpretation of above abstract constructions, the universal C*-algebra $\mathscr{A}_C$ can be thought of as ``choosing Charlie's labels'', and the maps $T_{b|y}$ can be though of as ``choosing Bob's labels'', moving the entire experiment into the state $\phi^\lambda_{a|x}$ which describes Alice's system.
This is reflected by the $\lambda$-independence of the Charlie's measurements and Bob's maps, and the $\lambda$-dependence of Alice's state.
As the universal objects in $\mathscr{A}_C$ and the universal maps $T_{b|y}$ know nothing about the concrete experiment beyond the labels, but forward all information about the labels $a,b,y,z$ to Alice, one can think of Alice's state as a simulation of all possible branches of the experiment (labelled by $a,b,y,z$), and the elements of $\mathscr{A}_{B\to C}$ of picking the correct branch that corresponds to Bob's and Charlie's choices.

Note that these constructions can be analogously done for POVMs instead of PVMs and for sequential POVMs instead of sequential PVMs. In this case, the maps $T_{b|y}$ between the universal C*-algebras loose their *-homomorphic structure and becomes only completely positive maps instead.
More details on these alternative constructions can be found in Appendix~\ref{app:alternative-proof}.

The correlations arising from tripartite quantum efficient strategies described in Section \ref{subsec:eff-quantum-str} corresponds to the concatenation of the two arrows at the bottom of Fig. \ref{fig:univAlgebras-3players}, which are respectively the adjoint of the instrument $B_{b|y}^{\lambda,*}$ and the trace with the state $\Tr(\sigma_{a|x}^\lambda \cdot)$.
The corresponding algebraic strategy is the concatenation of the two arrows above, i.e. the map between the universal algebras $T_{b|y}$ and the state $\phi_{a|x}^\lambda$.

\begin{definition}[C$^*$-algebraic compiled strategy for $\mathsf{k}=3$, $\mathcal{A}_\lambda^3$] Consider the compiled game $\mathcal{G}_\lambda$. A general algebraic efficient strategy is characterised by 
    \begin{enumerate}
        \item a universal $\text{C}^*$ algebra of PVMs $\mathscr{A}_C$, and a universal $\text{C}^*$ algebra of sequential PVMs $\mathscr{A}_{B\to C}$;
        \item PVMs, that are positive elements of the algebra $\mathfrak{c}_{c|z} \in \mathscr{A}_C^+$, that sums up to identity $\sum_c \mathfrak{c}_{c|z} = \mathfrak{1}$;
        \item a *-homomorphism $T_{b|y}: \mathscr{A}_C \to \mathscr{A}_{B\to C}$, that maps generator to generator;
        \item sub-normalised efficient states, which are positive functional on the algebra $\phi^\lambda_{a|x}:\mathscr{A}_{B\to C}\to\mathbb{C}$, such that $\phi^\lambda_{x} = \sum_a \phi^\lambda_{a|x}$ is a normalised state.
    \end{enumerate}
    All tripartite algebraic strategies have the following form
    \begin{equation*}
        p_\lambda(a,b,c|x,y,z) = \phi^\lambda_{a|x}\circ T_{b|y}(\mathfrak{c}_{c|z}).
    \end{equation*}
\end{definition}

Using the same definition of the algebraic states given in Eq. \ref{eq:def-alg-state}, we can identify again the quantum and algebraic correlations
\begin{equation}
    \phi^\lambda_{a|x}\circ T_{b|y}(\mathfrak{c}_{c|z})
    = \Tr\left[ \sigma^\lambda_{a|x} \theta^\lambda_{BC}(T_{b|y}(\mathfrak{c}_{c|z})) \right]
    = \Tr\left[ \sigma^\lambda_{a|x} B^{\lambda,*}_{b|y} (C_{c|z}^\lambda)\right]
\end{equation}
where in the last line we used the commutation relations of Eq. \ref{eq:comm-univ-rep} and Fig. \ref{fig:univAlgebras-3players}.

It is not difficult to extend these ideas to more than three players, by considering more universal C$^*$-algebras of sequential PVMs. To favour the readability we will comment this extension at the end of the Section.

\subsection{Asymptotic behaviour of efficient algebraic strategies}\label{subsec:asy-eff-strat}

We presented the necessary framework to characterise the set of efficient compiled algebraic strategies in the limit of perfect security $\lambda = \infty$.

The Banach-Alaoglu Theorem ensures that the dual space of a C$^*$-algebra $\mathscr{A}$ is compact, hence the limit exists. More precisely every sequence of possibly sub-normalised states on the algebra will weak$^*$-converge to a limit state, equally sub-normalised.
For the bipartite case, this means that there exist a subsequence $\lambda_k$ such that 
\begin{align*}    \lim_{k\rightarrow\infty}\phi_{a|x}^{\lambda_k}(\mathfrak{a})
    = \phi_{a|x}(\mathfrak{a}) \qquad \forall \mathfrak{a}\in \mathscr{A}_{B}.
\end{align*}
In the case where we have three parties we equivalently have 
\begin{align*}
    \lim_{k\rightarrow\infty} \phi_{a|x}^{\lambda_k}(\mathfrak{a})
    = \phi_{a|x}(\mathfrak{a}) \qquad &\forall \mathfrak{a}\in \mathscr{A}_{B\to C},
\end{align*}
but we can also consider the composition of $\phi_{a|x}^{\lambda} \circ T_{b|y} : \mathscr{A}_C\to \mathbb{C}$ as a state, such that
\begin{align*}
    \lim_{k'\rightarrow\infty} \phi_{a|x}^{\lambda_{k'}} \circ T_{b|y} (\mathfrak{b}) 
    = \phi_{a|x} \circ T_{b|y} (\mathfrak{b}) \qquad &\forall \mathfrak{b}\in \mathscr{A}_C.
\end{align*}

In Theorem \ref{th:3pl-ind-cpa}, we identified some constraints on the tripartite quantum efficient compiled strategies.
These constraints can be expressed in terms of algebraic strategies, with the advantage that the limit of asymptotic security is well defined for algebraic states.
The following theorem characterises precisely the set of asymptotic algebraic compiled strategies.

\begin{theorem}\label{th:asymptotic-constraints-3}
Consider an algebraic strategy for a tripartite compiled game.
In the asymptotic limit $\lambda \to \infty$, the sequence $\phi_{a|x}^\lambda$ weak-* converges to the state $\phi_{a|x}$, such that
\begin{enumerate}
    \item for the limit states $\phi_{a|x}$
    \begin{equation}\label{eq:op-NS-states}
        \sum_a \phi_{a|x}(\mathfrak{a}) = \sum_a\phi_{a|x'}(\mathfrak{a}) \qquad \forall \mathfrak{a} \in \mathscr{A}_{B\to C},
    \end{equation}
    \item for the composite limit strategies $\phi_{a|x}(\mathfrak{l^*} T_{b|y}(\cdot)\mathfrak{r})$
        \begin{equation}\label{eq:op-NS-comp-strat}
        \sum_b \phi_{a|x} \left( \mathfrak{l}^* T_{b|y}(\mathfrak{b}) \mathfrak{r} \right) = 
        \sum_b \phi_{a|x} \left( \mathfrak{l}^* T_{b|y'}(\mathfrak{b}) \mathfrak{r} \right) \qquad \forall \mathfrak{b} \in \mathscr{A}_{C}, \forall \mathfrak{l,r} \in \mathscr{A}_{B \to C}.
    \end{equation}
\end{enumerate}
\end{theorem}

\begin{proof}
1. Consider $\mathfrak{a}$ to be any fixed non-commutative polynomial $P(\{\mathfrak{f}_{bc|yz}\})$ in the generators of $\mathscr{A}_{B \to C}$. Write the limit state as
    \begin{align*}
    \phi_{a|x}(P(\{ \mathfrak{f}_{bc|yz}\})) 
    &= \lim_{k\rightarrow\infty}\phi_{a|x}^{\lambda_k}(P(\{ \mathfrak{f}_{bc|yz}\}))\\
    &= \lim_{k\rightarrow\infty}\Tr\left[\sigma_{a|x}^{\lambda_k}\theta_{BC}^{\lambda_k}(P(\{ \mathfrak{f}_{bc|yz}\}))\right]\\
    &= \lim_{k\rightarrow\infty}\Tr\left[\sigma_{a|x}^{\lambda_k} P(\{B_{b|y}^{\lambda_k,*}(C_{c|z}^{\lambda_k}) \})\right].
\end{align*}
where in the last line we used the commutation properties of Eq. \ref{eq:comm-univ-rep}.
Notice that $P(\{B_{b|y}^{\lambda_k,*}(C_{c|z}^{\lambda_k}) \})$ does not have a physical meaning, but has a QPT block-encoding.
Consider now the difference of two normalised states $\phi_{x} = \sum_a\phi_{a|x}$ evaluated on the same polynomial
\begin{align*}
    & \left| \phi_{x}(P(\{ \mathfrak{f}_{bc|yz}\})) - \phi_{x'}(P(\{ \mathfrak{f}_{bc|yz}\})) \right| \\
    &= \lim_{k \to \infty} \left| \sum_a \Tr\left[\sigma_{a|x}^{\lambda_k} P(\{B_{b|y}^{\lambda_k,*}(C_{c|z}^{\lambda_k}) \})\right] - \sum_a \Tr\left[\sigma_{a|x'}^{\lambda_k} P(\{B_{b|y}^{\lambda_k,*}(C_{c|z}^{\lambda_k}) \}) \right] \right| \\
    & \leq \lim_{k \to \infty} \eta(\lambda_k) =0,
\end{align*}
where in the last line we used the strong IND-CPA constraint derived in Theorem \ref{th:poly-bipartite}.
The polynomials of the generators are a dense subset of the full algebra, hence the asymptotic states for different input labels $x$ are the same on the full algebra
\begin{align*}
    \phi_x(\mathfrak{a})-\phi_{x'}(\mathfrak{a}) = 0 \qquad \forall \mathfrak{a} \in  \mathscr{A}_{B\to C}.
\end{align*}
\\
2. Consider $\mathfrak{b}$ to be any fixed non-commutative polynomial in the generators of $\mathscr{A}_{C}$ and $\mathfrak{l}$ and $\mathfrak{r}$ to be any fixed 
non-commutative polynomial in the generators of $\mathscr{A}_{B \to C}$.
\begin{align*}
    \phi_{a|x} &\left( P'(\{ \mathfrak{f}^*\}) T_{b|y}(P(\{\mathfrak{e}_{c|z} \})) P''(\{ \mathfrak{f}\}) \right) \\
    &= \lim_{k\to \infty} \phi_{a|x}^{\lambda_k}\left( P'(\{ \mathfrak{f}^*\}) T_{b|y}(P(\{\mathfrak{e}_{c|z} \})) P''(\{ \mathfrak{f}\}) \right)\\
    &=\lim_{k\to\infty}\Tr\left[\sigma_{a|x}^{\lambda_k}
    \theta_{BC}^{\lambda_k,*}(P'(\{ \mathfrak{f}\}))
    \theta_{BC}^{\lambda_k}(T_{b|y}(P(\{ \mathfrak{e}_{c|z}\})))
    \theta_{BC}^{\lambda_k}(P''(\{ \mathfrak{f}\}))
    \right]\\
    &= \lim_{k\to\infty}\Tr\left[\sigma_{a|x}^{\lambda_k}
    \theta_{BC}^{\lambda_k,*}(P'(\{ \mathfrak{f}\}))
    B_{b|y}^{\lambda_k,*}(P(\{C_{c|z}^{\lambda_k}\}))
   \theta_{BC}^{\lambda_k}(P''(\{ \mathfrak{f}\}))
    \right]
\end{align*}
where once again we use the commutation properties of Eq. \ref{eq:comm-univ-rep}.
Notice that $\mathcal{L}^{\lambda_k,*}=  \theta_{BC}^{\lambda_k,*}(P'(\{ \mathfrak{f}\}))$ and $\mathcal{R}^{\lambda_k}=  \theta_{BC}^{\lambda_k}(P''(\{ \mathfrak{f}\}))$ have a QPT block-encoding, and we can apply Theorem \ref{th:3pl-ind-cpa} to show that for all $y$ and $y'$:
\begin{align*}
     \Big| \phi_{a|x} &\left( P'(\{ \mathfrak{f}^*\}) T_{y}(P(\{\mathfrak{e}_{c|z})) P''(\{ \mathfrak{f}\}) \right)  
    -  \phi_{a|x} \left( P'(\{ \mathfrak{f}^*\}) T_{y'}(P(\{\mathfrak{e}_{c|z})) P''(\{ \mathfrak{f}\}) \right) \Big|\\
    &= \lim_{k\to \infty} \left|
    \Tr\left[\sigma_{a|x}^{\lambda_k} \mathcal{L}^{\lambda_k,*}  B_{y}^{\lambda_k,*}(P(\{C_{c|z}^{\lambda_k}\})) \mathcal{R}^{\lambda_k}\right]
    - \Tr\left[\sigma_{a|x}^{\lambda_k} \mathcal{L}^{\lambda_k,*}  B_{y'}^{\lambda_k,*}(P(\{C_{c|z}^{\lambda_k}\})) \mathcal{R}^{\lambda_k}\right]
    \right|\\
    & \leq \lim_{k \to \infty} \eta(\lambda_k) =0.
\end{align*}
The polynomials of the generators are dense in the full-algebra, which means that this proves the statement on the full algebras $\mathscr{A}_{C}$ and $\mathscr{A}_{B\to C}$.
\end{proof}

\subsubsection{More than three players}
In this subsection the $\mathsf{n}$ players are labelled with numbers, instead of $A, B$ and $C$. We collect the inputs and the outputs in vectors like $\Vec{v}= (v_1, \dots v_\mathsf{n})$; $\Vec{v}_{[v_i]}$ and $\Vec{v}_{[v_i']}$ are equal vectors, up to the $i$-th element that is respectively $v_i$ and $v_i'$.

We begin by generalising the notion of universal C$^*$-algebra of $\mathsf{n}$ sequential PVMs. The only non-trivial condition to generalise is the last point of Lemma~\ref{lem:universal_sequential_povm_algebra}, which encodes the sequentiality constraint $B \to C$ by imposing that the generators are non-signalling the input of $C$, when we sum over the outputs of $C$.
To generalise this to sequential constraints like $1 \to 2 \dots \to \mathsf{n}$, we impose that for all $j\in[2,\mathsf{n}]$ the generators are non-signalling the input of the last $j$ players, whenever we sum over their corresponding outputs.

\begin{lemma}[Universal C$^*$-algebra of sequential PVMs, for $\mathsf{n}>3$]
Consider $\mathsf{n} \in \mathbb{N}$, and let $I_i, O_i$ be finite sets for all $i \in [\mathsf{n}]$. Let $G = \{ \mathfrak{f}_{a_1\dots a_\mathsf{n}|x_1 \dots x_\mathsf{n}} \; | \; x_i\in I_i, a_i \in O_i, \forall i \in [\mathsf{n}]\}$ be a set of generators; we will use the shorthand notation  $\Vec{a}= (a_1, \dots a_\mathsf{n})$ and  $\Vec{x}= (x_1, \dots x_\mathsf{n})$.
Consider the relations $R$ given as
    \begin{enumerate}
        \item $\mathfrak{f}_{\Vec{a}|\Vec{x}} = \mathfrak{f}_{\Vec{a}|\Vec{x}}^*$ \hfill $\forall \Vec{a}, \Vec{x}$,
        \item $0 \leq \mathfrak{f}_{\Vec{a}|\Vec{x}} \leq 1$ \hfill $\forall \Vec{a}, \Vec{x}$,
        \item $\mathfrak{f}_{\Vec{a}|\Vec{x}}  \mathfrak{f}_{\Vec{a'}|\Vec{x}}= \delta_{a_i,a_i'} \dots \delta_{a_\mathsf{n},a_\mathsf{n}'} \mathfrak{f}_{\Vec{a}|\Vec{x}}$ \hfill $\forall \Vec{a},\Vec{a'}, \Vec{x}$,
        \item $\sum_{i\in[\mathsf{n}]} \sum_{a_i\in O_i} \mathfrak{f}_{\Vec{a}|\Vec{x}} = \mathds{1}$ \hfill $\forall \Vec{x}$,
        \item $\sum_{j \in [i,\mathsf{n}]} \sum_{a_j \in O_j} \mathfrak{f}_{\Vec{a}| \Vec{x}_{[x_i \dots x_n]}} = \sum_{j \in [i,\mathsf{n}]} \sum_{a'_j \in O_j} \mathfrak{f}_{\Vec{a'}| \Vec{x}_{[x'_i \dots x'_n]}}$ \hfill $\forall \Vec{x}, \Vec{a}, \Vec{a'}, \forall i \in [2, \mathsf{n}]$.
    \end{enumerate}
    Then, $C^*(G,R)$ exists and we call it the sequential universal PVM C$^*$-algebra $\mathscr{A}^\text{PVM}_{1\to 2 \dots \to \mathsf{n}}$.
    In the following, we will also use the shorthand notation $\mathscr{A}_{1, \mathsf{n}}$.
\end{lemma}
\begin{proof}
    To see that there exists a representation of $(G,R)$, we construct a concrete realisation of these generators and relations as bounded operators on Hilbert space. 

    Consider any PVM $\{M_{a_\mathsf{n}|x_\mathsf{n}}\} \subseteq \mathsf{B}(\mathcal{H}_\mathsf{n})$ on some Hilbert space $\mathcal{H}_\mathsf{n}$ with inputs $I_\mathsf{n}$ and outputs $O_\mathsf{n}$, and efficient quantum instruments $\Tilde{I}_{a_i|x_i}^{(i)} : \mathsf{B}(\mathcal{H}_i) \to \mathsf{B}(\mathcal{H}_{i+1})$ with inputs $I_i$ and outputs $O_i$, $i \in [1, \mathsf{n}-1]$. Recall that $\Tilde{I}_{a_i|x_i}^{(i),*}$ is a *-homomorphism; a simple calculation then shows that the operators $\left\{\Tilde{I}_{a_1|x_1}^{(1),*} \circ \dots \circ \Tilde{I}_{a_{\mathsf{n}-1}|x_{\mathsf{n}-1}}^{({\mathsf{n}-1}),*} (M_{a_\mathsf{n}|x_\mathsf{n}})\right\} \in \mathsf{B}(\mathcal{H}_1)$ give rise to a representation of $(G,R)$. 

    Then the proof follows as in Lemma~\ref{lem:universal_sequential_povm_algebra}.
\end{proof}

Consider the universal C$^*$-algebra of sequential PVMs $\mathscr{A}_{2,\mathsf{n}}$ and $\mathscr{A}_{1,\mathsf{n}}$.
Consider the sub-algebra of $\mathscr{A}_{1,\mathsf{n}}$ generated by the elements that have fixed input $\bar{x}_1$ and output $\Bar{a}_1$ for the player $1$; this satisfies all of the conditions of $\mathscr{A}_{2,\mathsf{n}}$.
Then, the universality property implies the existence of a $*$-homomorphism $T^{(1)}_{\Bar{a}_1|\Bar{x}_1} : \mathscr{A}_{2,\mathsf{n}} \to \mathscr{A}_{1, \mathsf{n}}$.

With these techniques, we can properly define the algebraic strategies for more than three players.

\begin{definition}[C$^*$-algebraic compiled strategy for $\mathsf{k}$ players, $\mathcal{A}_\lambda^\mathsf{k}$] Consider the compiled game $\mathcal{G}_\lambda$ for $\mathsf{k}$ players. A general algebraic efficient strategy is characterised by 
    \begin{enumerate}
        \item a universal $\text{C}^*$ algebra of PVMs $\mathscr{A}_\mathsf{k}$, and $\mathsf{k}-2$ universal $\text{C}^*$ algebras of sequential PVMs $\mathscr{A}_{j,\mathsf{k}}$, with $j\in[2, \mathsf{k}-1]$
        \item PVMs, that are positive elements of the algebra $\mathfrak{a}_{a_\mathsf{k}|x_\mathsf{k}} \in \mathscr{A}_\mathsf{k}^+$, that sums up to identity $\sum_{a_\mathsf{k}} \mathfrak{a}_{a_\mathsf{k}|x_\mathsf{k}} = \mathfrak{1}$;
        \item $\mathsf{k}-2$ *-homomorphism, labelled by $j \in [2, \mathsf{k}-1]$, that we call $T_{a_j|x_j}^{(j)}: \mathscr{A}_{j+1,\mathsf{k}} \to \mathscr{A}_{j,\mathsf{k}}$, and map generator to generator;
        \item sub-normalised efficient states, which are positive functional on the algebra $\phi^\lambda_{a_1|x_1}:\mathscr{A}_{2,\mathsf{k}}\to\mathbb{C}$, such that $\phi^\lambda_{x} = \sum_{a_1} \phi^\lambda_{a_1|x_1}$ is a normalised state.
    \end{enumerate}
    All $\mathsf{k}$-partite algebraic strategies have the following form
    \begin{equation*}
        p_\lambda(\Vec{a}|\Vec{x}) = \phi^\lambda_{a_1|x_1}
        \circ T^{(2)}_{a_2|x_2} 
        \circ \dots 
        \circ T^{(\mathsf{k}-1)}_{a_{\mathsf{k}-1}|x_{\mathsf{k}-1}}(\mathfrak{a}_{a_\mathsf{k}|x_\mathsf{k}}).
    \end{equation*}
\end{definition}

To help with the intuition, in Fig. \ref{fig:univAlgebras-4players} there is a representation of the case of $4$ players.
\begin{figure}[htbp!]
    \centering
\begin{tikzpicture}
    \node (D) at (-2.5,2) {$ \mathscr{A}_4 $};
    \node (CtoD) at (0,2) {$ \mathscr{A}_{3,4} $};
    \node (BH0) at (-2.5,0) {$ \mathsf{B}(\mathcal{H}^\lambda) $};
    \node (BH1) at (0,0) {$ \mathsf{B}(\mathcal{H}^\lambda) $};

    \node (BtoCtoD) at (2.5,2) {$ \mathscr{A}_{2,4} $};
    \node (C) at (5,0) {$ \mathbb{C} $};
    \node (BH2) at (2.5,0) {$ \mathsf{B}(\mathcal{H}^\lambda) $};

    \draw[->] (D) -- (CtoD) node[midway, above] {$ T^{(3)}_{a_3|x_3} $};
    \draw[->] (CtoD) -- (BtoCtoD) node[midway, above] {$ T^{(2)}_{a_2|x_2} $};
    \draw[->] (D) -- (BH0) node[midway, left] {$ \theta^\lambda_{4} $};
    \draw[->] (CtoD) -- (BH1) node[midway, left] {$ \theta^\lambda_{3,4} $};
    \draw[->] (BH0) -- (BH1) node[midway, below] {$ I^{(3),\lambda,*}_{a_3|x_3} $};

    \draw[->] (BtoCtoD) -- (C) node[midway, above, xshift=3pt] {$ \phi_{a_1|x_1}^\lambda $};
    \draw[->] (BtoCtoD) -- (BH2) node[midway, left] {$ \theta^\lambda_{2,4} $};
    \draw[->] (BH1) -- (BH2) node[midway, below] {$ I^{(2),\lambda,*}_{a_2|x_2} $};
    \draw[->] (BH2) -- (C) node[midway, below] {$ \mathsf{Tr}(\sigma_{a_1|x_1}^\lambda \cdot) $};

\end{tikzpicture}
     \caption{A commuting diagram representing the structure of the universal C*-algebras of sequential PVMs in the case of 4 players.}
    \label{fig:univAlgebras-4players}
\end{figure}

\begin{figure}[htbp!]
    \centering
\begin{tikzpicture}
    \node (Ak) at (0,2) {\( \mathscr{A}_\mathsf{k} \)};
    \node (Akm1) at (3.5,2) {\( \mathscr{A}_{\mathsf{k}-1,\mathsf{k}} \)};
    \node (A2) at (7,2) {$\dots$};
    \node (A1) at (10.5,2) {\( \mathscr{A}_{2,\mathsf{k}}\)};
    \node (BHk) at (0,0) {\( \mathsf{B}(\mathcal{H}_\mathsf{k}^\lambda) \)};
    \node (BHkm1) at (3.5,0) {\( \mathsf{B}(\mathcal{H}_{\mathsf{k}-1}^\lambda) \)};
    \node (BH2) at (7,0) {$\dots$};
    \node (BH1) at (10.5,0) {\( \mathsf{B}(\mathcal{H}_{2}^\lambda) \)};
    \node (C) at (13.5,0) {\( \mathbb{C} \)};

    \draw[->] (Ak) -- (Akm1) node[midway, above] {\( T^{(\mathsf{k})}_{a_\mathsf{k}|x_\mathsf{k}} \)};
    \draw[->] (Akm1) -- (A2) node[midway, above] {\( T^{(\mathsf{k}-1)}_{a_{\mathsf{k}-1}|x_{\mathsf{k}-1}} \)};
    \draw[->] (A2) -- (A1) node[midway, above] {\( T^{(2)}_{a_{2}|x_{2}} \)};
    \draw[->] (A1) -- (C) node[midway, above,  yshift=5pt] {\( \phi_{a_1|x_1}^\lambda \)};
    \draw[->] (Ak) -- (BHk) node[midway, left] {\( \theta^\lambda_{\mathsf{k}} \)};
    \draw[->] (Akm1) -- (BHkm1) node[midway, left] {\( \theta^\lambda_{\mathsf{k}-1,\mathsf{k}} \)};
    \draw[->] (A1) -- (BH1) node[midway, left] {\( \theta^\lambda_{2,\mathsf{k}} \)};
    \draw[->] (BHk) -- (BHkm1) node[midway,below] {\( I^{(\mathsf{k})}_{a_\mathsf{k}|x_\mathsf{k}} \)};
    \draw[->] (BHkm1) -- (BH2) node[midway,below] {\( B^{(\mathsf{k}-1)}_{a_{\mathsf{k}-1}|x_{\mathsf{k}-1}} \)};
    \draw[->] (BH2) -- (BH1) node[midway,below] {\( B^{(2)}_{a_{2}|x_{2}} \)};
    \draw[->] (BH1) -- (C) node[midway,below] {\( \mathsf{Tr}(\sigma_{a_1|x_1}^\lambda \cdot) \)};

\end{tikzpicture}
     \caption{A commuting diagram representing the structure of the universal C*-algebras of sequential PVMs in the case of $\mathsf{k}$ players.}
    \label{fig:univAlgebras-k-players}
\end{figure}

We now have the necessary tools to generally characterise the set of asymptotic algebraic strategies for $\mathsf{k}$-partite games.

\begin{theorem}\label{th:asymptotic-constraints-k}
Consider an algebraic strategy for a $\mathsf{k}$-partite compiled game.
In the asymptotic limit $\lambda \to \infty$, the sequence of (sub-normalised) states $\phi_{a_1|x_1}^\lambda$ weak-* converges to the (sub-normalised) state $\phi_{a_1|x_1}$, such that
\begin{enumerate}
    \item for the limit states $\phi_{a_1|x_1}$
    \begin{equation}\label{eq:op-NS-states-k}
        \sum_{a_1} \phi_{a_1|x_1}(\mathfrak{a}) = \sum_a\phi_{a|x'}(\mathfrak{a}) \qquad \forall \mathfrak{a} \in \mathscr{A}_{2,\mathsf{k}},
    \end{equation}
    \item for every $j\in[2,\mathsf{k}-1]$,
    the composite limit strategies satisfy the following property
    for all $\mathfrak{l},\mathfrak{r} \in \mathscr{A}_{j,\mathsf{k}}$ and for all $\mathfrak{b} \in \mathscr{A}_{j+1,\mathsf{k}}$
    \begin{equation}\label{eq:op-NS-comp-strat-k}
        \sum_{a_j} \phi_{a_1|x_1} \circ \dots \circ T^{(j-1)}_{a_{j-1}|x_{j-1}}
        \left(\mathfrak{l^*} T^{(j)}_{a_j|x_j}(\mathfrak{b})\mathfrak{r}\right) = 
       \sum_{a_j} \phi_{a_1|x_1} \circ \dots \circ T^{(j-1)}_{a_{j-1}|x_{j-1}}
        \left(\mathfrak{l^*} T^{(j)}_{a_j|x'_j}(\mathfrak{b})\mathfrak{r}\right).
    \end{equation}
\end{enumerate}
\end{theorem}
\begin{proof}
    It is sufficient to combine the proof structure of Theorem \ref{th:asymptotic-constraints-3}, and the constraints coming from Theorem 
    \ref{th:kpl-ind-cpa}.
\end{proof}
 
\newpage
\section{Asymptotic quantum soundness of compiled multipartite non-local games}\label{sec:asymptotic-bound}
In this section we go back to the original motivation of this work: the quantum soundness of the KLVY compiler introduced in Sec. \ref{sec:compiled-non-local-games}, with a focus on multipartite non-local games.

In Section \ref{sec:quantum-compiler} we characterised the quantum strategies of a computationally bounded quantum (QPT) prover $\mathcal{Q}_\text{comp}^\mathsf{k}$,
by identifying in Theorem \ref{th:kpl-ind-cpa} a set of constraints that can be directly deduced from the security of the QFHE scheme. 
In Section \ref{sec:algebraic-strategies} we introduced the C$^*$-algebraic framework, that allows to properly define the asymptotic regime of perfect security.
In this regime, the set of constraints coming from the security of the cryptography simplifies, as reported in Theorem \ref{th:asymptotic-constraints-k}.

In this section, in particular in Section \ref{subsec:connection-compiled}, we show that these constraints are sufficient to apply the generalisation of the chained Radon-Nikodym Theorem that we developed in Section \ref{sec:chain-rule-rn}. The bound on the asymptotic quantum value of compiled games is a direct consequence.

For a more pedagogical exposition, as an instructive and intermediate step, we start presenting a new class of games, that we call sequential games with operational-non-signalling. 
These provide a very intuitive and natural use-case for the chain rule for the Radon-Nikodym theorem.
We show that the set of correlations they produce is equivalent to the set of correlations produced by $\mathsf{k}$-players commuting-operator strategies.
We believe that the definition of operational-non-signalling and its connection to commuting operator strategies might be of independent interested, hence we structured \ref{subsec:preliminaries-sequential-algebraic} and \ref{subsec:operationally-no-sig-commuting-operator} to be as independent as possible from the rest of the paper.

\subsection{Preliminaries on sequential algebraic correlations}%
\label{subsec:preliminaries-sequential-algebraic}

In this subsection, we focus on correlations generated by a fixed directed sequence of communicating players, referred to as sequential correlations.
These are a strictly more general class than non-local correlations, where communication between players is not allowed.
However, not all communication is allowed in the sequential setting;  from the fixed order of the players there is a clear notion of future and past, and players in the future cannot communicate to the past.

It is an interesting question to understand the necessary and sufficient constraints that strategies in a sequential setting have to satisfy to also admit a non-local explanation.
In the case of two players, the constraints are well understood in the finite dimensional case \cite{Wright_2023}, and have recently been extended to infinite dimensions in \cite{KMPSW24bound}.
However, much less is known when considering three players. For instance, an attempt to generalise the bipartite result was presented in \cite{postquantum-steering}, but the constraints they proposed resulted to be too weak to guarantee a correspondence with tripartite non-local correlations.
 
We restrict our attention to the case of sequential quantum correlations for $\mathsf{k}$ players.
We propose a new constraint, that we call operational-non-signalling, and we show that sequential strategies that satisfy this have a $\mathsf{k}$-partite non-local model.
Specifically, we work in the C*-algebraic framework, that correctly characterises finite and infinite dimensional Hilbert spaces.

A sequential game is characterised by a fixed directed sequence of $\mathsf{k}$ players, which are allowed to communicate with their next neighbour. More precisely, every player receives an independent classical input $\{x_i\}_{[\mathsf{k}]}$, and we allow one-way communication between players $i$ and $i+1$, for $i \in [1,\mathsf{k}-1]$.
At the end, every player has to produce a classical output $\{a_i\}_{[\mathsf{k}]}$.

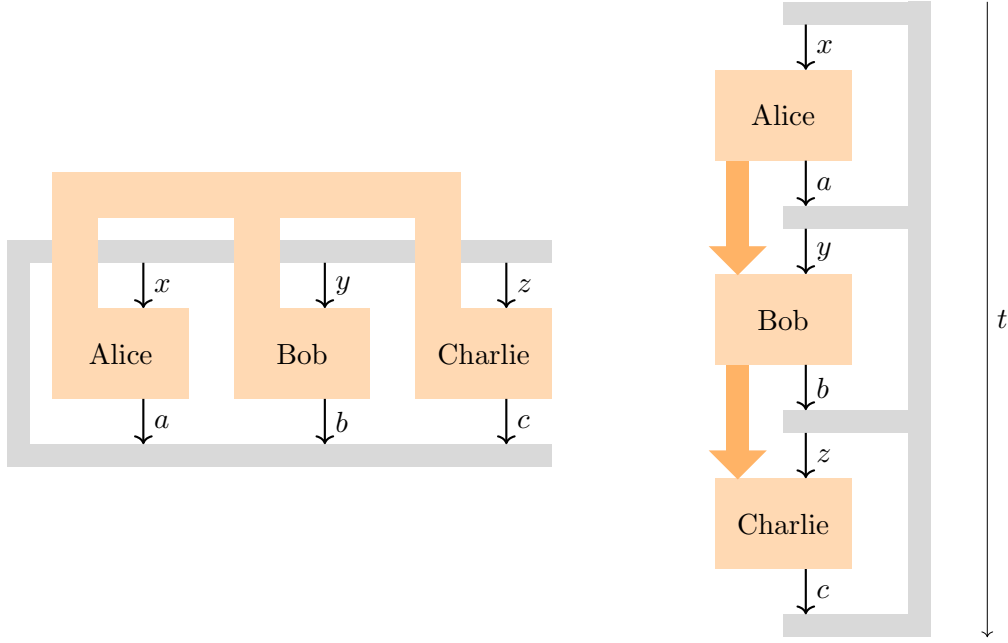
\begin{figure}[ht!]
  \centering
  \begin{minipage}{0.65\textwidth}
    \centering
\begin{tikzpicture}[scale=0.3]
    \fill[gray!30] (0,0) -- (24,0) -- (24,1) -- (0,1) -- cycle;
    \fill[gray!30] (0,0) -- (1,0) -- (1,10) -- (0,10) -- cycle;
    \fill[gray!30] (0,9) -- (24,9) -- (24,10) -- (0,10) -- cycle;

    \fill[orange!30] (4,3) rectangle (8,7);
    \fill[orange!30] (12,3) rectangle (16,7);
    \fill[orange!30] (20,3) rectangle (24,7);

    \draw[->, thick] (6,9) -- (6,7) node[right,midway] {$x$};
    \draw[->, thick] (14,9) -- (14,7)node[right,midway] {$y$};
    \draw[->, thick] (22,9) -- (22,7) node[right,midway] {$z$};
    
    \draw[->, thick] (6,3) -- (6,1)node[right,midway] {$a$};
    \draw[->, thick] (14,3) -- (14,1)node[right,midway] {$b$};
    \draw[->, thick] (22,3) -- (22,1)node[right,midway] {$c$};

    \fill[orange!30] (2,11) -- (20,11) -- (20,13) -- (2,13) -- cycle;
    \fill[orange!30] (2,3) -- (4,3) -- (4,11) -- (2,11) -- cycle;
    \fill[orange!30] (10,3) -- (12,3) -- (12,11) -- (10,11) -- cycle;
    \fill[orange!30] (18,3) -- (20,3) -- (20,11) -- (18,11) -- cycle;

    \node at (5,5){Alice};
    \node at (13,5) {Bob};
    \node at (21,5) {Charlie};

\end{tikzpicture}   \end{minipage}
  \hfill
  \begin{minipage}{0.3\textwidth}
    \centering

\begin{tikzpicture}[scale=0.3]
    \fill[gray!30] (8.5,0) -- (9.5,0) --(9.5,28) -- (8.5,28) -- cycle;
    \fill[gray!30] (3,0) -- (8.5,0) --(8.5,1) -- (3,1) -- cycle;
    \fill[gray!30] (3,9) -- (8.5,9) --(8.5,10) -- (3,10) -- cycle;
    \fill[gray!30] (3,18) -- (8.5,18) --(8.5,19) -- (3,19) -- cycle;
    \fill[gray!30] (3,27) -- (8.5,27) --(8.5,28) -- (3,28) -- cycle;

    \fill[orange!30] (0,3) rectangle (6,7);
    \fill[orange!30] (0,12) rectangle (6,16);
    \fill[orange!30] (0,21) rectangle (6,25);

    \fill[orange!60] (0.5,12) rectangle (1.5,8);
    \node[isosceles triangle,
	isosceles triangle apex angle=90,
	fill=orange!60,
        rotate=-90,
	minimum size =5] (T90)at (1,7.8){};

        \fill[orange!60] (0.5,21) rectangle (1.5,17);
    \node[isosceles triangle,
	isosceles triangle apex angle=90,
	fill=orange!60,
        rotate=-90,
	minimum size =5] (T90)at (1,16.8){};

    \node at (3,5){Charlie};
    \node at (3,14) {Bob};
    \node at (3,23) {Alice};

    \draw[->, thick] (4,27) -- (4,25) node[right,midway] {$x$};
    \draw[->, thick] (4,18) -- (4,16)node[right,midway] {$y$};
    \draw[->, thick] (4,9) -- (4,7) node[right,midway] {$z$};
    
    \draw[->, thick] (4,21) -- (4,19)node[right,midway] {$a$};
    \draw[->, thick] (4,12) -- (4,10)node[right,midway] {$b$};
    \draw[->, thick] (4,3) -- (4,1)node[right,midway] {$c$};

    \draw[->, very thin] (12,28) -- (12,0)node[right,midway] {$t$};

	isosceles triangle apex angle=90,
	fill=orange!30,
        rotate=-90,
	minimum size =5] (T90)at (6,0){};
    
\end{tikzpicture}
   \end{minipage}
    \caption{On the left, the standard tripartite non-local structure. On the right, a tripartite sequential game, where the parties are connected with a one-way unbounded quantum channel. Time is flowing downwards.}
    \label{fig:operaional-ns}
\end{figure}

In the simplest case of $\mathsf{k}=2$, the first player receives a classical input $x$ and produce a classical output $a$ and a quantum state $\phi_{a|x}$, which she will send to the second player. These states are sub-normalised, such that the sum over the output label is a correctly normalised state; these set of states are called assemblages.
The second player receives this quantum state, and can perform a measurement that depends on his classical input $y$ to generate a classical output $b$.
\begin{definition}[C$^*$-algebraic sequential correlations with $\mathsf{k}=2$, $\mathcal{A}_{seq}^2$ ]
A bipartite algebraic sequential strategy is specified by 
\begin{enumerate}
    \item a C$^*$ algebra $\mathscr{A}$;
    \item measurements $\mathfrak{b}_{b|y} \in \mathscr{A}_+$, such that $\sum_b \mathfrak{b}_{b|y} = \mathfrak{1}$ ;
    \item sub-normalised states $\phi_{a|x} : \mathscr{A} \to \mathbb{C}$, such that $\sum_a \phi_{a|x}(\mathfrak{a}) =\phi_x(\mathfrak{a})$ is a normlised state.
\end{enumerate}
All bipartite C$^*$-algebraic sequential correlations can be written in the following way : 
    \begin{equation*}
        p(a b|x y) = \phi_{a|x} (\mathfrak{b}_{b|y}).
    \end{equation*}
\end{definition}

Quantum instruments are a powerful and compact tool to characterise correlations generated by more than $2$ players. More precisely, the behaviour of the agents in between the first and the last one can be modelled as a set of quantum instruments, i.e. collections of completely-positive (CP) trace-non-increasing maps $\{I_a\}_a$, whose sum is a completely-positive trace-preserving (CPTP) map.
These properties are important to enforce that the valid assemblages of quantum states are mapped to valid quantum assemblages.
In the algebraic framework it is more natural to consider the adjoint of these maps; notice that the adjoint of a trace-preserving map is a unital map.
Hence, instead of the standard quantum instruments we are going to consider set of sub-unital CP maps $\{T_b\}_b$, whose sum is a unital CP map $T$; we will refer to this as transformations.

Let us explicitly write the case of $3$ players, which is represented in Fig. \ref{fig:operaional-ns}.
\begin{definition}[C$^*$-algebraic sequential correlations with $\mathsf{k}=3$, $\mathcal{A}_{seq}^3$]
A tripartite algebraic sequential strategy is specified by 
\begin{enumerate}
    \item two C$^*$-algebras $\mathscr{A}$ and $\mathscr{B}$;
    \item measurements $\{ \mathfrak{c}_{c|z} \}_{c,z} \in \mathscr{B}_+$;
    \item CP sub-unital maps $T_{b|y} : \mathscr{B} \to \mathscr{A}$, such that $\sum_b T_{b|y} (\mathfrak{b}) = T_y(\mathfrak{b})$ are unital;
    \item sub-normalised states $\phi_{a|x} : \mathscr{A} \to \mathbb{C}$.
\end{enumerate}
All tripartite C$^*$-algebraic sequential correlations can be written in the following way : 
    \begin{equation*}
        p(a b c|x y z) = \phi_{a|x} (T_{b|y}(\mathfrak{c}_{c|z})).
    \end{equation*}
\end{definition}
To generalise to $\mathsf{k}>3$ players, it is sufficient to consider $\mathsf{k}-1$ C$^*$-algebras and $\mathsf{k}-2$ maps $T$ between the algebras.

Every non-local strategy can be realised with a sequential strategy; instead of performing remote-state preparation using an entangled resource, the players can directly communicate quantum states.
The opposite is clearly not true, because a general sequential strategy can be signalling.
For instance, a player is allowed to communicate its input to the following one, and the output of the second can depend non-trivially on both inputs.

\subsection{Sequential games with operational-non-signalling}%
\label{subsec:operationally-no-sig-commuting-operator}
Which constraints must be imposed on sequential strategies to ensure that the resulting correlations admit a non-local counterpart?
To address this question, we introduce the notion of non-signalling operations, tailored to assemblages and transformations in the sequential scenario.
Readers familiar with contextuality will find many similarities with preparation and transformation equivalences for quantum models \cite{Spekkens_2005}.

\begin{definition}[Operational-non-signalling assemblage $\phi_{a|x}$]\label{def:op-ns-states}
A set of sub-normalised states $\{\phi_{a|x} \}_{a,x} \in \mathsf{P}(\mathscr{A})$ is said to be operational-non-signalling if summing over the output label $a$, the dependency on the input label $x$ is lost; equivalently, if there exists a state $\phi \in \mathsf{S}(\mathscr{A})$ such that 
    \begin{equation*}
        \sum_a \phi_{a|x} = \phi \qquad \forall x.
    \end{equation*}
\end{definition}

In the steering community, assemblages that satisfy this operational-non-signalling property are simply called non-signalling \cite{postquantum-steering}.

\begin{definition}[Operational-non-signalling transformation $T_{b|y}$]\label{def:op-ns-transformations}
A set of sub-unital maps $\{T_{b|y}\}_{b,y} \in \mathsf{CP}(\mathscr{A}, \mathscr{B})$ is operational-non-signalling if summing over the output label $b$, the dependency on the input label $y$ is lost; equivalently, if there exists a unital CP map $T \in \mathsf{UCP}(\mathscr{A}, \mathscr{B})$ such that
        \begin{equation*}
        \sum_b T_{b|y} = T \qquad \forall y.
    \end{equation*}
\end{definition}

Consider the sequential structure on the right of Fig. \ref{fig:operaional-ns}. We say that it is producing operationally-non-signalling correlations if the quantum states that Alice outputs satisfy Def. \ref{def:op-ns-states} and the quantum operations that Bob performs satisfy Def. \ref{def:op-ns-transformations}.

In the following we show that operational-non-signalling states and transformations are the minimal and sufficient condition to identify the subset of sequential correlations that have a non-local counterpart.
We first present known results for the bipartite case, in finite and infinite dimensions. Then, we present our original result for three players, which can be extended to any $\mathsf{k}>3$.

\subsubsection{Two players: S-G-HJW Theorem and Radon-Nikodym derivatives}
For $2$ players in finite dimensions, operational-no-signalling assemblages imply non-local strategies.
This important result can be formulated in many ways, and it has been interestingly rediscovered many times in slightly different forms; it is usually referred to as the S-G-HJW theorem.
Let us state it in its standard formulation commonly used in the steering literature.

\begin{theorem}[S-G-HJW Theorem \cite{schrodingerhjwtheorem}]\label{th:SHJW}
    A non-signalling assemblage $\{\sigma_{a|x}\}_{a,x}$ on a finite-dimensional Hilbert space $\mathcal{H}$ always has a bipartite quantum realisation, i.e. there exist an auxiliary Hilbert space $\mathcal{H}_A$, a POVM on this space $M_{a|x} \in \mathsf{B}(\mathcal{H}_A)$ and a state $\rho \in \mathsf{S}(\mathcal{H}_A \otimes \mathcal{H})$ such that 
    \begin{equation*}
        \sigma_{a|x} = tr_{\mathcal{H}_A} \left[(M_{a|x} \otimes \mathds{1}_\mathcal{H}) \rho\right].
    \end{equation*}
\end{theorem}

What in the steering scenario is a remote state preparation, in the sequential scenario simply becomes an assemblage of states communicated through a quantum channel.
\begin{corollary}[From \cite{Wright_2023} and \cite{KMPSW24bound}]
    A $2$-players quantum sequential strategy in finite dimension for which the states prepared $\{\sigma_{a|x}\}_{a,x}$ satisfy the non-signalling assemblage condition $\sum_a \sigma_{a|x}= \sigma \in \mathsf{S}(\mathcal{H})$, always have a $2$-player non-local quantum realisation.
\end{corollary}
\begin{proof}
    Start with a general quantum sequential correlation. Using Theorem \ref{th:SHJW} rewrite the non-signalling assemblage as the partial trace of a bigger state living in $\mathcal{H}_A \otimes \mathcal{H}$, where $\mathcal{H}_A$ is the auxiliary Hilbert space. Then simply pull out the partial trace
    \begin{align*}
        tr_{\mathcal{H}}(N_{b|y} \sigma_{a|x}) 
        &= tr_{\mathcal{H}}(N_{b|y} tr_{\mathcal{H}_A} \left[(M_{a|x} \otimes \mathds{1}_\mathcal{H}) \rho\right]) \\
        &= tr_{\mathcal{H}_A \mathcal{H}} \left[(M_{a|x} \otimes N_{b|y}) \rho \right]
    \end{align*}
    to obtain a bipartite quantum non-local strategy as in Def. \ref{def:nl-2}
\end{proof}

In the algebraic framework, the S-G–HJW Theorem can be seen as a consequence of the Radon-Nikodym Theorem (theorem~\ref{th:rn}). This theorem comes in many different flavours; it is usually cited as an important result in measure theory, theorem~\ref{th:rn} states its version for positive linear functionals.

The similarity with the S-G–HJW Theorem is evident if we consider finite decompositions of positive functionals $\phi= \sum_i \phi_i$: every addend is dominated by the sum, hence they can all be represented with the GNS construction (Theorem \ref{th:GNS}) of the sum, at the cost of considering the Radon-Nikodym derivatives $D_i$. 
The set of Radon-Nikodym derivatives $\{D_i\}_i$ can be interpreted as Alice's measurement.
This is formalised by Lemma~\ref{lem:RN-sums}.
Notice that in the algebraic framework, commuting-operator strategies naturally arise, instead of the tensor product structure.
\begin{corollary}[From \cite{KMPSW24bound}] %
    The set of correlations produced by 2-players sequential strategies with operation-non-signalling is equal to the set of correlations produced by commuting operator strategies.
\end{corollary}
\begin{proof}
    Consider $\phi = \sum_a \phi_{a|x}$ for all $x$.
    For every $\phi_{a|x}$, apply Radon-Nikodym (Theorem \ref{th:rn}) for $\phi_{a|x} \leq \phi$:
    \begin{align*}
        \phi_{a|x} (\mathfrak{b}_{b|y}) 
        &= \bra{\Omega_\phi} A_{a|x} \pi_\phi(\mathfrak{b}_{b|y})\ket{\Omega_\phi}.
    \end{align*}
    The commutation is directly given by Radon-Nikodym.
    Use Lemma \ref{lem:RN-sums} to show that the new operators $A_{a|x}$ are positive and sum to the identity, i.e. are valid POVMs. The representation $\pi_\phi$ is a $*$-homomorphism, hence $\pi_\phi(\mathfrak{b}_{b|y})$ is also a valid POVM.
    This proves that we obtain a bipartite commuting operator strategy as in Def. \ref{def:co-2}.
\end{proof}

\subsubsection{Three players: an application of the chain rule for Radon-Nikodym derivatives}

Consider now a tripartite experiment with operationally-non-signalling states $\phi_{a|x}$ and transformations $T_{b|y}$.
Intuitively, composing together operationally-non-signalling states and transformations, we aim to retrieve a tripartite commuting operator strategy.

A naive first approach is to sequentially apply Theorem \ref{th:rn} for the states and Theorem \ref{thm:rn-for-cp-maps} for the transformations 
\begin{equation*}
    \phi_{a|x}(T_{b|y}(\mathfrak{c}_{c|z})) = \bra{\Omega_\phi} A_{a|x} V_T^* B_{b|y} \pi_T(\mathfrak{c}_{c|z})V_T\ket{\Omega_\phi}.
\end{equation*}
The Radon-Nikodym derivatives $A_{a|x} \in \mathsf{B}(\mathcal{H}_\phi)$ and $B_{b|y} \in \mathsf{B}(\mathcal{H}_T)$ can be interpreted as POVMs, together with $\pi_T(\mathfrak{c}_{c|z})\in \mathsf{B}(\mathcal{H}_T)$.
We also have the following commutation relations between them 
\begin{align*}
    [B_{b|y}, \pi_T(\mathfrak{c}_{c|z})] =0, \qquad [A_{a|x}, V_T^* B_{b|y} \pi_T(\mathfrak{c}_{c|z})V_T]=0.
\end{align*}
However, to prove that they all pairwise commute is not trivial. What makes this hard is that the POVMs live in different Hilbert spaces, connected by the isometry $V_T$.

This is the same problem that arised, and that we solved, in the proof of the chained rule for Radon-Nikodym (Theorem \ref{th:chain-2-RN}).
In other words, a very natural application of the chained rule for Radon-Nikodym is to show that sequential tripartite strategies with operationally-non-signalling states and transformations are equivalent to tripartite commuting operator strategies.
\begin{corollary}\label{cor:3seqplayers}
    The set of correlations produced by 3-player sequential operationally-non-signalling strategies is equal to the set of correlations produced by 3-player commuting operator strategies.
\end{corollary}
\begin{proof}
    Consider general tripartite sequential correlations
    \begin{align*}
        p(abc|xyz) = \phi_{a|x} \circ T_{b|y} (\mathfrak{c}_{c|z})
    \end{align*}
    where $\{ \phi_{a|x}\}_{a,x}$ is an operationally-non-signalling assemblage of states, with $\sum_a \phi_{a|x} = \phi$,  and $\{T_{b|y}\}_{b,y}$ are operationally-non-signalling transformations, with $\sum_b T_{b|y} = T$.
    We then apply the version of the chain rule for sequential Radon-Nikodym derivatives formulated in Lemma~\ref{lemma:chain-2-rn-sums}, yielding 
    \begin{align*}
        p(abc|xyz)= \bra{\Omega} F_{a|x} F_{b|y} \pi(\mathfrak{c}_{c|z}) \ket{\Omega},
    \end{align*}
    taking into account that $\phi_{a|x}, \phi$ are functionals and thus Remark~\ref{remark:rn-2-chain-functional} applies. Note, that $\ket{\Omega} \in \mathcal{K}$ and $\|\ket{\Omega}\| = 1$.
    Collecting all the commutation relations we get
    \begin{align*}
        [F_{a|x}, F_{b|y}] = [F_{a|x}, \pi(\mathfrak{c}_{c|z})] = [F_{b|y}, \pi(\mathfrak{c}_{c|z})] = 0, \qquad \forall a,x,b,y,c,z.
    \end{align*}
    In order for this to be a commuting operator strategy, we also need to check that the operators form valid sets of measurements.
    Consider $\pi(\mathfrak{c}_{c|z}) \in \mathsf{B}(\mathcal{K}) $; we know that $\{ \mathfrak{c}_{c|z}\}_{c,z}$ is a valid set of measurements, and $\pi(\cdot)$ is a $*$-homomorphism, therefore $\{\pi(\mathfrak{c}_{c|z}) \}_{c,z}$ is a set of measurements on $\mathsf{B}(\mathcal{K})$.
    Lemma~\ref{lemma:chain-2-rn-sums} further implies that $\{F_{a|x}\}_{a,x}$ and $\{F_{b|y}\}_{b,y}$ are describing valid POVMs on $\mathsf{B}(\mathcal{K})$.
    
    This shows that the POVMs $\{F_{a|x}\}_{a,x}$, $\{F_{b|y}\}_{b,y}$, and $\{\pi(\mathfrak{c}_{c|z}) \}_{c,z}$, together with the resource state $\ket{\Omega}$ defines a commuting operator strategy for three parties.
\end{proof}

\subsubsection{More than three players}
Let us start with a more precise characterisation of algebraic sequential correlations in the case of more than three provers.
\begin{definition}[C$^*$-algebraic sequential correlations with $\mathsf{k}$, $\mathcal{A}_{seq}^\mathsf{k}$]\label{def:k-alg-seq}
A $\mathsf{k}$-partite algebraic sequential strategy is specified by 
\begin{enumerate}
    \item $\mathsf{k}-1$ C$^*$-algebras $\{\mathscr{A}_i\}_{i=1}^{\mk -1}$;
    \item measurements $\mathfrak{a}_{\mathsf{k}-1}|x_{\mathsf{k}-1} \in \mathscr{A}_{\mk-1}^+$;
    \item $\mk-2$ CP sub-unital maps $T_{a_j|x_j} : \mathscr{A}_j \to \mathscr{A}_{j-1}$, such that $\sum_{a_j} T_{a_j|x_j} (\mathfrak{g}) = T_{x_j}(\mathfrak{g})$ is unital;
    \item sub-normalised states $\phi_{a_1|x_1} : \mathscr{A}_1 \to \mathbb{C}$, such that $\sum_{a_1} \phi_{a_1|x_1} = \phi_{x_1}$ is a normalised state.
\end{enumerate}
All $\mk$-partite C$^*$-algebraic sequential correlations can be written in the following way : 
    \begin{equation*}\label{eq:k-seq-corr}
        p(a_1\dots a_\mathsf{k}| x_1\dots x_\mathsf{k}) = \phi_{a_1|x_1} \circ T_{a_2|x_2} \circ \dots \circ T_{a_{\mathsf{k}-1}|x_{\mathsf{k}-1}}(\mathfrak{a}_{a_\mathsf{k}|x_\mathsf{k}}).
    \end{equation*}
\end{definition}

The chain rule can be naturally extended to any finite number of players $\mathsf{k} \in \mathbb{N}$; consequently, the previous corollary also naturally generalises to $\mathsf{k}$-partite correlations.

\begin{corollary}\label{cor:k-seq-players}
    For any $\mathsf{k}\in\mathbb{N}$, the correlations generated by sequential games with $\mathsf{k}$ players where all communication observes operational-no-signalling are exactly those correlations generated by $\mathsf{k}$ players with commuting operator strategies.
\end{corollary}
\begin{proof}
   Consider a general $\mathsf{k}$-player sequential correlation as in Definition \ref{def:k-alg-seq}, where $\{ \phi_{a_1|x_1}\}_{a_1,x_1}$ is an operationally-non-signalling assemblage, with $\sum_{a_1} \phi_{a_1|x_1} = \phi$,  and $\{T_{a_j|x_j}\}_{a_j,x_j}, j \in [2,\mathsf{k}-1]$, are operationally-non-signalling transformations, with $\sum_{a_j} T_{a_j|x_j} = T^{(j)}$.
    We then apply the chain rule for sequential Radon-Nikodym derivatives (Theorem \ref{th:chained-rn-k}), yielding 
    \begin{align*}
        p(a_1\dots a_\mathsf{k}| x_1\dots x_\mathsf{k})= \bra{\Omega} F_{a_1|x_1} \dots F_{a_{\mathsf{k}-1}|x_{\mathsf{k}-1}} \pi(\mathfrak{a}_{a_\mathsf{k}|x_\mathsf{k}}) \ket{\Omega},
    \end{align*}
    where all the operators inside the bra-ket sandwich pairwise commute.
   The same considerations presented in Lemma~\ref{lemma:chain-2-rn-sums} can be used to show that these operators are POVMs, hence this is a valid $\mathsf{k}$-partite commuting operator strategy.
\end{proof}

\subsection{Connection with compiled multipartite non-local games}\label{subsec:connection-compiled}

In this section we show that the constraints on the asymptotic algebraic strategies of the compiled game (Theorem \ref{th:asymptotic-constraints-k}) are the necessary and sufficient conditions to apply the generalisation of the chain rule for Radon-Nikodym (Theorem \ref{th:gen-chained-rn-k}); this implies that the asymptotic quantum value of every compiled game is upper bounded by the quantum commuting operator value~(Def.~\ref{def:co-k}). Let us first state the result for tripartite games,
which can be seen as a generalisation of Theorem 6.1 of \cite{KMPSW24bound} for tripartite games.
\begin{theorem}\label{th:as-soundness-3}
    Let $\mathcal{G}$ be any three-player non-local game and let $\mathcal{Q}_\text{comp}^3$ be the set of all efficient quantum strategies for the compiled game $\mathcal{G}_\text{comp}= \{ \mathcal{G}_\lambda \}_\lambda$. Then it holds that
    \begin{equation*}
        \omega_q(\cG_{comp}) = \sup_{\{p_\lambda\}_\lambda \in \mathcal{Q}_\text{comp}^3} \limsup_{\lambda \to \infty} \omega(\mathcal{G}_\lambda, p_\lambda) \leq \sup_{p \in \mathcal{C}_{qc}} \omega(\mathcal{G},p) = \omega_{qc}(\mathcal{G})
    \end{equation*}
\end{theorem}
\begin{proof}
    Let $\{p_\lambda\}_\lambda \in \mathcal{Q}_\text{comp}^3$ be any efficient quantum strategy in the compiled game $\mathcal{G}_\text{comp}= \{ \mathcal{G}_\lambda \}_\lambda$, and let $(\lambda_k)_k$ be an arbitrary subsequence.
    Then, by the Banach-Alaoglu Theorem, and with the same argumentation as in Section~\ref{subsec:asy-eff-strat}, by passing to a further subsequence $(\lambda_k)_k$ and slightly abusing notation, it holds that
    \begin{equation*}
        \lim_{\lambda_k \to \infty} p_\lambda(a,b,c|x,y,z) = 
        \phi_{a|x}\circ T_{b|y}(\mathfrak{c}_{c|z}),
    \end{equation*}
    for all $a,b,c,x,y,z$.
    According to Theorem~\ref{th:asymptotic-constraints-3}, this satisfies the following properties:
    \begin{align}
    & 
        \sum_a \phi_{a|x}(\mathfrak{a}) = \sum_a\phi_{a|x'}(\mathfrak{a}) & \forall \mathfrak{a} \in \mathscr{A}_{B\to C}, \label{eq:prop1}\\
    &
        \sum_b \phi_{a|x} \left( \mathfrak{l}^* T_{b|y}(\mathfrak{b}) \mathfrak{r} \right) = 
        \sum_b \phi_{a|x} \left( \mathfrak{l}^* T_{b|y'}(\mathfrak{b}) \mathfrak{r} \right)  & \forall \mathfrak{b} \in \mathscr{A}_{C}, \forall \mathfrak{l,r} \in \mathscr{A}_{B \to C}.\label{eq:prop2}
    \end{align}
    Equation \ref{eq:prop1} is the definition of operationally-non-signalling states.
    Unfortunately, Eq. \ref{eq:prop2} does not correspond to operational-non-signalling transformations, and thus Theorem \ref{th:chain-2-RN} cannot be directly applied. Nevertheless, the constraint gives enough structure to apply the generalisation formulated in Theorem \ref{th:gen-chainRN-2}.

    Consider $(\mathcal{H}, \pi, \ket{\Omega} )$ the cyclic GNS representation of the state $\phi_{a|x}$. Equation \ref{eq:prop2} can be reformulated as
    \begin{align*}
        0 =\phi_{a|x} \left( \mathfrak{l}^* (T_{y} - T_{y'})(\mathfrak{b}) \mathfrak{r} \right) 
        &= \bra{\Omega} \pi^*(\mathfrak{l}) \pi\circ (T_{y} - T_{y'})(\mathfrak{b}) \pi(\mathfrak{r}) \ket{\Omega}\\
        &= \bra{h'} \pi\circ (T_{y} - T_{y'})(\mathfrak{b}) \ket{h} \qquad \forall \ket{h},\ket{h'} \in \mathcal{H}
    \end{align*}
    where in the last line we used the cyclicity property of the representation, i.e. $\overline{\pi(\mathscr{A})} = \mathcal{H}$.
    This implies that the operators $\pi \circ T_y(\mathfrak{b}) \in \mathsf{B}(\mathcal{H})$ do not depend on the label $y$ :
    \begin{equation*}
    \pi\circ T_{y}(\mathfrak{b}) = \pi\circ T_{y'}(\mathfrak{b}) \qquad \forall y,y' \in \mathcal{I}_2.
    \end{equation*}

    We apply the general version of the chain rule for Radon-Nikodym derivatives, as stated in Theorem \ref{thm:rn-for-cp-maps}, to the following completely positive maps:
    \begin{align*}
        & S_{x_1}^{(1)} = \sum_a \phi_{a|x} = \phi_x &&R^{(1)} = \phi_x = \phi_{x'} = \phi, \\
        &  S_{x_2}^{(2)} = \sum_b T_{b|y} = T_y   &&R^{(2)} =\pi_\phi \circ T_y = \pi_\phi \circ T_{y'}.
    \end{align*}
    This proves that all asymptotic algebraic compiled correlations can be written as
    \begin{align*}
        \phi_{a|x} \circ T_{b|y} (\mathfrak{c}_{c|z})
        &= \bra{\Omega}
        F_{a|x} F_{b|y} \pi(\mathfrak{c}_{c|z})
        \ket{\Omega}
    \end{align*}
    where all the operators are pairwise commuting.
    Using the same arguments of Cor. \ref{cor:3seqplayers}, it is not difficult to show that they are POVMs; hence all asymptotic correlations can be written as a $3$-partite commuting operator strategy.
\end{proof}

As in the previous sections, the generalization to more than $3$ players follows straightforwardly. We state it here as a theorem.

\begin{theorem}[Asymptotic quantum soundness of the \cite{KLVY22Quantum} compiler]%
\label{th:asymptotic-quantum-soundness-k-players}
    Let $\mathcal{G}$ be any $\mathsf{k}$-players non-local game and let $\mathcal{Q}_\text{comp}^\mathsf{k}$ be the set of all efficient quantum strategies for the compiled game $\mathcal{G}_\text{comp}= \{ \mathcal{G}_\lambda \}_\lambda$. Then it holds that
    \begin{equation*}
        \omega_q(\cG_{comp}) = \sup_{\{p_\lambda\}_\lambda \in \mathcal{Q}_\text{comp}^\mathsf{k}} \limsup_{\lambda \to \infty} \omega(\mathcal{G}^\lambda, p_\lambda) \leq \sup_{p \in \mathcal{C}_{qc}} \omega(\mathcal{G},p) = \omega_{qc}(\mathcal{G})
    \end{equation*}
\end{theorem}
\begin{proof}
The theorem before is the proof of the special case for $\mathsf{k}=3$.
To prove it for all $\mathsf{k}\in \mathbb{N}^+$, it is sufficient to combine the proof structure presented above together with the following $\mathsf{k}$-partite extensions: the set of constraints on the asymptotic algebraic compiled quantum strategies presented in Theorem \ref{th:asymptotic-constraints-k}, and the generalised chained Radon-Nikodym Theorem formulated in Theorem \ref{th:gen-chained-rn-k}.    
\end{proof} 
\newpage
\appendix
\section{Purifying provers in compiled games}\label{app:homomorphism}

Consider the maps labelled by the encrypted labels $\Tilde{A}_{\enca|\encx}^{\lambda}$. These are the efficient quantum instruments that the prover implements during the protocol; we can use Lemma~\ref{lem:homomorphism} to show that the adjoint maps $\Tilde{A}_{\enca|\encx}^{\lambda,*}$ are homomorphisms.

We are interested in the coarse-grained map with plaintext variables
\begin{equation*}
    \Tilde{A}_{a|x}^{\lambda} = \mathop{\mathds{E}}_{\encx:\dec(\encx) = x} \sum_{\enca:\dec(\enca)=a} \Tilde{A}_{\enca|\encx}^{\lambda}.
\end{equation*}
This does not directly correspond to the physical operations implemented by the prover in the compiled protocol, but turns out to be useful in describing the correlations generated by the provers on the plaintext variables.
We want to show that this map $\Tilde{A}_{a|x}^{\lambda}$ can be efficiently dilated to a purified map whose adjoint is consequently a homomorphism.

To motivate this, we give an operational meaning of the map $\Tilde{A}_{a|x}^{\lambda}$.
We claim that all of the operations necessary to implement this map can be purified, especially we need to purify the encryption and the decryption operations.
Then, we will use Lemma~\ref{lem:homomorphism} to show that the adjoint of the map is a homomorphism.

To be more concrete, the encryption and the decryption gates are not pure because they make internal use of classical randomness.
To purify these, we can think of the de-randomised version of these gates.
For example, consider the encryption; instead of sampling a random coin every time to select a specific ciphertext, encryption is performed by a deterministic classical circuit which accepts the random coins as an additional input.
This, in turn, gives rise to a quantum circuit realizing the encryption procedure as a unitary operation, as represented in Fig. \ref{fig:oracles}.
Thus defined encryption oracle $O_\enc$ will receive as inputs the plaintext $x$, the random coins, and an ancilla which serves at the same time as the register which holds the outcome in the following way
\begin{align}
    O_\enc : \ket{x} \ket{r} \ket{y} \mapsto \ket{x} \ket{r} \ket{y \oplus \enc^\text{sk}(x,r)}.
\end{align}
On input of a superposition of random coins $r$ and a single plaintext message $x$, this oracle yields a superposition of encryptions of this plaintext.
Analogously, the purified decryption oracle $O_\dec$ accepts superpositions of ciphertexts and coherently decrypts them as
\begin{align}
    O_\dec : \ket{c} \ket{y} \mapsto \ket{c} \ket{y \oplus \dec^\text{sk}(c)}.
\end{align}
Furthermore, let us consider the instrument with encrypted labels $\Tilde{A}^\lambda_{\enca_i|\encx_i}$. We already showed in the main text that this can always be modelled as a unitary $U_{\encx_i}$, followed by a partial measurement $\Pi_{\enca_i}$ on some of the output registers, and a controlled unitary on the remaining registers.
By the deferred measurement principle, one can then postpone any such measurement to a later point and apply any operations -- that were previously classically controlled on the measurement outcome -- coherently controlled instead, without changing the effect of the operations.

\begin{figure}[ht]
    \centering
    \begin{minipage}{0.45\textwidth}
        \centering
\begin{quantikz}[wire types={c,q,q}]
\lstick{$x$}  &\gate[3]{O_\enc}   & \rstick{$x$} \\
\lstick{$\ket{r}$}    &                   & \rstick{$\ket{r}$}\\
\lstick{$\ket{y}$}    &                   & \rstick{$\ket{y \oplus \enc^\mathsf{sk}(x,r)}$}
\end{quantikz}
     \end{minipage}
    \hfill
    \begin{minipage}{0.45\textwidth}
        \centering
\begin{quantikz}[wire types={q,q}]
\lstick{$\ket{\mathtt{c}}$}& \gate[2]{O_\dec} & \rstick{$\ket{\mathtt{c}}$} \\
\lstick{$\ket{y}$} & &\rstick{$\ket{y \oplus \dec^\mathsf{sk}(\mathtt{c})}$}
\end{quantikz}     \end{minipage} 
    \caption{The purified encryption and decryption oracles.}
    \label{fig:oracles}
\end{figure}
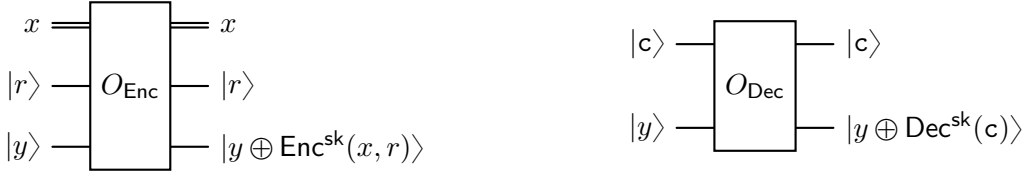

\begin{figure}[ht]
    \centering
\begin{quantikz}[wire types={c,n,n}]
\lstick{$x$} &\gate{\enc}    & \gategroup[3,steps=7,style={draw=none,rounded
corners,fill=blue!10, inner
xsep=2pt},background,label style={label
position=below,anchor=north,yshift=-0.2cm}]{{
Prover}} &                     &                &\ctrl{1}          &
&                  &                        & \gate{\dec} & \rstick{$x$}
\\
             &               & &\lstick{$\rho$}      &\setwiretype{q} &\gate[2]{U_\encx} &
&                  &\gate{V_{\enca,\encx}}  &             &
\\
             &               & &\lstick{$\ket{0}$}   &\setwiretype{q} &                  &\meter{\enca}
&\setwiretype{c}   & \ctrl{-1}              & \gate{\dec} & \rstick{$a$}
\end{quantikz}     \caption{The circuit representation of a prover in the compiled non-local game. Since the prover's measurement collapses the quantum register holding its encrypted output (rather than its plaintext output), we will call this the \emph{fine-grained} representation of the prover's operations.}
    \label{fig:fine-grained-prover}
\end{figure}
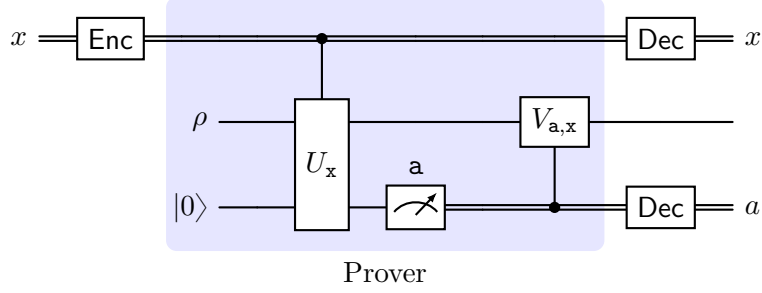

With these tools, we can give an operational meaning to the coarse-grained map $A_{a|x}^\lambda$, represented with the circuit represented in Fig.~\ref{fig:coarse-graining-circuit}.
Coarse-graining the prover's operations consists of prepending the encryption oracle and appending the decryption oracle which captures the effect of the prover's actions on plaintext input $x$, yielding plaintext output $a$.
Instead of performing the prover's measurement before the classical decryption, deferring this measurement allows us to measure the plaintext response $a$ only, without disturbing the internal quantum state before the, now coherent, decryption routine, which at the end of the circuit holds a superposition of ciphertexts corresponding to the plaintext measurement outcome $a$.
Note that this step, together with the fact that all intermediate quantum registers are kept until the end of the circuit without being measured or traced out makes the circuit in Figure~\ref{fig:coarse-graining-circuit} the description of a purified quantum instrument, consisting solely of the composition of a unitary operation and a projective measurement of part of the output wires.
Lemma~\ref{lem:homomorphism} implies that the adjoint of the described map is indeed a *-homomorphism.

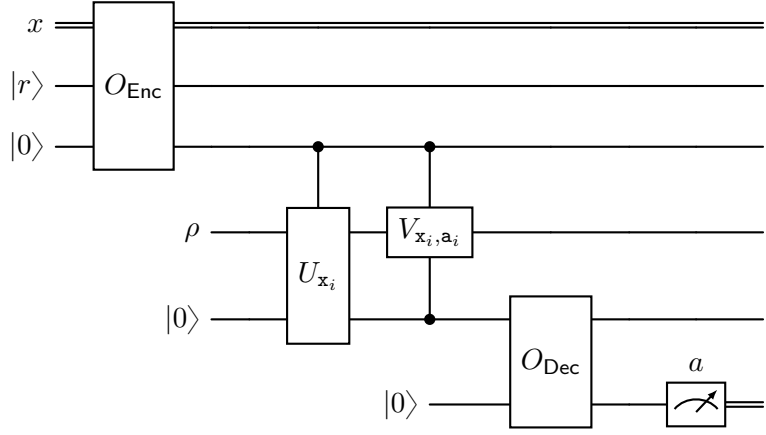
\begin{figure}[ht]
    \centering
\begin{quantikz}[wire types={c,q,q,n,n,n}]
\lstick{$x$}      &\gate[3]{O_\enc}   &                   &                   &                       &                           &                   & &&
\\
\lstick{$\ket{r}$}      &                   &                   &                   &                       &                           &                   & &&
\\
\lstick{$\ket{0}$}    &                   &                   &               &\ctrl{1}               &\ctrl{1}                           &                   & &&
\\
                &                   &\lstick{$\rho$}    & \setwiretype{q}   &\gate[2]{U_{\encx_i}}  &\gate{V_{\encx_i,\enca_i}} &                   & &&
\\
                &                   &\lstick{$\ket{0}$} & \setwiretype{q}   &                       &\ctrl{-1}                  &\gate[2]{O_\dec}   & &&
\\
                &                   &                   &                   &                       &\lstick{$\ket{0}$}         &\setwiretype{q}   &\setwiretype{q} & \meter{a} &\setwiretype{c} 
\end{quantikz}     \caption{The purified circuit representation of the coarse-grained map $A_{a|x}^\lambda$.}
    \label{fig:coarse-graining-circuit}
\end{figure}

We remark that, by construction, all operations in the circuit in Figure~\ref{fig:coarse-graining-circuit} are efficient.
Furthermore, the correlations between $x$ and $a$ generated by the coarse-grained circuit in Figure~\ref{fig:coarse-graining-circuit} are identical to those generated by the original, fine-grained circuit in Figure~\ref{fig:fine-grained-prover}.

The Hilbert space in which the output state of the purified and coarse-grained circuit from Figure~\ref{fig:coarse-graining-circuit} is living has been considerably enlarged when compared to the original Hilbert space pre-coarse-graining.
This is a necessary step in the purification, \textit{i.e.}, dilation of the prover's map.
All following quantum instruments and measurements, those of the subsequent provers, are straightforwardly lifted to work on quantum states of the enlarged space, by simply ignoring and wires that the provers did not have access to previously.

We already argued that this purification does not change any of the correlations that are generated by the involved operations.
However, as we work more generally with the operator algebras generated by the considered quantum instruments and measurements in this paper, it remains to show that this mathematical purification trick does not impact any of the statements that we make involving the quantum information accessible in these operator algebras.
To this end, it is sufficient to observe that, even though these operator algebras formally live in a larger Hilbert space post-purification, they do not act non-trivially outside of the subalgebra that describes the quantum information accessible pre-purification.
This is due to the fact that quantum instruments and measurements are dilated to larger Hilbert space trivially, by ignoring any additional registers available to them.
Formally, this corresponds to forming the tensor product of the operators describing the instruments or measurements with the identity operator on any additional registers.
This extension acts *-homomorphically: given any family of finite-dimensional operators $\{A_i\}$ and a non-commutative *-polynomial $P$, it holds that
\begin{align}
    P(\{A_i \otimes \mathds{1}\}) = P(\{A_i\}) \otimes \mathds{1}.
\end{align}
This reduces any algebraic statement about the operator algebra generated by the family $\{A_i \otimes \mathds{1}\}$ in the larger, purified Hilbert space to the algebra generated by $\{A_i\}$ in the smaller, original Hilbert space.

\newpage

\section{An alternative proof using lifted CP maps and POVMs}%
\label{app:alternative-proof}

This section means to provide an alternative proof of the main result of this paper, the asymptotic quantum soundness of the KLVY-compiler for all multipartite games, where the upper bound is given by the commuting-operator quantum value of the original nonlocal game.

The alternative proof differs from the proof in the main manuscript by avoiding the need to purify the prover's operations.
Instead of dealing with PVMs, it uses POVMs, and instead of $*$-homomorphisms describing purified quantum channels, it uses CP-maps which are the natural maps taking POVMs to POVMs.
Because of this, we need to rely on a different construction of universal C*-algebras, namely those that capture POVMs and sequential POVMs. This modification works entirely analogously to the constructions of universal C*-algebras of PVMs and of sequential PVMs, presented previously, and will therefore not be discussed in extensive detail.
The price to pay for avoiding the purification of the prover is a more involved construction of the CP-maps between the universal C*-algebras. As explained in the following, this requires the notion of nuclearity of C*-algebras, as well as a lifting theorem for CP-maps.

Consider two parties Bob and Charlie. In the context of PVMs, the natural map between the universal C*-algebra of Charlie's PVMs and the universal C*-algebra of Bob's and Charlie's joint sequential PVMs, that takes generators to generators, is a *-homomorphism which is easily constructed through the universal property of the C*-algebras.
For POVM algebras, we cannot expect that this map exists and retains homomorphic structure, as complete positivity and unitality is sufficient to preserve the structure of POVMs.
The main technical obstacle is therefore the construction of a CP-map between the universal C*-algebras, with the property that the analogous diagram to Figure~\ref{fig:univAlgebras-3players} commutes.

For our later arguments we shall need one structural property of C*-algebras, nuclearity, but only its consequences, not the full technical machinery.  Roughly speaking, a C*-algebra $\mathscr{A}$ is nuclear if, whenever we form the algebraic tensor product $\mathscr{A}\odot \mathscr{B}$ with any other C*-algebra $\mathscr{B}$, there is only one C*-norm on $\mathscr{A}\odot \mathscr{B}$. Intuitively, nuclearity guarantees that $\mathscr{A}$ behaves as though it were “finite‑dimensional at every scale”, preventing pathological tensor‑norm phenomena and making dilation arguments more transparent. We will exploit this property later without delving into its proof theory.  What matters here is the following guiding lemma.

\begin{lemma}\label{lem:nuclear}
    Every finite‑dimensional C*-algebra is nuclear.
\end{lemma}

Hence all algebras that arise in finite‑level quantum information (e.g.\ $\mathcal{B}(\mathbb{C}^d)$ and their finite direct sums) are nuclear.

To conclude this background subsection we recall the Choi–Effros lifting theorem.

\begin{theorem}[Choi–Effros lifting Theorem, Corollary 3.11 in~\cite{ChoiEffros}] \label{thm:choieffros}
      Let $\mathscr{A}$ be a separable $C^*$-algebra, let $\mathscr{B}$ be a unital $C^*$-algebra, and let $J\subseteq \mathscr{B}$ be a closed two‑sided ideal with quotient map $\pi:\mathscr{B}\twoheadrightarrow \mathscr{B}/J$.  Suppose $\Phi:\mathscr{A}\to \mathscr{B}/J$ is a completely positive map.  If any one of the three $C^*$-algebras $\mathscr{A}$, $\mathscr{B}$, or $\mathscr{B}/J$ is nuclear, then there exists a completely positive map $\widetilde{\Phi}:\mathscr{A}\to \mathscr{B}$ (called a lifting of $\Phi$) such that $\pi\circ\widetilde{\Phi}=\Phi$.
\end{theorem}

In words, nuclearity of either the domain, codomain, or ambient algebra guarantees that every completely positive map into the quotient can be realised already at the level of the larger algebra.

\begin{definition}[Algebraic three-round strategy, $\mathcal{A}_\lambda^3$]\label{def:3rAseq}
An algebraic three-round strategy consists of the following objects:
    \begin{itemize}
        \item a universal $C^*$-algebra $\mathscr{A}_C$, generated by elements $\{\mathfrak{c}_{c|z}\}_{c \in \mathcal{O}_C}$, where each $\mathfrak{c}_{c|z} >0$ and for each question $z \in \mathcal{I}_C$ it holds $\sum_c \mathfrak{c}_{c|z} = \mathfrak{1}$. This algebra abstractly encodes Charlie’s measurement operations.
        \item a universal $C^*$-algebra $\mathscr{A}_{B\to C}$, generated by elements $\{\mathfrak{f}_{b,c|y,z}\}_{b,c \in \mathcal{O}_B,\mathcal{O}_C}$, where each $\mathfrak{f}_{b,c|y,z} >0$, for each question $z \in \mathcal{I}_C$ it holds $\sum_c \mathfrak{f}_{b,c|y,z} = \mathfrak{f}_{b|y}$ and for each question $y \in \mathcal{I}_B$ it holds $\sum_b \mathfrak{f}_{b|y} = \mathfrak{1}$. This algebra abstractly encodes jointly Bob's instruments and Charlie's POVMS.
        \item a  sequence of CP maps, algebraically encoding Bob’s instruments, $\{T^\la_{b|y}: \mathscr{A}_C \to \mathscr{A}_{B\to C}\}_{b \in \mathcal{O}_B}$, such that $\sum_bT_{b|y}^\la = T_y^\la$ are unital for every $y \in \mathcal{I}_B$;
        \item For each input $x\in \mathcal{I}_A$ and outcome $a\in\mathcal{O}_A$, a positive linear functional $\phi^\lambda_{a|x}$ on $\mathscr{A}_{B\to C}$, representing Alice’s prepared state conditioned on input $x$ and outcome $a$. These functionals sum to a normalized state $\phi^\lambda_{x} = \sum_a \phi^\lambda_{a|x}$.
    \end{itemize}
    The correlation given by this strategy is $$p^\la(a,b,c|x,y,z)= \phi_{a|x}^\la(T_{b|y}^\la(\mathfrak{c}_{c|z}))$$
The strategy is $\la$-no-signalling if
    \begin{enumerate}
        \item states $\phi^\la_x$ satisfy
    \begin{equation}\label{eq:algnsl}
        \left|\phi^\la_{x}(\mathfrak{b}) - \phi^\la_{x'}(\mathfrak{b}) \right| \leq \mathsf{negl}(\la),
    \end{equation}
    for every $x,x'\in \mathcal{I}_A$, and for every $\mathfrak{b} \in \mathscr{A}_{B\to C}$.
    \item CPTP maps $T_y^\la$ satisfy
    \begin{equation}\label{eq:algnslIns}
        \left|\phi_{a|x}^\la(\mathfrak{l}^*T_y^\la(\mathfrak{b})\mathfrak{r}) - \phi_{a|x}^\la(\mathfrak{l}^*T_{y'}^\la(\mathfrak{b})\mathfrak{r}) \right| \leq \mathsf{negl}(\la), 
    \end{equation}
    for every $a \in \mathcal{O}_A$, $x\in\mathcal{I}_A$, and for every $\mathfrak{b} \in \mathscr{A}_C$, $\mathfrak{l},\mathfrak{r} \in \mathscr{A}_{B\to C}$.
    \end{enumerate}
\end{definition}

The universal $C^*$-algebras $\mathscr{A}_B$ and $\mathscr{A}_{B\to C}$ in this definition are exactly those defined in Lemmas~\ref{lem:universal_povm_algebra} and~\ref{lem:universal_sequential_povm_algebra}. Condition~\eqref{eq:algnsl} enforces that after the first round, the reduced state on Bob (and Charlie) is essentially independent of Alice’s input $x$. Condition~\eqref{eq:algnslIns} ensures that no information about Bob’s input $y$ leaks forward: even if Bob applied any extra (bounded) map before his measurement, Charlie’s outcomes remain (up to negligible function od $\la$) independent of $y$.

Now we proceed with connecting three-round algebraic sequential strategies represent and three-round quantum sequential strategies. Our main task is to establish the correspondence between Bob’s instruments  ${B_{b|y}^\la,*}$ and and the CP maps $T_{b|y}^\la$. But let us start from the tools already defined in the two-round case. Namely, as in two-round case, the third-round POVMs $C_{c|z}^\la$ satisfy the relations defining the universal POVM algebra, and hence by the universal property there exists a representation $\theta_C^\la: \mathscr{A}_C\to \mathsf{B}(\mathcal{H}^\lambda)$ mapping generators $\mathfrak{e}_{c|z}$ to POVM elements $C_{c|z}^\la$. Consider now  POVM set $B^{\lambda,*}_{b|y} (C_{c|z}^\lambda)= M_{bc|yz}^\lambda \in \mathsf{B}(\mathcal{H}^\lambda)$. It respects the defining relations of universal $C^*$-algebra of sequential POVMs from Lemma~\ref{lem:universal_sequential_povm_algebra} as they are all positive, when summed over $c$ do not depend on $z$, and sum to identity. Hence universality property (Lemma~\ref{lemma:universality_c_star_algebra}) ensures the existence of representation $\theta_{BC}^\la:\mathscr{A}_{B\to C}\to  \mathsf{B}(\mathcal{H}^\lambda)$ such that
\begin{equation}\label{eq:universal-property-seqPOVM}
    \theta_{BC}^\lambda (\mathfrak{f}_{bc|yz}) = M_{bc|yz}^\lambda = B^{\lambda,*}_{b|y} (C_{c|z}^\lambda).
\end{equation}

We now analyze explicitly how algebraic completely positive (CP) maps arise naturally in this scenario. Consider the CP map defined by $\tau_{b|y}^\la = B_{b|y}^\lambda\circ \theta_C^\lambda: \mathscr{A}_{C}\to \mathsf{B}(\mathcal{H}^\lambda)$, which sends elements from Charlie’s universal algebra $\mathscr{A}_C$ to bounded operators on the Hilbert space $\mathcal{H}^\la$. By construction, the range of this CP map is precisely contained within the image of the representation $\theta_{BC}^\la(\mathscr{A}_{B\to C}) \in \mathsf{B}(\mathcal{H}^\la)$, as ensured by the universal property (Lemma~\ref{lemma:universality_c_star_algebra}). Therefore, we can view $\tau_{b|y}^\la$ as a CP map:
$$\tau_{b|y}^\la:\mathscr{A}_{C} \to \theta_{BC}^\lambda(\mathscr{A}_{B\to C}) \cong \mathscr{A}_{B\to C}/\ker \theta_{BC}^\lambda,$$
where $\mathscr{A}_{B\to C}/\ker \theta_{BC}^\lambda$ is the quotient algebra corresponding to the two-sided closed ideal $\ker \theta_{BC}^\lambda$.
Now, since $\theta_{BC}^\lambda(\mathscr{A}_{B\to C})$ is finite-dimensional (spanned by matrices $B_{b|y}^{\la,*
}(C_{c|z}^\la)$), it is nuclear by Lemma~\ref{lem:nuclear}. Thus, by the Choi-Effros lifting Theorem (Theorem~\ref{thm:choieffros}), the CP map $\tau_{b|y}^\la$  can always be lifted to a CP map $T_{b|y}^\la:\mathscr{A}_C\to \mathscr{A}_{B\to C}$ satisfying $\theta_{BC}\circ T_{b|y}^
\la = \tau_{b|y}^\la$, and hence
$$\theta^\la_{BC}(T^\la_{b|y}(\mathfrak{e}_{c|z})) = B_{b|y}^{\la,*}(C_{c|z}^\la).$$

Finally, preparations $\rho_{a|x}^\la$ induce the sequence of positive linear functionals $\phi_{a|x}^\la:\mathscr{A}_{C\to B}:\to \mathbb{C}$ $$\phi^\la_{a|x}(\mathfrak{a}) = \Tr[\sigma_{a|x}^\la\theta_{BC}^\la(\mathfrak{a})],$$
and in particular
\begin{align*}\phi^\la_{a|x}(T_{b|y}^\la(\mathfrak{c}_{c|z})) &= \Tr[\sigma_{a|x}^\la\theta_{BC}^\la(T_{b|y}^\la(\mathfrak{c}_{c|z}))]\\
&= \Tr[\sigma_{a|x}^\la B_{b|y}^{\la,*}(C_{c|z}^\la)], \end{align*}
showing that every three-round quantum sequential strategy has its algebraic equivalent.

This shows that every three-round quantum sequential $\la$-no-signalling strategy leads to correlations that can always be achieved by algebraic sequential $\la$-no-signalling strategies.

\newpage
\section*{Acknowledgments}
The authors would like to thank Flavio Baccari, Eleni Diamanti, Alex Grilo, Damian Markham, Vern Paulsen, Marco Túlio Quintino, Marc-Olivier Renou, Lucas Tendick, Quoc-Huy Vu, and Xiangling Xu for useful discussions.
The authors also thank the anonymous referees for their comments, which helped improve the presentation of this work.
MB acknowledges funding from QuantEdu France, a state aid managed by the French National Research Agency for France 2030 with the reference ANR-22-CMAS-0001.
DL acknowledges support from the Quantum Advantage Pathfinder (QAP) research program within the UK’s National Quantum Computing Center (NQCC).
I\v{S} acknowledges funding from the PEPR integrated project EPiQ ANR-22-PETQ-0007 part of Plan France 2030.
\bibliographystyle{alpha}
\bibliography{references}

\newcommand{\etalchar}[1]{$^{#1}$}
\begin{thebibliography}{SHMSS{\etalchar{+}}22}

\bibitem[BCP{\etalchar{+}}14]{Brunner_2014}
Nicolas Brunner, Daniel Cavalcanti, Stefano Pironio, Valerio Scarani, and Stephanie Wehner.
\newblock Bell nonlocality.
\newblock {\em Reviews of Modern Physics}, 86(2):419–478, April 2014.

\bibitem[BGKM{\etalchar{+}}23]{brakerski2023simple}
Zvika Brakerski, Alexandru Gheorghiu, Gregory~D Kahanamoku-Meyer, Eitan Porat, and Thomas Vidick.
\newblock Simple tests of quantumness also certify qubits.
\newblock In {\em Annual International Cryptology Conference}, pages 162--191. Springer, 2023.

\bibitem[BKM{\etalchar{+}}24]{Bacho_compiled_trapdoor}
Kaniuar Bacho, Alexander Kulpe, Giulio Malavolta, Simon Schmidt, and Michael Walter.
\newblock Compiled nonlocal games from any trapdoor claw-free function.
\newblock Cryptology {ePrint} Archive, Paper 2024/1829, 2024.

\bibitem[Bla06]{blackadar2006operator}
Bruce Blackadar.
\newblock {\em Operator algebras: theory of C*-algebras and von Neumann algebras}, volume 122.
\newblock Springer Science \& Business Media, 2006.

\bibitem[Bra18]{brakerski2018quantum}
Zvika Brakerski.
\newblock Quantum {FHE} (almost) as secure as classical.
\newblock In {\em Annual International Cryptology Conference}, pages 67--95. Springer, 2018.

\bibitem[BS86]{BS86radon}
V.P. Belavkin and P.~Staszewski.
\newblock A {Radon-Nikodym} theorem for completely positive maps.
\newblock {\em Reports on Mathematical Physics}, 24(1):49--55, 1986.

\bibitem[BVB{\etalchar{+}}24]{baroni2024quantumboundscompiledxor}
Matilde Baroni, Quoc-Huy Vu, Boris Bourdoncle, Eleni Diamanti, Damian Markham, and Ivan Šupić.
\newblock Quantum bounds for compiled {XOR} games and $d$-outcome chsh games, 2024.

\bibitem[CE76]{ChoiEffros}
Man-Duen Choi and Edward~G. Effros.
\newblock The completely positive lifting problem for c*-algebras.
\newblock {\em Annals of Mathematics}, 104(3):585--609, 1976.

\bibitem[CMM{\etalchar{+}}24]{CMMN2024Computational}
David Cui, Giulio Malavolta, Arthur Mehta, Anand Natarajan, Connor Paddock, Simon Schmidt, Michael Walter, and Tina Zhang.
\newblock A computational {T}sirelson's theorem for the value of compiled {XOR} games, 2024.

\bibitem[Con00]{conway}
J.B. Conway.
\newblock {\em A Course in Operator Theory}.
\newblock American Mathematical Soc., 2000.

\bibitem[Ege54]{Egervary54}
E.~Egerv\'ary.
\newblock On the contractive linear transformations of {$n$}-dimensional vector space.
\newblock {\em Acta scientiarum mathematicarum}, 1954.

\bibitem[GSLW19]{Gily_n_2019}
András Gilyén, Yuan Su, Guang~Hao Low, and Nathan Wiebe.
\newblock Quantum singular value transformation and beyond: exponential improvements for quantum matrix arithmetics.
\newblock In {\em Proceedings of the 51st Annual ACM SIGACT Symposium on Theory of Computing}, STOC ’19, page 193–204. ACM, June 2019.

\bibitem[JNV{\etalchar{+}}22]{ji2022mipre}
Zhengfeng Ji, Anand Natarajan, Thomas Vidick, John Wright, and Henry Yuen.
\newblock {MIP*=RE}, 2022.

\bibitem[Kir05]{schrodingerhjwtheorem}
K.~A. Kirkpatrick.
\newblock The {Schrodinger-HJW} theorem, 2005.

\bibitem[KLVY23]{KLVY22Quantum}
Yael Kalai, Alex Lombardi, Vinod Vaikuntanathan, and Lisa Yang.
\newblock Quantum {A}dvantage from {A}ny {N}on-local {G}ame.
\newblock In {\em Proceedings of the 55th Annual ACM Symposium on Theory of Computing}, STOC 2023, page 1617–1628, 2023.

\bibitem[KMP{\etalchar{+}}24]{KMPSW24bound}
Alexander Kulpe, Giulio Malavolta, Connor Paddock, Simon Schmidt, and Michael Walter.
\newblock A bound on the quantum value of all compiled nonlocal games, 2024.

\bibitem[KPR{\etalchar{+}}25]{Xiangling}
Igor Klep, Connor Paddock, Marc-Olivier Renou, Simon Schmidt, Lucas Tendick, Xiangling Xu, and Yuming Zhao.
\newblock Quantitative quantum soundness for bipartite compiled bell games via the sequential npa hierarchy.
\newblock {\em in preparation}, 2025.

\bibitem[KRR13]{KRR14}
Yael~Tauman Kalai, Ran Raz, and Ron~D. Rothblum.
\newblock How to delegate computations: The power of no-signaling proofs.
\newblock Cryptology {ePrint} Archive, Paper 2013/862, 2013.

\bibitem[Mah20]{mahadev2020classical}
Urmila Mahadev.
\newblock Classical homomorphic encryption for quantum circuits.
\newblock {\em SIAM Journal on Computing}, 52(6):FOCS18--189, 2020.

\bibitem[MNZ24]{metger2024succinctargumentsqmastandard}
Tony Metger, Anand Natarajan, and Tina Zhang.
\newblock Succinct arguments for {QMA} from standard assumptions via compiled nonlocal games, 2024.

\bibitem[MPW24]{mehta2024selftestingcompiledsettingtiltedchsh}
Arthur Mehta, Connor Paddock, and Lewis Wooltorton.
\newblock Self-testing in the compiled setting via tilted-chsh inequalities, 2024.

\bibitem[MRTC21]{Martyn_2021}
John~M. Martyn, Zane~M. Rossi, Andrew~K. Tan, and Isaac~L. Chuang.
\newblock Grand unification of quantum algorithms.
\newblock {\em PRX Quantum}, 2(4), December 2021.

\bibitem[NZ23]{NZ2023Bounding}
Anand Natarajan and Tina Zhang.
\newblock Bounding the quantum value of compiled nonlocal games: from {CHSH} to {BQP} verification, 2023.

\bibitem[PSZZ23]{Paddock_2023}
Connor Paddock, William Slofstra, Yuming Zhao, and Yangchen Zhou.
\newblock An operator-algebraic formulation of self-testing.
\newblock {\em Annales Henri Poincaré}, 25(10):4283–4319, October 2023.

\bibitem[Rag03]{Rag03radon}
Maxim Raginsky.
\newblock {Radon--Nikodym} derivatives of quantum operations.
\newblock {\em Journal of Mathematical Physics}, 44(11):5003--5020, 2003.

\bibitem[Ral21]{Rall_2021}
Patrick Rall.
\newblock Faster coherent quantum algorithms for phase, energy, and amplitude estimation.
\newblock {\em Quantum}, 5:566, October 2021.

\bibitem[SBC{\etalchar{+}}15]{postquantum-steering}
Ana~Bel\'en Sainz, Nicolas Brunner, Daniel Cavalcanti, Paul Skrzypczyk, and Tam\'as V\'ertesi.
\newblock Postquantum steering.
\newblock {\em Phys. Rev. Lett.}, 115:190403, Nov 2015.

\bibitem[SHMSS{\etalchar{+}}22]{Schlimgen_2022}
Anthony~W. Schlimgen, Kade Head-Marsden, LeeAnn~M. Sager-Smith, Prineha Narang, and David~A. Mazziotti.
\newblock Quantum state preparation and nonunitary evolution with diagonal operators.
\newblock {\em Physical Review A}, 106(2), August 2022.

\bibitem[Spe05]{Spekkens_2005}
R.~W. Spekkens.
\newblock Contextuality for preparations, transformations, and unsharp measurements.
\newblock {\em Physical Review A}, 71(5), May 2005.

\bibitem[Sti55]{Stinespring}
W.~Forrest Stinespring.
\newblock Positive functions on c*-algebras.
\newblock {\em Proceedings of the American Mathematical Society}, 6(2):211--216, 1955.

\bibitem[TMM{\etalchar{+}}25]{taranto2025higherorderquantumoperations}
Philip Taranto, Simon Milz, Mio Murao, Marco~Túlio Quintino, and Kavan Modi.
\newblock Higher-order quantum operations, 2025.

\bibitem[WF23]{Wright_2023}
Victoria~J. Wright and Máté Farkas.
\newblock Invertible map between bell nonlocal and contextuality scenarios.
\newblock {\em Physical Review Letters}, 131(22), November 2023.

\end{thebibliography}

\end{document}